\documentclass[aps,pra,10pt,twocolumn,letterpaper,groupedaddress,amsmath,amssymb,showpacs,english,superscriptaddress]{revtex4-2}

\usepackage{babel}
\usepackage{graphicx} 
\usepackage{dcolumn} 
\usepackage{bm}
\usepackage[table]{xcolor}
\usepackage{hyperref} 
\usepackage{braket}
\usepackage[ruled,vlined]{algorithm2e}
\usepackage{placeins}
\usepackage{physics}
\usepackage{amsmath} 
\usepackage{amsthm}
\usepackage{mathdots}
\usepackage{array} 
\usepackage{mathtools}
\usepackage{soul}
\usepackage[capitalize]{cleveref}
\usepackage[mathlines]{lineno}
\allowdisplaybreaks
\newcolumntype{C}[1]{>{\centering\arraybackslash}p{#1}}
\usepackage{upgreek}
\usepackage{braket}
\usepackage{multirow}
\usepackage{enumitem}
\usepackage{lineno}

\usepackage{tikz}
\usepackage[compat=newest]{yquant}\useyquantlanguage{groups}

\usepackage{booktabs} 
\usepackage{tabularx} 
\definecolor{lightgray}{gray}{0.9}
\definecolor{awesome}{rgb}{1.0, 0.13, 0.32}
\usepackage{comment}
\usepackage[normalem]{ulem}

\usepackage{subcaption}

\usepackage{booktabs} 
\usepackage{tabularx} 
\definecolor{lightgray}{gray}{0.9}
\definecolor{awesome}{rgb}{1.0, 0.13, 0.32}
\usepackage{comment}
\usepackage[normalem]{ulem}

\makeatletter

\usepackage{xcolor}

\newcommand{\calC}{\mathcal{C}}
\newcommand{\calA}{\mathcal{A}}
\newcommand{\calB}{\mathcal{B}}

\newcommand{\calS}{\mathcal{S}}
\newcommand{\calG}{\mathcal{G}}
\newcommand{\calP}{\mathcal{P}}
\newcommand{\calL}{\mathcal{L}}
\newcommand{\lrang}[1]{\left< #1 \right>}
\newcommand{\sig}{\sigma}
\newcommand{\aut}[1]{{\mathrm{Aut}}(#1)}
\newcommand{\mat}[1]{\left( \begin{matrix} #1 \end{matrix} \right)}
\newcommand{\clif}{\mathrm{Clif}_n}

\newcommand{\gf}{\mathbb{F}_2}

\newcommand{\diag}{\mathrm{Diag}}
\renewcommand{\i}{\mathrm{i}}

\newcommand{\diagZ}{\mathrm{DZ}}
\newcommand{\diagX}{\mathrm{DX}}
\newcommand{\cnot}{\mathrm{CX}}

\DeclarePairedDelimiter\ceil{\lceil}{\rceil}
\DeclarePairedDelimiter\floor{\lfloor}{\rfloor}

\theoremstyle{plain}
\newtheorem{thm}{\protect\theoremname}[section]
\theoremstyle{definition}
\newtheorem{defn}[thm]{\protect\definitionname}
\theoremstyle{plain}
\newtheorem{cor}[thm]{\protect\corollaryname}
\theoremstyle{plain}
\newtheorem{lem}[thm]{\protect\lemmaname}
\theoremstyle{plain}
\newtheorem{prop}[thm]{\protect\propositionname}
\theoremstyle{remark}
\newtheorem*{rem*}{\protect\remarkname}
\newtheorem{exmp}[thm]{\protect\examplename}
\theoremstyle{plain}
\newtheorem*{theorem*}{Theorem}

\makeatother

\providecommand{\corollaryname}{Corollary}
\providecommand{\definitionname}{Definition}
\providecommand{\lemmaname}{Lemma}
\providecommand{\propositionname}{Proposition}
\providecommand{\remarkname}{Remark}
\providecommand{\theoremname}{Theorem}
\providecommand{\examplename}{Example}

\DeclareMathAlphabet{\mathcal}{OMS}{cmsy}{m}{n}

\begin{document}

\title{Computing Efficiently in QLDPC Codes}
\author{Alexander J. Malcolm}\affiliation{Photonic Inc.}
\author{Andrew N. Glaudell}\affiliation{Photonic Inc.}
\author{Patricio Fuentes}\affiliation{Photonic Inc.}
\author{Daryus Chandra}\affiliation{Photonic Inc.}
\author{Alexis Schotte}\affiliation{Photonic Inc.}
\author{Colby DeLisle}\affiliation{Photonic Inc.}
\author{Rafael Haenel}\affiliation{Photonic Inc.}
\author{Amir Ebrahimi}\affiliation{Photonic Inc.}
\author{Joschka Roffe}\affiliation{Photonic Inc.}\affiliation{University of Edinburgh, United Kingdom}
\author{Armanda O. Quintavalle}\affiliation{Photonic Inc.}\affiliation{Freie Universit\"at Berlin, Germany}
\author{Stefanie J. Beale}\affiliation{Photonic Inc.}
\author{Nicholas R. Lee-Hone}\affiliation{Photonic Inc.}
\author{Stephanie Simmons}\affiliation{Photonic Inc.}

\date{\today}

\begin{abstract}
    It is the prevailing belief that quantum error correcting techniques will be required to build a utility-scale quantum computer able to perform computations that are out of reach of classical computers.
    The quantum error correcting codes that have been most extensively studied and therefore highly optimized, surface codes, are extremely resource intensive in terms of the number of physical qubits needed.
    A promising alternative, quantum low-density parity check (QLDPC) codes, has been proposed more recently.
    These codes are much less resource intensive, requiring significantly fewer physical qubits per logical qubit than practical surface code implementations.
    A successful application of QLDPC codes would therefore drastically reduce the timeline to reaching quantum computers that can run algorithms with exponential speedups like Shor's algorithm and Quantum Phase Estimation (QPE).
    However to date QLDPC codes have been predominantly studied in the context of quantum memories; there has been no known method for implementing arbitrary logical Clifford operators in a QLDPC code proven efficient in terms of circuit depth.
    In combination with known methods for implementing $T$ gates, an efficient implementation of the Clifford group unlocks resource-efficient universal quantum computation.
    In this paper, we introduce a new family of QLDPC codes that enable efficient compilation of the full Clifford group via transversal operations.
    Our construction executes any $m$-qubit Clifford operation in at most $O(m)$ syndrome extraction rounds, significantly surpassing state-of-the-art lattice surgery methods. 
    We run circuit-level simulations of depth-126 logical circuits to show that logical operations in our QLDPC codes attains near-memory performance.
    These results demonstrate that QLDPC codes are a viable means to reduce the resources required to implement all logical quantum algorithms, thereby unlocking a reduced timeline to commercially valuable quantum computing.
\end{abstract}

\maketitle
\setlength{\belowdisplayskip}{1em} \setlength{\belowdisplayshortskip}{1em}
\setlength{\abovedisplayskip}{1em} \setlength{\abovedisplayshortskip}{1em}

\section{Introduction}
\vspace{-1em}
Quantum computers are poised to deliver the next major evolution in computational technology, with known applications in several high-impact sectors, including drug discovery, materials design, and cryptanalysis. However, by their nature, quantum technologies are highly susceptible to noise and are inherently incompatible with classical error correction techniques, leading to the development of unique quantum-specific error correction solutions. In quantum error correction (QEC), physical redundancy is introduced so that errors can be detected and corrected, ideally without harming the encoded information. QEC researchers have spent decades optimizing so-called ``planar" QEC codes, such as the surface code. These codes have many positive attributes, including nearest-neighbour connectivity for syndrome extraction, competitive error thresholds, and a high degree of symmetry. However, the physical resource requirements of the surface code are of surface code error correction are so onerous that experts have assessed that with these codes quantum technologies capable of breaking RSA-2048 (a benchmark often used to assess commercial utility) are likely to only arrive decades in the future~\cite{evolutionq_main}. Until efficient QEC is unlocked at scale, commercial quantum applications will not deliver material value to society.

Tantalizingly, there are other QEC codes --known as quantum low-density parity check (QLDPC) codes~\footnote{Note that throughout this paper we follow the convention of the field and implicitly exclude surface codes from consideration when discussing QLDPC codes, instead focusing on codes with higher encoding rate that meet the LDPC criteria.}-- that do not suffer from such high physical overheads. These codes have many of the same positive traits that planar codes do, leading to an explosion of interest in them over the past few years \cite{Breuckmann2024_main, Eberhardt2024_main, Gong2024_main, cross2024_main, Cowtan.2024_main, correlated-decoding-2Zhou2024_main, Quintavalle2023_main, Brown.2022,Quintavalle.2021_main}. 
However, despite significant work in this direction \cite{Quintavalle2023_main, cross2024_main, Krishna_2021_main}, it is not presently known how to compile depth-efficient logical operations in QLDPC codes. If provably efficient universal logical gate sets were to be found for QLDPC codes, the timeline for the availability of commercially relevant quantum computers could be brought in by years if not decades, owing to the reduced cost in physical qubit count.

In this paper, we propose QLDPC codes, \emph{Subsystem HYpergraph Product Simplex (SHYPS) codes} (see Fig. \ref{fig:stabsmain}), designed from the bottom up with logical gate implementation as the core consideration. We construct highly symmetric codes capable of implementing immense numbers of logical gates transversally \cite{Sayginel2024_main}, facilitating arbitrary logic with asymptotic circuit depths matching those for \emph{unencoded} logic. This approach is generally compatible with single-shot error correction, providing benefits in reduced qubit-overhead, reduced decoding complexity, and faster logical clock speeds. We demonstrate this compatibility by performing the most advanced circuit-level logical simulation to date, demonstrating near-memory performance for a compiled circuit derived from a randomly sampled logical Clifford circuit on 18 logical qubits.
\begin{figure}[ht]
    \includegraphics[width=0.68\linewidth]{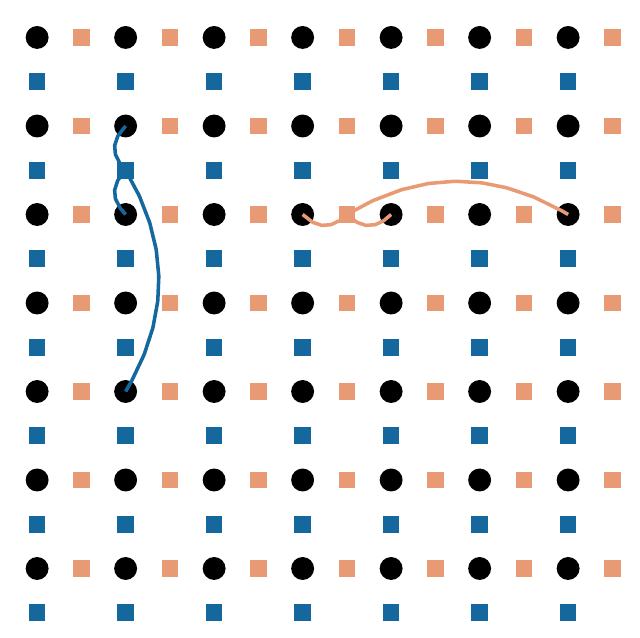}
    \caption{Gauge generator connectivity for $r=3$ SHYPS code. Black circles denote data qubits. Orange (blue) squares indicate auxiliary qubits used to measure 
$X$ ($Z$) gauge generators. Each gauge generator has weight 3 and is measured using three CNOT gates, depicted as edges connecting the corresponding data qubits to a single auxiliary qubit. Two representative gauge generators are illustrated: one $X$-type and one $Z$-type. The full set of gauge generators can be obtained by translating this pattern under periodic boundary conditions.}
    \label{fig:stabsmain}
\end{figure}
\vspace{-1em}
\section{Logic in QLDPC Codes}
\vspace{-1em}
The Clifford group is of particular interest for quantum computation as it provides a universal gate set when combined with any single non-Clifford operator. As there are known methods for implementing non-Clifford operations (e.g., $T$ gates) fault-tolerantly using state injection \cite{Litinski2019_main}, an efficient fault-tolerant implementation of the Clifford group for a given code is sufficient to unlock universal computation in that code.

To date, QLDPC codes have been predominantly studied in the context of quantum memories. Various techniques based on generalized lattice surgery \cite{Brown.2022_main, Cowtan.2024_main} have been proposed to rework logical Clifford implementations in planar codes for application in QLDPC codes. However, these methods often introduce substantial overheads in both time and space: operations are not guaranteed to be parallelizable, large auxiliary patches are required, and any single-shot properties \cite{Quintavalle.2021_main} of the QLDPC code can be compromised, as lattice surgery does not inherently support them.

In the generalized lattice surgery framework, computation is performed via joint logical operator measurements. State-of-the-art schemes for QLDPC surgery \cite{cross2024_main} measure single weight-$d$ logical operators with overheads of $O(d)$ in both space and time for a single logical cycle, where $d$ is the code distance. A constant number of such measurements may then be combined to implement elementary Clifford gates such as CNOT,  Hadamard and Phase. It may be possible that compiling methods will be developed to optimize depth-efficiency for these techniques. However, even assuming a high degree of parallelization commensurate with planar codes --which may or may not be possible in general-- direct compilation with these elementary gates to implement a worst-case $m$-qubit Clifford operator requires $O(m)$ logical cycles, leading to a total depth of $O(md)$.

Recent research \cite{correlated-decoding-2Zhou2024_main} suggests that logical Clifford execution based on transversal gates may avoid the $O(d)$ overhead for a single logical cycle, allowing syndrome information to accrue simultaneously with logical Clifford execution. This indicates that a QLDPC code family with sufficiently many transversal gates that also exhibits single-shot capabilities might execute arbitrary Clifford logic in lower depths with efficient decoding strategies. 
We generate a new family of highly symmetric QLDPC codes --\emph{Subsystem HYpergraph Product Simplex} (SHYPS) codes-- by combining the subsystem hypergraph product (SHP) code construction \cite{li2020numerical_main} with classical simplex codes \cite[Ch. 1.9]{book:MW}, and show that for these codes such low-depth logical Clifford implementations are indeed possible.
The SHYPS code family permits implementations of elementary Clifford gates in $O(1)$ time with typically zero space overhead (see \cref{tab:logcial_ops}). However, the chief figure of merit we analyze is the depth of a worst-case encoded Clifford.

Using new compilation techniques, we achieve a depth-efficient implementation of a worst-case $m$-qubit Clifford in only $4m(1+o(1))$ logical cycles, each consisting of a single depth-1 circuit followed by a depth-6 syndrome extraction round. This is comparable to state-of-the-art depth-efficiency, $2m+O(log^2(m))$, for a purely unitary implementation of Cliffords acting on unencoded qubits \cite{Maslov2022_main}.
The compilation strategy employed to achieve this low-depth implementation constructs circuits not as a sequence of elementary one- and two-qubit logical operations, but rather by leveraging many-qubit logical operations that occur naturally via low physical depth logical generators. Circuit-level simulations of these compiled circuits show near-memory performance (see \cref{fig:logic-simulations-plot}), demonstrating the feasibility of time-efficient logical execution in QLDPC codes.

\vspace{-1em}
\section{Compiling and Costs}
\label{sec:compiling-and-costs}
\vspace{-1em}

When computing in an error correction setting, logical operations are interleaved with rounds of syndrome extraction \footnote{Syndrome extraction results inform error correction, and corrections are often tracked in software and folded into the operators applied later in the circuit. For the purposes of this discussion, we restrict attention to syndrome extraction as this is the piece that is costly in terms of operations applied to qubits.}. The operators are typically drawn from a subset of the full set of logical operations, and we will refer to elements of this subset as \emph{logical generators}.
A sequence of logical generators combine to synthesize a logical operation that the circuit performs. To compute efficiently in a QEC code, a generating set of (ideally low-depth) logical generators is needed. The size of this generating set relative to the space of logical operations it generates, informs the worst-case number of generators needed to implement an arbitrary logical operation, and therefore the number of syndrome extraction steps that need to be interleaved between them. Since syndrome extraction is costly and lengthier computations are less desirable, the aim of compiling is to reduce the number of steps needed to implement a desired logical circuit. This reduction can be achieved by constructing codes which have more native low-depth logical generators. We restrict attention to the space of the $m$-qubit Clifford group, and examine the notion of efficient compilation in that context to quantify the number of logical generators needed to compute efficiently.

The size of the $m$-qubit Clifford group scales as $2^{2 m^2 + O(m)}$ \cite{Rengaswamy2018_main} (see also \cref{subsec:symplectic-representations}), which tends to pose a problem for compiling with a number of logical generators that doesn't also grow at least exponentially in $m$. Consider the case where we work with $\gamma$ fixed-depth logical generators of the $m$-qubit Clifford group. The set of all circuits composed of up to depth $D$ of these logical generators cannot produce more than $(\gamma+1)^D$ unique Clifford operators. Ignoring the action of the Pauli group (all Paulis can be pushed to the end and implemented via a depth-1 circuit in stabilizer codes), this means we need $D$ to be such that
\[
(\gamma+1)^D \geq 2^{2 m^2 + O(m)}
\]
to possibly produce every Clifford operator. Rearranging, this lower bounds the required depth to attain an arbitrary operation in the Clifford group:
\begin{align}
    \label{eq:lower_bound}
    D \geq \frac{2 m^2 + m - 1}{\log_2(\gamma+1)}=:D^*.
\end{align}
We can estimate that the fraction of total Clifford operators achievable with all depths $D < D^*$ is $1/(\gamma+1)$. Since we generally consider cases where $\gamma$ is a monotonically increasing function of $m$, asymptotically we expect $1/(\gamma+1) \rightarrow 0$ so that \emph{almost all} Cliffords require depth \emph{at least} $D^*$.

Consider synthesizing $m$-qubit Clifford operators using depth-1 circuits of arbitrary one- and two-qubit Clifford gates. These circuits are fixed under conjugation by the Clifford subgroup $\lrang{S,H,\text{SWAP}}$, and a counting exercise demonstrates that there are $\gamma\leq e/\sqrt{\pi}\left(10 m /e\right)^{m/2} - 1$ nontrivial depth-1 circuits, up to the right action of this subgroup. Substituting this expression into \cref{eq:lower_bound} yields an asymptotic worst-case (and average-case) depth lower bound of $D\geq 4m/\log_2(10 m/e) + O(1)$. While there are indeed compiling algorithms that achieve depth $O(m/\log_2 m)$ asymptotically \cite{jiang2020optimal_main}, the leading constants are impractically large \cite{Maslov2022_main}. Rather than accept these large overheads, researchers have derived alternative decompositions that achieve better depths in practice for realistic qubit counts with worst-case depths of $O(m)$ \cite{Maslov2022_main, bravyi2021hadamard_main, duncan2020graph_main, Sayginel2024_main}.

We employ an alternative decomposition (see Supplementary Material \ref{supmat:sec:SHSCompSummary}), which focuses on minimizing the rounds of specific subsets of Clifford gates. Every Clifford operator can be written in the five-stage decomposition $\diagZ - \diagX - \diagZ - \diagX - \diagZ(1)$ \footnote{One can consider this decomposition for the inverse or the Hadamard-conjugated version of a Clifford operator to find a total of four related decomposition (see Supplementary Materials \cref{supmat:sec:SHSCompSummary} for details). 
This freedom in reordering allows compilers to leverage interaction with neighbouring T gates effectively.
}, where we have defined subsets of the Clifford group as follows:
\begin{itemize}
	\item $\diagZ = \langle S, CZ\rangle$; $Z$-diagonal Cliffords,\vspace{-0.75em}
	\item $\diagX$; $X$-diagonal Cliffords (Hadamard-rotated $\diagZ$),\vspace{-0.75em}
	\item $\diagZ(1)$; depth-1 $Z$-diagonal operations.
\end{itemize}

Combined with a few additional insights \cite{Maslov2022_main}, any of the listed decompositions yield circuits with asymptotic depths of $2m + O(\log^2 m)$, which are state of the art for realistic qubit counts.
Compiling worst case logical circuits on $m$ patches of surface code with all-to-all connectivity for an arbitrary $m$-qubit logical Clifford will thus require roughly $2m$ logical cycles.

Having established via study of decompositions a worst-case depth for a logical Clifford in the surface code, we can now revisit \cref{eq:lower_bound} and consider the number of logical Clifford generators we require for a different error correcting code to achieve similar efficiency.
Clearly, $\gamma\sim 2^{O(m)}$ is necessary to achieve a logical cycle depth of $O(m)$.
However, note that because CSS codes always have transversal CNOT operations available as logical generators, when using large numbers of code blocks of a CSS code, we will always satisfy this requirement. 
This holds because the number of ways to pair up code blocks to apply cross-block operations scales exponentially in the number of code blocks (and hence the total number of logical qubits). Instead, we would like to capture the compiling behavior in code families at both the ``few'' code block scale and the ``many'' code block scale.

For $b$ code blocks of an $[n,k,d]$ code, the \emph{Clifford compiling ratio}
\[
    \frac{\textrm{depth of a worst case Clifford on }b\cdot k\textrm{ logical qubits }}{b\cdot k}
\]
captures compiling properties for an arbitrary number of code blocks \footnote{Note that the focus on worst case Cliffords is equivalently a focus on average case Cliffords. While we might be tempted to introduce a metric that looks more directly at the worst-case depth increase for a \emph{specific} Clifford operator, it turns out that such a metric is not particularly useful outside this instance. In particular, while there are circuits that achieve low depths in an unencoded setting that require high depths in the encoded setting, there will in general \emph{also} be logical circuits that require low physical depths to implement whose corresponding unencoded counterpart only has high depth implementations. Since the worst case Clifford depth comes from a purely combinatoric argument, these code-specific details become irrelevant, allowing for a straightforward comparison with real predictive power.}. Note that, throughout this paper, we adopt the following notation for error-correcting codes to ensure clarity: $(n,k,d)$ for classical codes, $[n, k, d]$ for subsystem codes, and $[[n, k, d]]$ for stabilizer codes. Any code that achieves an $O(1)$ ratio is said to generate the logical Clifford group \emph{efficiently}.
A code family with parameters $[n(r), k(r), d(r)]$ for which every member generates the logical Clifford group efficiently is also said to have this feature.
Any code family whose associated compiling is such that the Clifford compiling ratio scales with $k(r)$ is failing to keep pace with the surface code due to the depth of logical operations \emph{within} a code block, whereas scaling with $b$ is associated with overheads for compiling \emph{between} code blocks.

Motivated by the advantages of computing in a code that has many fault-tolerant logical generators of low depth, we introduce the SHYPS code family in \cref{sec:SHS}. This family has $\gamma = 2^{O(k)}$ logical generators for a \emph{single} code block, each implemented by a depth-1 physical circuit. Moreover, these logical generator implementations often require 0 additional qubits, rising to at most $n$ for in-block CNOT operators where a scheme involving an auxiliary code block is used \footnote{There exist methods to remove the need for this additional auxiliary block with the same asymptotic cost, but the actual circuit length tends to be larger in the low code block regime}.

In addition to possessing a sufficient number of logical generators, the SHYPS codes achieve the desired $O(1)$ Clifford compiling ratio, with logical Clifford operators across $b$ blocks implemented fault-tolerantly in depth at most $4bk(1+o(1))$ (see \cref{tab:logcial_ops}). Crucially, the depth of Clifford operations compiled in our SHYPS code framework remains independent of the code distance, compared to state-of-the-art lattice surgery methods where the depth scales as $ O(md) $. For a moderately sized code with distance $ d = 20 $, an SHYPS-compiled CNOT gate would achieve an order-of-magnitude reduction in depth relative to the equivalent compiled using lattice surgery ($4$ vs. roughly $40$). This example highlights that --in addition to reducing qubit overheads relative to the surface code-- an SHYPS-based quantum computer would provide substantially faster clock-speeds at the logical level. The efficient Clifford compiling and substantially lower qubit overheads of SHYPS codes result in a lower space-time volume for Clifford operators in SHYPS compared to surface codes with either lattice surgery or transversal gates.

\begin{table}[ht]
    \setlength{\abovecaptionskip}{0pt}
    \setlength{\belowcaptionskip}{-1em}
    \centering
    \setlength\tabcolsep{2.5mm}
    {\renewcommand{\arraystretch}{1.2}

    \begin{tabular}{c | c | c}
        \textbf{Logical Gate} & \textbf{Time cost} & \textbf{Space cost}
        \\ \hline
         CNOT (cross-block) & 4 & 0 \\
         CNOT (in-block) & 4 & $n$ \\
         $S$ (in-block) & 6 & 0 \\
         $CZ$ (cross/in-block) & 4 & 0 \\
         $H$ (in-block) & 8 & 0\\ \hline
         Arbitrary $b$-block Clifford & $4bk(1+o(1))$ & $0$
    \end{tabular}
    }
    \caption{Time and space costs for logical Clifford operations of $SHYPS(r)$ codes ($r\geq 4$) with parameters $[n,k]$. For $r=3$, the logical $S$ and $H$ gates have depths $9$ and $11$, respectively. Here, time corresponds to the number of complete syndrome extraction rounds.
    }
    \label{tab:logcial_ops}
\end{table}

\vspace{-1em}
\section{Clifford operators and automorphisms}
\vspace{-1em}

\subsection{Symplectic representations}\label{subsec:symplectic-representations}
\vspace{-1em}

The Clifford group $\calC_n$ is a collection of unitary operators that maps the Pauli group $\calP_n$ to itself, under conjugation. For example, the two-qubit controlled-not operator $CNOT_{i,j}$, the single-qubit phase gate $S_i$, and the single-qubit Hadamard gate $H_i$, are all Clifford operators, and in fact these suffice to generate the full group.
When considering logical gates of codes, it is convenient to utilise the well-known symplectic representation of Clifford gates: by definition $\calP_n$ is a normal subgroup of $\calC_n$ and the quotient $\calC_n/\calP_n$ is isomorphic to $Sp_{2n}(2)$, the group of $2n \times 2n$ binary symplectic matrices \cite[Thm. 15]{Rengaswamy2018}. Hence each Clifford operator is, up to Pauli, specified by a unique element $g \in Sp_{2n}(2)$. That this representation ignores Pauli factors is of no concern, as for stabilizer codes, any logical Pauli may be implemented transversally. 

The following examples illustrate the symplectic representation of some common families of Clifford operators; note that by convention we assume that elements of $Sp_{2n}(2)$ act on row vectors from the right. 
\begin{exmp}
    \textbf{(CNOT circuits)} The collection of CNOT circuits $\lrang{CNOT_{i,j}\,:\, 1\leq i,j \leq n}$ have symplectic representations
\[
\left\{ \begin{bmatrix}C & 0 \\ 0 & C^{-T}\end{bmatrix} \, : \,C \in GL_{n}(2)\right\},
\]
 where $GL_n(2)$ is the group of invertible $n \times n$ binary matrices. 
\end{exmp}
\begin{exmp}
    \textbf{(Diagonal Clifford operators)} The Clifford operators that act diagonally on the computational basis form an abelian group, generated by single-qubit phase gates $S_i$ and the two-qubit controlled-Z gate $CZ_{i,j}$. They are represented by symplectic matrices of the form
\[
\left\{ \begin{bmatrix}I_n&B\\0&I_n\end{bmatrix}\,:\,B\in M_n(2), B^T=B\right\},
\]
where the diagonal and off-diagonal entries of the symmetric matrix $B$, determine the presence of $S$ and CZ gates, respectively.
\end{exmp}

\vspace{-1em}
\subsection{Code automorphisms}\label{subsec:code-automorphisms}
\vspace{-1em}
Code automorphisms are permutations of the physical qubits that preserve the codespace. They are a promising foundation for computing in QLDPC codes as they can provide nontrivial logical gates implementable by low-depth SWAP circuits, or simply relabelling physical qubits. Moreover, combining automorphisms with additional transversal Clifford gates can give greater access to fault-tolerant logical gate implementations \cite{Breuckmann2024_main, Eberhardt2024_main, Quintavalle2023_main}.

Let $\calC$ be an $[n,k,d]$ CSS code with $X$- and $Z$-type gauge generators determined by matrices $G_X \in \gf^{r_X \times n}$ and $G_Z \in \gf^{r_Z \times n}$, respectively. 
The \emph{(permutation) automorphism group} $\aut{\calC}$ is the collection of permutations $\pi \in S_n$ that preserve the gauge generators, and therefore the codespace.
I.e., $\pi \in \aut{\calC}$ if there exist $g_{\pi,X}\in GL_{r_X}(2)$ and $g_{\pi,Z}\in GL_{r_Z}(2)$ such that
$
g_{\pi,X}G_X = G_X\pi$ and $g_{\pi,Z}G_Z = G_Z\pi.
$

The logical gate implemented by a given $\pi \in \aut\calC$ is determined by its action on the code's logical Pauli operators. In particular, as permutations preserve the $X/Z$-type of a Pauli operator, $\pi$ implements a logical CNOT circuit \cite[Thm. 2]{Grassl2013}. Furthermore, following \cite{Grassl2013_main}, a larger set of fault-tolerant CNOT circuits across two copies of $\calC$ may be derived by conjugating the target block of the standard transversal CNOT operator \footnote{The physical transversal CNOT operator implements logical transversal CNOT in any subsystem CSS code \cite[Sec. 5]{shor1997}} by $\pi$ (see Supplementary Materials \cref{supmat:fig:generalised-transversal-cnot}).

The symplectic representation for this combined operator is given by 
\begin{align}\label{eq:cross-block-cnot}
    \begin{bmatrix}
        I_n & \pi  & \\ 0 & I_n  & &  \\ &  & I_n & 0 \\  &   & \pi^{-1} & I_n\\
    \end{bmatrix}\in Sp_{4n}(2),
\end{align}
where we identify $\pi \in \aut{\calC}$ with the permutation matrix in $GL_n(2)$ whose $(i,j)$-th entry is 1 if $i=\pi(j)$, and zero otherwise. Note that as conjugation by $\pi$ simply permutes the targets of the physical transversal CNOT, this is a fault-tolerant circuit of depth 1. Moreover, (\ref{eq:cross-block-cnot}) implements a \emph{cross-block} CNOT operator, with all controls in the first code block of $k$ logical qubits, and all targets in the latter code block. 

More recently, code automorphisms have been generalized to include qubit permutations that exchange vectors in $G_X$ and $G_Z$. These so-called \emph{$ZX$-dualities} lead to low-depth logical operator implementations involving qubit permutations and single qubit Hadamard gates \cite{Breuckmann2024_main}. Moreover, $ZX$-dualities allow for the construction of logical diagonal Clifford operators in the following manner:
Let $\tau \in S_n$ be an involution ($\tau^2=1$) such that $G_X\tau = G_Z$, and suppose that $\pi \in \aut{\calC}$ is such that $\pi\tau$ is also an involution. Then the physical diagonal Clifford operator given by 
\begin{align}
    \label{eq:diag-type-1}
    \begin{bmatrix}
        I_n & \pi\tau \\ 0 & I_n
    \end{bmatrix} \in Sp_{2n}(2),
\end{align}
is a depth-1 circuit that implements a logical diagonal Clifford operator up to Pauli correction (see \cite{Breuckmann2024_main, Eberhardt2024_main} and Supplementary Materials Lem.~\ref{supmat:lem:depth 1 diagonal lifts}). As there is always a Pauli operator with the appropriate (anti)commutation relations with the stabilizers and logical operators of the code to fix any logical/stabilizer action sign issues, we can ignore this subtlety \cite{Sayginel2024}.

The requirement that $\pi\tau$ is an involution guarantees that the upper-right block of (\ref{eq:diag-type-1}) is symmetric, and thus corresponds to a valid diagonal Clifford operator. This generally restricts the number of automorphisms that may be leveraged to produce valid logical gates. Crucially, this is insignificant for the SHYPS codes we introduce in \cref{sec:SHS}, where we have sufficient symmetry to efficiently implement all diagonal operators in a code block. 

\vspace{-1em}
\section{Code constructions and logical gates}\label{sec:SHS}
\vspace{-1em}

The constructions (\ref{eq:cross-block-cnot}) and (\ref{eq:diag-type-1}) provide a framework for implementing logical Clifford operators with fault-tolerant, low-depth circuits. However the number of such operators that exist for a given subsystem CSS code $\calC$ clearly scales with the size of $\aut{\calC}$. 
This motivates a search for quantum codes with high degrees of permutation symmetry, to achieve the number of fixed-depth Clifford generators necessary for efficient compilation.
There are many methods for constructing quantum codes \cite{Calderbank1996_main, Tillich.2013_main, Panteleev.2021_main, Breuckmann_2021_main, Breuckmann.2021_main, Panteleev.2022_main, ibm-qmem_main}, but in this work we focus on a subsystem hypergraph product construction that allows us to leverage known highly symmetric classical codes, to produce quantum code automorphisms:
Let $H_i$ be parity check matrices for two classical $(n_i,k_i,d_i)$-codes with $i=1,2$ and codespace $\ker H_i$. The \emph{subsystem hypergraph product code} \cite{li2020numerical_main} (SHP), denoted $SHP(H_1,H_2)$, is the subsystem CSS code with gauge generators
\[
G_X = (H_1 \otimes I_{n_2}),\ G_Z = (I_{n_1} \otimes H_2),
\]
and parameters $[n_1n_2,k_1k_2, \min(d_1,d_2)]$ \cite[3.B]{li2020numerical}.

Now the classical codes $\ker H_i$ have analogously defined automorphism groups, and crucially these \emph{lift} to distinct automorphisms of $SHP(H_1,H_2)$.
\begin{lem}\label{lem:lifting-pairs-of-classical-auts}
Let $(\sig_1,\sig_2)\in \aut{\ker H_1}\times \aut{\ker H_2}$. Then $\sig_1\otimes \sig_2 \in \aut{SHP(H_1,H_2)}$. 
\end{lem} 

Note that a similar prescription for lifting classical symmetries has been recently proposed for the \emph{stabilizer} hypergraph product codes (see \cite[App. A.1]{xu2024} and \cite[App. D]{hong2024}), in the restricted instance that the linear transformations of the classical checks $H_i$ induced by each $\sigma_i$, are themselves also permutations. By considering the subsystem code here, we remove such limitations and produce a far greater number of quantum code automorphisms (see Supplementary Materials \cref{supmat:sec:SupMat-lifting classical auts} for further details).

To capitalise on \cref{lem:lifting-pairs-of-classical-auts}, we pair the SHP construction with the highly symmetric classical simplex codes, referring to these as \emph{subsystem hypergraph product simplex} (SHYPS) codes. A complete description of the SHYPS family (parameterized by integers $r\geq 3$) is given in Supplementary Materials \ref{supmat:sec:SHSsuppMat} but we note here that each instance, denoted $SHYPS(r)$ has parameters 
\[[n(r),k(r),d(r)] = [(2^r-1)^2,r^2,2^{r-1}].\] 
Moreover, this is a QLDPC code family, as each $SHYPS(r)$ has weight-3 gauge generators.
An immediate corollary of \cref{lem:lifting-pairs-of-classical-auts} and \cite[Ch. 8.5]{book:MW} is that
    \[
    \vert \aut{SHYPS(r)} \vert \geq \vert GL_{r}(2) \vert^2 = O(2^{2r^2}),
    \]
which grows exponentially in the number of logical qubits $k(r)=r^2$, as required in \cref{sec:compiling-and-costs}. By utilising these automorphisms with the constructions outlined in \cref{subsec:code-automorphisms} we are able to efficiently generate all CNOT and diagonal Clifford operators in the SHYPS codes. 

The logical action of operators (\ref{eq:cross-block-cnot}) and (\ref{eq:diag-type-1}) may be characterized explicitly:
For any pair $g_1,g_2 \in GL_r(2)$ there exists a corresponding automorphism $\sig_1 \otimes \sig_2 \in \aut{SHYPS(r)}$ such that the logical cross-block CNOT operator
\begin{align}\label{eq:cross-block-cnot-logical}
    \begin{bmatrix}
        I_k & g_1\otimes g_2  & \\ 0 & I_k  & &  \\ &  & I_k & 0 \\  &  & g_1^{-T}\otimes g_2^{-T} & I_k\\
    \end{bmatrix}\in Sp_{4k}(2),
\end{align}
is implemented by the depth-1 physical circuit of type (\ref{eq:cross-block-cnot}). Furthermore, arbitrary logical CNOT circuits on $b$ blocks of $k$ logical qubits can be constructed from a sequence of $2bk(1+o(1))$ such operators (see Supplementary Materials Thm. \ref{supmat:thm:cross-block-cnot-operators} and Cor. \ref{supmat:cor:b-block-cnot}).

To characterize the logical action of diagonal operators (\ref{eq:diag-type-1}) we first observe that the physical qubits of $SHYPS(r)$ may be naturally arranged in an $2^r-1 \times 2^r-1$ array such that the reflection across the diagonal is a $ZX$-duality exchanging $G_X$ and $G_Z$ (see Supplementary Materials Lem.~\ref{supmat:lem:depth 1 diagonal lifts}). We similarly arrange the logical qubits in an $r \times r$ array, and denote the reflection that exchanges rows and columns by $\tau$. 
Then for all $g \in GL_r(2)$, there exists $\sig \otimes \sig^T \in \aut{SHYPS(r)}$ such that the logical diagonal Clifford operator
\begin{align}
\begin{bmatrix}I& \left(g\otimes g^T \right)\cdot\tau\\0&I\end{bmatrix} \label{eq:diag-type-1-logical}
\end{align}
is implemented by a corresponding generator of type (\ref{eq:diag-type-1}).
The operators (\ref{eq:diag-type-1-logical}) have depth 1 and are fault-tolerant with circuit distance equal to the code distance. Moreover, they are alone sufficient to generate all logical diagonal Clifford operators on an SHYPS code block in depth at most $k(1+o(1))$ (Supplementary Materials Thm.~\ref{supmat:thm:DiagGensType1Generate}). 

For generation of the full logical Clifford group, we note that $SHYPS(r)$ possesses a \emph{Hadamard type} \cite{Breuckmann2024_main} fold-transversal gate whereby the logical all-qubit Hadamard operator (up to logical SWAP) is implemented fault-tolerantly (see Supplementary Materials Lem.~\ref{supmat:lem:all qubit hadamard}). By then applying a Clifford decomposition as discussed in \cref{sec:compiling-and-costs} (see also Supplementary Materials \ref{supmat:sec:SHSCompSummary}), we bound the cost of implementing an arbitrary Clifford operator:
\begin{thm}
    Let $r\geq 3$. An arbitrary logical Clifford operator on $b$ blocks of the $SHYPS(r)$ code may be implemented fault-tolerantly in depth $4bk(r)(1+o(1))$.
    \label{thm:arbitraryCliffordDepth}
\end{thm}

In particular, the SHYPS codes achieve the desired $O(1)$ Clifford compiling ratio. Moreover this bound is competitive with best known depths of $2bk+O(log^2(bk))$ for compiling Cliffords on $bk$ unencoded qubits \cite{Maslov2022_main}. We achieve further reductions in overhead for logical permutations and arbitrary Hadamard circuits (see Supplementary Materials Tables \ref{supmat:tab:cnot_depth}, \ref{supmat:tab:diagonal_depth} and \ref{supmat:tab:had_swap_depth} for details).

\vspace{-1em}
\section{Performance of the SHYPS Code}
\vspace{-1em}

We present numerical simulations to evaluate the circuit-level noise performance of the SHYPS code family. Two types of simulations were performed. First, in section \ref{sec:memorysim}, memory simulations for two different instances of SHYPS codes are benchmarked against comparably scaled surface codes. Second, in section \ref{sec:cliffordsim}, we present a circuit-level noise logic simulation on two blocks of the $[49, 9, 4]$ SHYPS code. The logical Clifford circuit is randomly sampled from the $18$ qubit Clifford group and synthesized into low-depth generators using the techniques explained in \cref{sec:SHS}. For a more detailed treatment of the simulations presented in this section, see Supplementary Materials \ref{supmat:fault-tolerant-demonstration}.

\vspace{-1em}
\subsection{Memory Simulation}\label{sec:memorysim}
\vspace{-1em}

\begin{figure}[t!]
\setlength{\abovecaptionskip}{0pt}
\setlength{\belowcaptionskip}{-1em}
    \includegraphics[width=\linewidth]{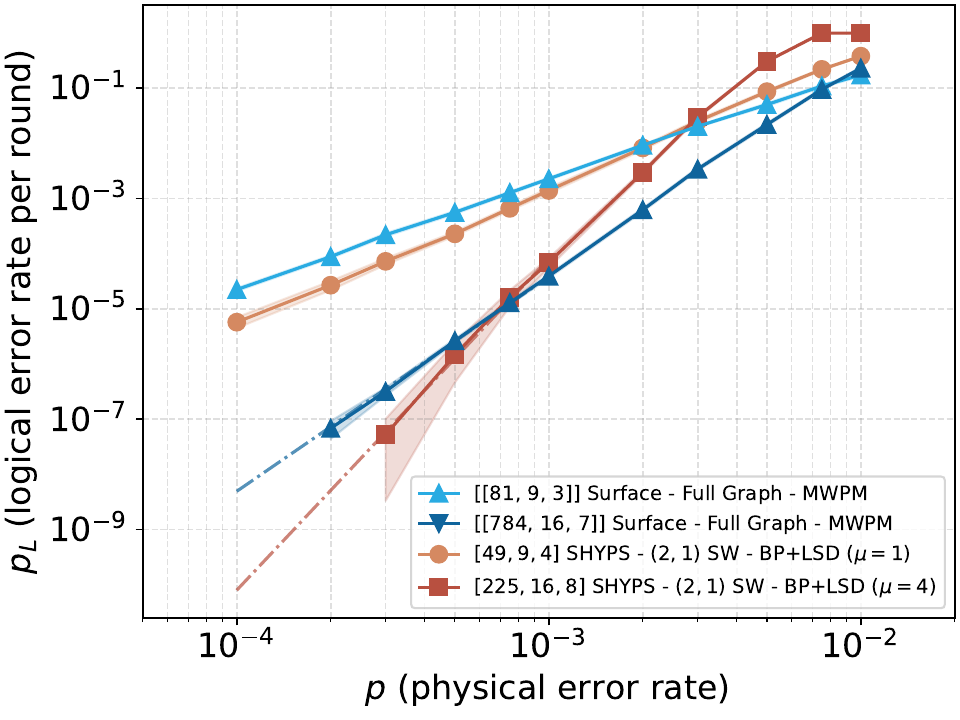}
    \caption{Simulation results for quantum memories under circuit-level noise for SHYPS and surface codes. For these simulations, only $Z$-type detectors are used.}
    \label{fig:memory-simulations-plot}
\end{figure}

Figure \ref{fig:memory-simulations-plot} shows the normalized logical error rate $p_L$ from memory simulations with $d$ syndrome extraction rounds. The logical error rates for the scaled surface codes are adjusted according to the number of logical qubits~\footnote{Note that this is equivalent to considering multiple distance-$d$ surface code patches required to provide the same number of logical qubits as the distance $d$ SHYPS code.} to demonstrate that the SHYPS code is competitive with the surface code while requiring fewer physical qubits per logical qubit. Specifically, we compare $[n, k, d]$ SHYPS codes with even $d$, to $k$ copies of [[$\underline{d}^{2}, 1, \underline{d}$]] rotated surface codes with uneven $\underline{d} = d-1$ \footnote{The maximum weight of guaranteed correctable errors for distance $d$ and $d-1$ codes is the same when $d$ is even.}. This comparison is shown in Table \ref{tab:code-comparison}, where we have used the code pseudo-thresholds and logical error rates at $p=\{3\times10^{-4}, 10^{-3}\}$ as the benchmarking metrics. 
The smallest member of the SHYPS family, the $[49,9,4]$ code, confers no pseudo-threshold advantages, but significantly outperforms the $[[81,9,3]]$ scaled distance-3 surface code in both physical qubit overhead (2x reduction) and logical error rates below pseudo-threshold. The other SHYPS code we simulated, the $[225,16,8]$ SHYPS code, is comparable in terms of performance to the $[[784,16,7]]$ scaled distance-7 surface code whilst reducing qubit count by $3.5\times$. It is true that the $[225,16,8]$ code has a lower pseudo threshold ($\approx 0.35\%$) than the $[[784,16,7]]$ scaled distance-7 surface code ($\approx 0.8\%$), but logical error rate decreases more aggressively in the sub-theshold regime: the slope of the logical error rates for both codes intersect at $p=0.05\%$, a point below which the $[[225, 16, 8]]$ SHYPS code outperforms the $[[784,16,7]]$ scaled distance-7 surface code. We anticipate that a decoder specifically tailored to SHYPS codes would further improve the pseudo-thresholds reported here.

\begin{table}[ht]
    \setlength{\abovecaptionskip}{0pt}
    \setlength{\belowcaptionskip}{-1em}
    \centering
    \setlength\tabcolsep{2.5mm}
    {\renewcommand{\arraystretch}{1.2}

    {\footnotesize
    \begin{tabular}{c | c | c | c}
         \textbf{Code} & \textbf{$P_L(3\times10^{-4})$} & \textbf{$P_L(10^{-3})$} & \textbf{$p_{th}$} 
        \\ \hline
         $[49, 9, 4]$ &$5\times 10^{-5}$ &$1.2\times 10^{-3}$ &$3.2\times 10^{-3} $  \\
         $[[81, 9, 3]]$ &$2\times 10^{-4}$ &$2.3\times 10^{-3}$ &$3.8\times 10^{-3} $ \\
         $[225, 16, 8]$ &$6.2\times 10^{-8}$ &$7\times 10^{-4}$ &$3.5\times 10^{-3} $  \\
         $[[784, 16, 7]]$ &$3\times 10^{-7}$ &$4\times 10^{-4}$ &$8\times 10^{-3} $ \\
    \end{tabular}
    }}
    \vspace{1em}
    \caption{\centering Logical error rates per round of error correction at physical error rates $3\times10^{-4} \ (0.03 \%)$, $10^{-3} \ (0.1 \%)$ and pseudo-thresholds. Here, $P_L(p)$ denotes the block logical error rate per round of syndrome extraction at physical error rate $p$. The pseudo-threshold is defined by the condition $P_L(p_{th} = p) = kp$, where $k$ is the number of logical qubits.}
    \label{tab:code-comparison}
\end{table}

\vspace{1em}

Most impressively, the SHYPS code simulations rely on a sliding window decoding approach with a small window size (2) and commit size (1). The use of a (2,1) sliding window decoder to achieve a high error correction performance is consistent with single-shot properties --the ability to decode based on a single syndrome extraction round between each logical operation. For a more detailed explanation about the single-shot properties of SHYPS codes, see Supplementary Materials \ref{supmat:fault-tolerant-demonstration}. 

\vspace{-1em}
\subsection{Clifford Simulation}\label{sec:cliffordsim}
\vspace{-1em}

We now apply our decomposition and logical generator constructions in a simulation of a randomly sampled Clifford operator on 18 logical qubits in two code blocks of the $[49, 9, 4]$ SHYPS code, up to an in-block CNOT circuit. This prevents the need to use extra auxiliary code blocks to implement in-block CNOTs while ensuring all types of transversal operations appear in the simulation. 
Synthesizing the sampled Clifford using the $\diagZ - \cnot - \diagX - \diagZ(1)$ decomposition requires 63 fault-tolerant logical generators, and consequently a total of 126 logical generators to implement both the Clifford circuit and its inverse.
The simulation proceeds as follows: initialize both code blocks in the encoded all-zero logical state; interleave each logical generator (both of the Clifford and its inverse) with gauge generator measurements; read out the state using a transversal $Z$ measurement.

\begin{figure}[t!]
\setlength{\abovecaptionskip}{0pt}
\setlength{\belowcaptionskip}{-1em}
    \includegraphics[width=\linewidth]{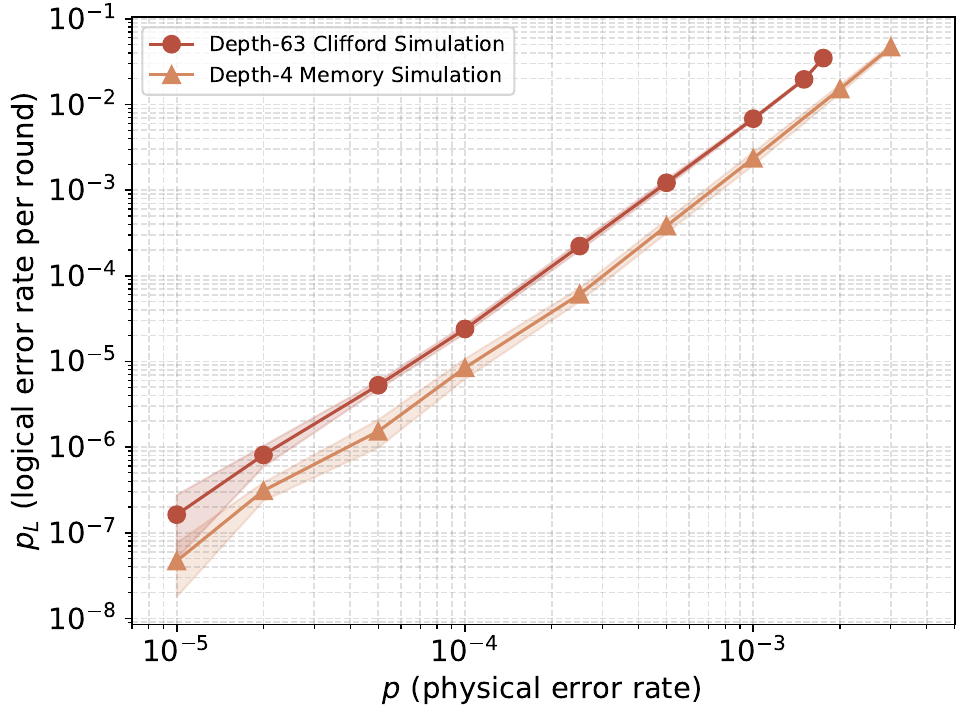}
    \caption{Simulation results for a depth-$126$ logical quantum circuit and a memory composed from two code blocks of the $[49, 9, 4]$ SHYPS code. 
    $Z$- and $X$-type detectors are used in both cases.}
    \label{fig:logic-simulations-plot}
\end{figure}

Figure \ref{fig:logic-simulations-plot} shows the normalized logical error rate $p_L$ 
for the described simulation. For comparison, we also include the logical error rate per syndrome extraction round for a $d$-round SHYPS quantum memory with $18$ logical qubits built from two blocks of the $[49, 9, 4]$ SHYPS code (adjusting SHYPS error rates to account for multiple codeblocks is analogous to scaling surface code data, see Supplementary Materials \ref{supmat:fault-tolerant-demonstration} for details).
Simulation of the Clifford operator achieves near-memory error correction performance in terms of error suppression, demonstrating that logical operations can be executed efficiently and fault-tolerantly in SHYPS codes.

\vspace{-1em}
\section{Conclusion}
\vspace{-1em}

QLDPC codes promise to reduce physical qubit overheads compared to existing surface codes. Despite their advantage as memories, it has been unclear whether any multi-logical qubit code (QLDPC or otherwise) could execute logical operations with as much parallelism as a surface code. Were this to remain unresolved, it may have crippled the speed of any QLDPC-code-based quantum computer when running highly-parallel circuits. By using a product construction of highly-symmetric classical simplex codes alongside novel compiling methods, we have shown that the resultant SHYPS code family shatters this barrier, achieving the same asymptotic logical circuit depth as \emph{unencoded} circuits under state-of-the-art algorithms. Remarkably, the resulting logical circuits retain strong fault-tolerance guarantees that are reflected in deep logical simulations showing near-memory performance even under circuit-level noise.

Two critical directions demand further exploration. Developing QEC codes with analogous properties but even better rates would enable further space improvements without making a time trade-off. More significantly, extending compiling parallelism to measurement parallelism would unlock for QLDPC codes every trick the surface code has at its disposal for reducing run-times via auxiliary code blocks. Without these advantages in parallelism, the reasons to consider using the surface code over QLDPC codes reduces to two factors: connectivity and simplicity. For any architecture where the necessary connectivity is achievable, it now seems all but certain: QLDPC codes are capable of driving down physical overheads \emph{without} increasing time overheads, and, as a result of this work, appear to be the most compelling path to quantum computers that can perform commercially relevant algorithms.

\section*{End notes}

\subsection*{Acknowledgements}

We thank Polina Bychkova, Zach Schaller, Bogdan Reznychenko, and Kyle Wamer for their contributions to the development of simulation infrastructure.

\subsection*{Author contributions}

A.J.M. and A.N.G. designed the quantum codes studied in this manuscript.
The depth costing of logical operators was done by A.J.M., A.S. and A.N.G.
P.F., D.C., J.R. and A.O.Q. designed the sliding window decoder and optimized decoder parameters used in the circuit-level numerical simulations.
P.F., D.C., R.H., J.R., A.O.Q. and A.N.G. performed the numerical simulations and post-processed the resulting data. The software used for these simulations was designed and written by C.D., R.H., A.E., P.F., J.R., and D.C..
A.J.M., A.N.G., P.F., D.C., A.S., J.R., A.O.Q., R.H., S.J.B., N.R.L.-H. and S.S. contributed to writing and editing the manuscript.

\subsection*{Competing interests} 

US Patent Application 19/263,208 (filed 8 July 2025, naming A.J.M and A.N.G. as co-inventors), PCT Application PCT/IB2025/057033 (filed 11 July 2025, naming A.J.M and A.N.G. as co-inventors), US Provisional Patent Application 63/720,973 (filed on November 15, 2024, naming A.N.G. and A.S as co-inventors), and US Provisional Patent Application 63/764,648 (filed 28 February 2025, naming A.N.G as the inventor) contain technical aspects from this paper.

\subsection*{Additional information}
Supplementary Information is available for this paper. Correspondence and requests for materials should be addressed to Stephanie Simmons at \url{ssimmons@photonic.com}.

\subsection*{Data availability} 
The simulation software to generate data reported in this paper is available at: \\ 
{\footnotesize{ \href{https://github.com/PhotonicInc/ComputingEfficientlyInQLDPCCodes}{https://github.com/PhotonicInc/ComputingEfficientlyInQLDPCCodes}.
}}


\begin{center}
  {\large \bf Supplementary Material} \\
\end{center}

The supplementary material is broken into 3 parts. First, in \cref{supmat:sec:math-prelims} we survey the necessary background information to introduce our new code family, the \emph{subsystem hypergraph product simplex codes} (SHYPS) codes. In particular, we discuss the notion of code automorphisms, and how the subsystem hypergraph product construction yields highly symmetric quantum codes, from well chosen classical inputs.

Next, in Sections \ref{supmat:sec: cnot gate characterisation}-\ref{supmat:sec: Hadamard} we demonstrate how automorphisms of the SHYPS codes can be leveraged to obtain low-depth implementations of logical Clifford operators. Taking each of the CNOT, diagonal, and Hadamard-SWAP families in turn, we produce fault-tolerant generators (typically physical depth 1), and demonstrate efficient compilation of arbitrary operators. Utilising a novel decomposition of Clifford operators, these results are combined in \cref{supmat:sec:SHSCompSummary} to yield an overarching bound on the depth of Clifford implementations in SHYPS codes. We refer the reader to Tables \ref{supmat:tab:cnot_depth}, \ref{supmat:tab:diagonal_depth} and \ref{supmat:tab:had_swap_depth}, for a detailed summary of results.

Lastly, in \cref{supmat:fault-tolerant-demonstration} we discuss syndrome extraction and decoding of the SHYPS codes, and we provide details about our numerical simulations.

\section{Mathematical preliminaries and code constructions}\label{supmat:sec:math-prelims}
We begin the preliminaries section with a review of the Pauli and Clifford groups, including the binary symplectic representation and examples of key types of Clifford operations that will be the focus of later sections.

\subsection{Review of Paulis and Cliffords} \label{supmat:subsec:review-of-paulis-and-cliffords}
\label{supmat:sec:backgroundSuppMat}

The \emph{Pauli group} on $n$ qubits is defined as
\begin{eqnarray*}
    \mathcal{P}_n & := & \langle \i, X_j, Z_j \mid j \in \{1,\dots,n\}\rangle\\
    & = & \langle X_j, Y_j, Z_j \mid j \in \{1,\dots,n\}\rangle.
\end{eqnarray*}\label{supmat:eq:pauli}
For many applications, it is convenient to ignore global phases involved in Pauli operators and instead consider elements of the phaseless Pauli group, $\calP_n/\lrang{\i}$. The phaseless Pauli group is abelian and has order $4^n$. Additionally, every non-trivial element has order 2, and hence $\calP_n/\lrang{\i} \cong \gf^{2n}$. This isomorphism can be made explicit in the following manner:  as $XZ = \i Y$, any  $P \in \calP_n/\lrang{\i}$ may be written uniquely as 
\[
\prod_{i=1}^nX_i^{u_i}Z_i^{v_i}=:X^uZ^v,
\]
where $u = (u_1,\dots,u_n)\in \gf^n$, and $v$ is defined analogously. The combined vector $(u\mid v)\in \gf^{2n}$ is known as the \emph{binary symplectic representation} of $P$. We equip $\gf^{2n}$ with a symplectic form $\omega:\gf^{2n}\times \gf^{2n} \longrightarrow \gf$ such that for $(u \mid v),(s \mid t)\in\gf^{2n}$,
\begin{eqnarray*}
    (u \mid v),(s \mid t) & \mapsto & ut^T+sv^T.
\end{eqnarray*}
This form captures the (anti-)commutativity of Paulis (that is lost when global phases are ignored), as $X^uZ^v$ and $X^sZ^t$ anti-commute if and only if \[\omega((u\mid v),(s \mid t))=1.\]
For example in $\calP_2$, the element $X_1$ anti-commutes with $Y_1Z_2$, and we check that
\begin{eqnarray*}
    X_1 & \mapsto & (1,0,0,0) \\
Y_1Z_2 \sim X_1 Z_1\cdot Z_2 & \mapsto & (1,0,1,1),
\end{eqnarray*}
and
\[
\omega((1,0,0,0),(1,0,1,1)) =1.
\]

The \emph{Clifford group} on $n$ qubits, denoted $\clif$, is the normalizer of the $n$-qubit Pauli group within $U(n)$. I.e., 
\begin{eqnarray*}
        \clif &:=& N_{U(n)}(\mathcal{P}_n) \\
        & = & \{g \in U(n)\,:\, gPg^{-1}\in \calP_n \text{ for all } P\in\calP_n\} \,.  
\end{eqnarray*}
For example, the two-qubit controlled-not operator $CNOT_{i,j}$, the single-qubit phase gate $S_i$, and the single-qubit Hadamard gate $H_i$, are all Clifford operators. Further examples include all gates of the form $e^{\i\theta}I$, but as the global phase is typically unimportant we restrict our attention to \[
\calC_n :=\clif/\{e^{\i\theta}\} \times \lrang{\frac{1+\i}{\sqrt{2}}}.\]
We likewise refer to this as the Clifford group, and note that
\[
 \calC_n = \lrang{CNOT_{i,j}, S_i, H_i \,:\, 1\leq i,j \leq n}.
\]
The action of a Clifford operator via conjugation corresponds to a linear transformation of $\gf^{2n}$ in the binary symplectic representation. Moreover, as conjugation preserves the (anti-)commutation of of Pauli operators, the corresponding linear transformations preserve the symplectic form. The collection of such linear transformations is known as the \emph{symplectic group} and is denoted $Sp_{2n}(2)$. Finally, observe that $\calP_n \leq \calC_n$ but that conjugation by a Pauli operator induces at most a change of sign, which is ignored by the binary symplectic representation. Taking the quotient by this trivial action, we see that $\calC_n/\calP_n \cong Sp_{2n}(2)$ \cite[Thm. 15]{Rengaswamy2018}.

The representation of Clifford operators by $2n \times 2n$ binary matrices is key to the efficient simulation of stabilizer circuits \cite{aaronson2004}. Moreover, as we'll demonstrate, it is a useful framework for synthesising efficient implementations of logical gates in quantum codes (see also \cite{Rengaswamy2018}).

We conclude this section with an explicit construction of some symplectic representations for Clifford operators. Following \cite{aaronson2004}, we assume that matrices in $Sp_{2n}(2)$ act on row vectors from the right. I.e., given $P\in \calP_n$ with binary symplectic representation $(u\mid v)$, and $g\in \calC_n/\calP_n$ with binary symplectic representation $G$, we have $gPg^{-1} \longleftrightarrow (u\mid v)G$. In particular, the images of the Pauli basis $X_i,Z_i$, are given by the $i$th and $(i+n)$th rows of $G$, respectively.

\begin{exmp}\phantomsection\label{supmat:ex:symplectic-rep-cnots}
    \textbf{(CNOT circuits)} As CNOT circuits map $X$-type Paulis to other $X$-type Paulis (and similar for $Z$-type), in $Sp_{2n}(2)$ they have the form
\[
\{ \begin{pmatrix}C & 0 \\ 0 & C^{-T}\end{pmatrix} \, : \,C \in GL_{n}(2)\}.
\]
 For simplicity we typically describe such operators solely by the matrix $C$. 
 
 For example, $CNOT_{1,2} \in \calC_3$ has $C$ defined as follows \[\begin{pmatrix}
     1 & 1 & 0 \\ 0 & 1 & 0 \\ 0 & 0 & 1
 \end{pmatrix} \in GL_3(2).\]
 
\end{exmp}
\begin{exmp}\label{supmat:ex:symplectic-rep-diagonal}
    \textbf{(Diagonal Clifford operators)} The Clifford operators that act diagonally on the computational basis form an abelian group, generated by single-qubit phase gates $S_i$ and the two-qubit controlled-Z gate $CZ_{i,j}$. Hence, modulo Paulis, they are represented by symplectic matrices of the form
\[
\{ \begin{pmatrix}I&B\\0&I\end{pmatrix}\,:\,B\in M_k(2),  B^T=B\},
\]
where the diagonal and off-diagonal entries of the symmetric matrix $B$ determine the presence of $S$ and CZ gates, respectively.

For example, $S_1\cdot CZ_{1,2} \in \calC_2$ corresponds to
\[
\begin{pmatrix}
    1 & 0 & 1 & 1 \\ 0 & 1 & 1 & 0 \\ 0 & 0 & 1 & 0 \\ 0 & 0 & 0 & 1
\end{pmatrix} \in Sp_4(2).
\]
And the action of this example on a Pauli, $X_1$,
\[
(S_1\cdot CZ_{1,2})\cdot X_1 \cdot (S_1\cdot CZ_{1,2})^{-1} = Y_1Z_2 \]
corresponds to \[(1, 0, 0, 0)\cdot \begin{pmatrix}
    1 & 0 & 1 & 1 \\ 0 & 1 & 1 & 0 \\ 0 & 0 & 1 & 0 \\ 0 & 0 & 0 & 1
\end{pmatrix} = (1,0,1,1).
\]
\end{exmp}
\begin{exmp} \label{supmat:ex:symplectic-rep-hadamard}\textbf{(Hadamard circuits)} In $Sp_{2n}(2)$, Hadamard circuits $H^v = \prod H_i^{v_i}$ have the form
\[\{
\begin{pmatrix}
I_n + \text{diag}(v) & \text{diag}(v) \\
\text{diag}(v) & I_n+\text{diag}(v) \\
\end{pmatrix}
\, :\, v\in \gf^n \},
\]
where $\text{diag}(v)$ is the diagonal matrix with entries $v_1,\dots,v_n$.

For example, the transversal Hadamard operator $H^{\otimes 3} = H^{(1,1,1)} \in \calC_3$ corresponds to 
\[
\begin{pmatrix}
    0 & I_3 \\ I_3 & 0
\end{pmatrix} \in Sp_6(2).
\]
\end{exmp}

\subsection{Subsystem codes}
\label{supmat:sec:subsystem codes}

A quantum stabilizer code on $n$ physical qubits is the common $+1$ eigenspace of a chosen abelian subgroup $\calS \leq \calP_n$ with $-I\not\in\calS$. The subgroup $\calS$ is known as the \emph{stabilizer group}, and moreover if $\calS$ admits a set of generators that are either $X$-type or $Z$-type Pauli strings, then the code is called \emph{CSS} \cite{Calderbank1996}.
Quantum subsystem codes are the natural generalisation of stabilizer codes, in that they are defined with respect to a generic subgroup $\calG \leq \calP_n$ known as the \emph{gauge group} \cite{Poulin2005}. Moreover, subsystem codes are typically interpreted as the subsystem of a larger stabilizer code whereby a subset of the logical qubits are chosen to not store information and the action of the corresponding logical operators is ignored. 
More formally, given gauge group $\calG$, the corresponding stabilizer group $\calS$ is the centre of $\calG$ modulo phases
\[\lrang{\calS,iI} = Z(\calG):=C_{\calP_n}(\calG) \cap \calG.\]
Phases are purposefully excluded (in particular, $-I \notin \calS$) to ensure the fixed point space of $\calS$ is nontrivial, and said space decomposes into a tensor product $\calC_{\calL}\otimes \calC_{\calG}$, where elements of $\calG \backslash \calS$ fix only $\calC_{\calL}$. The subspaces $\calC_{\calL}$ and $ \calC_{\calG}$ are said to contain the logical qubits, and gauge qubits, respectively.

The logical operators of the subsystem code are differentiated into two types, based on their action on $\calC_{\calL}\otimes \calC_{\calG}$: those given by $C_{\calP_n}(\mathcal{G})\backslash \calG$ that act nontrivially only on $\calC_{\calL}$ are known as \emph{bare} logical operators. Whereas operators acting nontrivially on both $\calC_{\calL}$ and $\calC_{\calG}$ are known as \emph{dressed} logical operators, and are given by $C_{\calP_n}(\mathcal{S})\backslash \calG$. Note that a dressed logical operator is a bare logical operator multiplied by an element of $\calG \backslash \calS$. The minimum distance $d$ of the subsystem code is the minimal weight Pauli that acts nontrivially on the logical qubits $\calC_{\cal{L}}$, i.e., the minimum weight of a dressed logical operator
\[
d = \min\{\vert P \vert \, : \, P \in C_{\calP_n}(\mathcal{S})\backslash \calG\}.
\]
We say that a subsystem code is an $[n,k,d]$ code if it uses $n$ physical qubits to encode $k$ logical qubits with distance $d$. The notation $[n, k, g, d]$ that additionally indicates the number of gauge qubits $g$, is also commonplace (but unused throughout this paper).

The advantages of subsystem codes are most evident when the gauge group $\calG$ has a generating set composed of low-weight operators but the stabilizer group $\calS$ consists of high-weight Paulis. Measuring the former operators requires circuits of lower depth (therefore reducing computational overheads), and the measurement results can be aggregated to infer the stabilizer eigenvalues. In fact, for the codes considered here, the difference in the stabilizer weights and gauge operator weights is a factor of the code distance (see \cref{supmat:fault-tolerant-demonstration} for more details). We note that it is exactly the measurement of operators in $\calG \backslash \calS$ that act nontrivially on $\calC_{\calG}$, that prevents the gauge qubits from storing information during computation, as these measurements impact the state of $\calC_{\calG}$.

In this work we restrict our attention to $[n, k, d]$-subsystem codes that are also of CSS type, with $X$- and $Z$-type gauge generators determined by matrices $G_X \in \gf^{r_X \times n}$ and $G_Z \in \gf^{r_Z \times n}$, respectively. Here each vector $v \in \gf^{n}\cap \text{RowSpace}(G_X)$ denotes an $X$-type gauge operator $X^v$, with $Z$-type gauge generators similarly defined. The associated stabilizers will also be of CSS type, and denoted by $S_X, S_Z$.

\subsection{Automorphisms of codes}
\label{supmat:sec:automorphisms of codes}

Code automorphisms are a promising foundation for computing in QLDPC codes as they can provide nontrivial logical gates implementable by permuting, or in practice simply relabelling, physical qubits.
In this section, we review permutation automorphisms of classical and quantum codes, which will serve as the backbone of our logical operation constructions.

Let's first set some notation: given a permutation $\sigma \in S_n$ of the symbols $\{1,2,\dots,n\}$ in cycle notation, we identify $\sigma$ with the permutation matrix whose $(i,j)$th entry is 1 if $i=\sigma(j)$, and zero otherwise. For example
\[
(1,2,3)(4,5) \in S_5 \mapsto \begin{pmatrix}
    0 & 0 & 1 & 0 & 0\\ 1 & 0 & 0 & 0 & 0 \\ 0 & 1 & 0 & 0 & 0 \\ 0 & 0 & 0 & 0 & 1 \\ 0 & 0 & 0 & 1 & 0
\end{pmatrix}.
\]
The permutation defined above maps the initial indices $\{1,2,3,4,5\}$ to $\{2,3,1,5,4\}$. So given a standard basis vector $e_i \in \gf^n$, $\sigma$ acts on row vectors on the right as $e_i\sigma = e_{\sigma^{-1}(i)}$, and on column vectors on the left as $\sigma e_i^T = e_{\sigma(i)}^T$.

\begin{defn}
    Let $C$ be an $(n,k,d)$-classical linear code described by a generator matrix $G \in \gf^{k \times n}$. Then the \emph{(permutation) automorphism group} $\aut{C}$ is the collection of permutations $\sigma \in S_n$ that preserve the codespace. I.e., $\sigma \in \aut{C}$ if for all $c =(c_1,\dots,c_n)\in C$,
    \begin{align} \label{supmat:eq:classical code aut}
         c\sig=(c_{\sig(1)},\dots,c_{\sig(n)})\in C.
    \end{align}
\end{defn}
As $C$ is linear, it suffices to check (\ref{supmat:eq:classical code aut}) on the basis given by the rows of $G$. Hence $\sig \in \aut{C}$ if and only if there exists a corresponding $g_{\sig}\in GL_k(2)$ such that $G\sig = g_{\sig}G$ \cite[Lem. 8.12]{book:MW}. In particular, $g_{\sigma}$ represents the invertible linear transformation of the $k$ logical bits, induced by the permutation.

The parity checks of the code $C$ are similarly transformed by automorphisms: letting $C^\perp$ denote the dual classical code, we have $\aut{C}= \aut{C^\perp}$ \cite[Sec. 8.5]{book:MW}. So given a parity check matrix $H\in \gf^{(n-k)\times n}$ for $C$ we have equivalently that $\sig \in \aut{C}$ if there exists $h_{\sig} \in GL_{n-k}(2)$ such that $H\sigma = h_{\sigma}H$.

This definition extends naturally to quantum codes:
\begin{defn}
    Let $C$ be an $[n,k,d]$ (CSS) subsystem code with $X$ and $Z$ type gauge generators determined by $G_X \in \gf^{r_X \times n}$ and $G_Z \in \gf^{r_Z \times n}$, respectively. Then $\aut{C}$ consists of permutations that preserve the gauge generators, i.e., $\sig \in S_n$ such that 
\[
g_{\sig,X}G_X = G_X\sig, \ \,g_{\sig,Z}G_Z = G_Z\sig,
\]
for some $g_{\sig,X}\in GL_{r_X}(2)$ and $g_{\sig,Z}\in GL_{r_Z}(2)$.
\end{defn}

The logical gate implemented by $\sig \in \aut{C}$ is determined by the permutation action on the code's logical Pauli operators. 
In particular, $\sig$ always gives rise to a permutation of the logical computational basis states of $C$ corresponding to a logical CNOT circuit \cite[Thm. 2]{Grassl2013}. 

In later sections we outline quantum constructions utilising classical codes, and consequently how classical automorphisms may be leveraged to produce quantum logical gates. In these instances, the logical CNOT implemented by the quantum code automorphism is a function of the associated classical linear transformations. This motivates an investigation of classical codes with high degrees of symmetry.

\subsection{Classical simplex codes} \label{supmat:sec:appendix simplex codes} Let $r\geq 3$ and define $n_r=2^{r}-1$, $d_r=2^{r-1}$. The classical simplex codes, denoted $C(r)$ are a family of $(n_r,r,d_r)$-linear codes, that are dual to the well-known Hamming codes. 
More specifically, we consider $C(r)$ with respect to a particular choice of parity check matrix: for each $3 \leq r < 500$ \footnote{The upper bound of $500$ clearly encapsulates all codes that will ever be used in practice. That the result in fact holds for all $r$ is an open conjecture \cite{blake96}}, there exists a three term polynomial $h(x)=1+x^a+x^b \in \gf[x]/\lrang{x^{n_r}-1}$ such that $\gcd(h(x),x^{n_r}-1)$ is a primitive polynomial of degree $r$ \cite{blake96}. Then the $n_r\times n_r$ matrix
\[
H = \begin{pmatrix}
    h(x) \\ xh(x) \\ \vdots \\ x^{n_r-1}h(x)
\end{pmatrix}
\]
is a parity check matrix (PCM) for $C(r)$ \cite[Lem. 7.5]{book:MW}, where here we adopt the usual polynomial notation for cyclic matrices
\[
\sum_{i=0}^{n_r-1} a_ix^i\mod x^{n_r}-1 \mapsto (a_0,a_1,\dots,a_{n_r-1})\in \gf^{n_r}.
\]

Note there are many alternative choices for the PCM of $C(r)$; in fact $H$ chosen here has greater than $n_r-r$ rows and so this description contains redundancy. However what the choice above guarantees is that each row and column of $H$ has weight 3, leading to low-weight gauge generators and optimal syndrome extraction scheduling, of an associated quantum code (see \cref{supmat:fault-tolerant-demonstration}).

The simplex codes are examples of highly symmetric classical codes with large automorphism groups.
\begin{lem}\label{supmat:lem:aut group of simplex code} 
    The automorphism group of the simplex codes are as follows:
    \[
     \aut{C(r)} \cong GL_r(2).
    \]
\end{lem}

Intuitively, \cref{supmat:lem:aut group of simplex code} means that each invertible linear transformation $g \in GL_r(2)$ of the $r$ logical bits, is implemented by a distinct permutation $\sig$ of the $n_r$ physical bits \cite[Ch. 8.5]{book:MW}.
\begin{exmp}\label{supmat:exmp: classical simplex r=3}
    Let $r=3$. The polynomial $h(x)=1+x^2+x^3$ is primitive and hence the overcomplete parity check matrix
    \[
    H = \begin{pmatrix}
        1 & 0 & 1 & 1 & 0 & 0 & 0 \\
        0 & 1 & 0 & 1 & 1 & 0 & 0 \\
        0 & 0 & 1 & 0 & 1 & 1 & 0 \\
        0 & 0 & 0 & 1 & 0 & 1 & 1 \\
        1 & 0 & 0 & 0 & 1 & 0 & 1 \\
        1 & 1 & 0 & 0 & 0 & 1 & 0 \\
        0 & 1 & 1 & 0 & 0 & 0 & 1  
    \end{pmatrix}
    \]
    defines the $(7,3,4)$-simplex code. A basis for $C(3)$ is given by the rows of generator matrix 
    \begin{align}\label{supmat:eq:r=3 gen matrix}
            G= \begin{pmatrix}
        1 & 0 & 1 & 1 & 1 & 0 & 0 \\
        0 & 1 & 0 & 1 & 1 & 1 & 0 \\
        0 & 0 & 1 & 0 & 1 & 1 & 1 
    \end{pmatrix},
    \end{align}

    and we observe that
    \[
    \begin{pmatrix}
        1 & 1 & 0 \\ 0 & 1 & 0 \\ 0 & 0 & 1
    \end{pmatrix}\cdot G = G \cdot (2,4)(5,6). 
    \]
    I.e., the bit permutation $(2,4)(5,6)\in \aut{C(3)}$ induces the linear transformation \[(v_1,v_2,v_3) \mapsto (v_1+v_2,v_2,v_3)\] on the 3 logical bits $(v_1,v_2,v_3)\in \gf^3$.
\end{exmp}
In the remainder of this work, we assume that all parity check matrices $H\in \gf^{n_r\times n_r}$ for the classical simplex code are taken as above.

\subsection{Subsystem hypergraph product simplex (SHYPS) codes}\label{supmat:sec:SHSsuppMat}
Here we describe our main quantum code construction, namely the \emph{subsystem hypergraph product simplex code}, in greater detail.
\begin{defn}
    Let $r\geq 3, n_r=2^{r}-1, d_r=2^{r-1}$, and let $H_r$ be the parity check matrix for the $(n_r,r,d_r)$-classical simplex code. Then the subsystem hypergraph product of two copies of $H_r$, denoted $SHYPS(r)$, is the subsystem CSS code with gauge generators 
\begin{align} \label{supmat:eqn:subsystem-hgp}
G_X = (H \otimes I_{n_r}),\ G_Z = (I_{n_r} \otimes H).
\end{align}
 We call $SHYPS(r)$ the \emph{subsystem hypergraph product simplex code}.
\end{defn}

It's clear by definition that the row/column weights of $G_X$ and $G_Z$ match those of $H$, and so the $SHYPS(r)$ codes are a QLDPC code family with gauge generators of weight 3. Moreover, the gauge generators have a particular geometric structure: we arrange the $n_r^2$ physical qubits of $SHYPS(r)$ in an $n_r \times n_r$ array with row major ordering, such that for standard basis vectors $\lrang{e_1,\dots,e_{n_r}}=\gf^{n_r}$, the vector $e_i \otimes e_j$ corresponds to the $(i,j)$th position of the $n_r \times n_r$ qubit array. 
For example, the first row of $H$ given in Example \ref{supmat:exmp: classical simplex r=3} is $r_1 = e_1 + e_3 + e_4$ and hence the first row $r_1 \otimes e_1 \in G_X$ indicates an $X$-type gauge generator supported on the vector
\[
r_1 \otimes e_1 = e_1\otimes e_1+e_3 \otimes e_1 + e_4 \otimes e_1,
\]
i.e., on the $(1,1),(3,1)$ and $(4,1)$ qubits of the array.
It follows that gauge generators in $G_X$ (respectively $G_Z$) are supported on single columns (respectively rows) of the qubit array. 

We follow \cite[Sec. 3B]{li2020numerical} to describe the parameters $[n,k,d]$ of $SHYPS(r)$: firstly recall from the above that $n=n_r^2$. Next, to calculate the number of encoded qubits $k$, observe that the Pauli operators that commute with all gauge generators, are generated by
\begin{align}\label{supmat:eq:logical ops}
    \mathcal{L}_X = (I_{n_r} \otimes G), \
\mathcal{L}_Z = (G \otimes I_{n_r}).
\end{align}
Here $G$ is a chosen generator matrix for the classical simplex code (so in particular $G$ is a matrix of rank $r$ such that $HG^T=0$).

The centre of the gauge group $\lrang{G_X,G_Z}$ determines the stabilizers of a subsystem code. In particular, for $SHYPS(r)$ these are generated by 
\begin{equation}\label{supmat:eq:stabilizers_from_guage}
    S_X = (H\otimes G),
\
S_Z = (G \otimes H).
\end{equation}

Finally, $k$ is calculated by comparing the ranks of $\calL_X$ and $S_X$:
\begin{eqnarray*}
    k &=& \rank \calL_X - \rank S_X \\
    & = & n_r\cdot r-(n_r-r)\cdot r \\
    & = & r^2.
\end{eqnarray*}
More specifically, $\text{RowSpan}(\mathcal{L}_X \backslash S_X)$ determines the space of logical $X$-operators, with logical $Z$-operators defined similarly. However, as indicated in \cref{supmat:sec:subsystem codes}, we consider the action of these operators up to multiplication by the gauge group. Hence the minimum distance $d$ of the code is given by the minimum weight operator in $\lrang{\mathcal{L}_X,\mathcal{L}_Z}\cdot \lrang{G_X,G_Z}-\lrang{G_X,G_Z}$. For subsystem hypergraph product codes, this is exactly the minimum distance of the involved classical codes, and hence $d(SHYPS(r))=d_r=2^{r-1}$ \cite[Sec. 3B]{li2020numerical}.

In summary, we have the following
\begin{thm} \label{supmat:thm:SHS code definition}
   Let $r\geq 3$ and let $H$ be the parity check matrix for the $(2^{r}-1,r,2^{r-1})$-classical simplex code described in \cref{supmat:sec:appendix simplex codes}. The subsystem hypergraph product simplex code $SHYPS(r)$ is an $[n,k,d]$-quantum subsystem code with gauge group generated by $3$-qubit operators and 
  \begin{eqnarray*}
       n &= & (2^r-1)^2, \\ k &=& r^2, \\ d&=&2^{r-1}.
  \end{eqnarray*}
\end{thm}
It's evident from the tensor product structure of (\ref{supmat:eq:logical ops}) that like the physical qubits, the logical qubits of $SHYPS(r)$ may be arranged in an $r \times r$ array, such that logical operators have support on \emph{lines} of qubits. In fact, recent work \cite{Quintavalle2023} demonstrates that a basis of logical operators may be chosen such that pairs of logical $X/Z$ operators have supports intersecting in at most one qubit. The following result is an immediate application of \cite[Thm. 1]{Quintavalle2023}:
\begin{thm}\label{supmat:thm:logical basis for SHP}
    Let $r\geq 3$. There exists a generator matrix $G \in \gf^{r\times n_r}$ for the classical simplex code and an ordered list of $r$ bit indices $\pi(G)\subset [1..n_r]$ known as \emph{pivots}, such that $PG^T=I_r$ for the \emph{pivot matrix}  $P\in \gf^{r\times n_r}$:  $P_{i,j}=1$ if and only if $ \pi(G)[i]=j$. Moreover, the matrices
    \[\label{supmat:eq:SHS canonical basis}
    L_X = (P \otimes G), \
L_Z = (G \otimes P),
    \]
    form a symplectic basis for the logical $X/Z$ operators of $SHYPS(r)$.
\end{thm}

The proof of \cite[Thm. 1]{Quintavalle2023} is constructive and, as the name pivots suggests, relies on a modified version of the Gaussian elimination algorithm \cite[Alg. 1]{Quintavalle2023}. This yields a so-called \emph{strongly lower triangular} basis for the simplex code, represented by $G$, with rows $\{g_i\}_{i \in \pi(G)}$ indexed by the pivots.
As $P$ has row weights equal to one, we see that the matrix products $PG^T=I_r$ and $L_XL_Z^T=I_{r^2}$ hold not only over $\gf$ but over $\mathbb{R}$. Hence, given pairs of pivots $(i,j),(k,l)\in \pi(G)^2$, the associated logical Paulis $\overline{X}_{i,j}, \overline{Z}_{k,l}$ given by basis vectors $e_i\otimes g_j \in L_X$ and $g_k \otimes e_l \in L_Z$ have intersecting support if and only if $(i,j)=(l,k)$. Moreover, this intersection is on the $(i,j)$th qubit of the array.
From this point, we assume that $G$, and the logical operators of $SHYPS(r)$, are of the form above.

\subsection{Lifting classical automorphisms}\label{supmat:sec:SupMat-lifting classical auts}

The geometric structure of $SHYPS(r)$ suggests a natural way to \emph{lift} automorphisms of the simplex code by independently permuting either rows or columns of the qubit array. In this manner, we see that $SHYPS(r)$ inherits $\vert GL_r(2)\vert^2 =O(2^{2r^2})$ automorphisms from the two copies of the classical simplex code, and hence $\aut{SHYPS(r)}$ grows exponentially with the number of logical qubits $k=r^2$.

\begin{lem}
    Let $r\geq 3$ and $C$ be the $(n_r,r,d_r)$-simplex code with automorphisms $\sig_1,\sig_2\in \aut{C}$. Then $\sig_1\otimes \sig_2\in \aut{SHYPS(r)}$. Furthermore, 
    \[
    \vert \aut{SHYPS(r)}\vert \geq \vert GL_r(2) \vert ^2.
    \]  
\end{lem}

\begin{proof}
    Clearly $\sig_1 \otimes \sig_2$ is a permutation of the required number of qubits $n=n_r^2$ and so we need only check that it preserves the gauge generators $G_X$ and $G_Z$. By definition of code automorphisms, for each $\sig_i \in \aut{C}$ there exists a corresponding $h_{\sig_i}\in GL_{n_r}(2)$ such that $h_{\sig_i}H=H\sig_i$. Hence 
   \begin{eqnarray*}
       G_X\cdot (\sig_1 \otimes \sig_2) & = & (H\sig_1 \otimes \sig_2) \\
       & = & (h_{\sig_1}H\otimes \sig_2) \\
       & = & (h_{\sig_1}\otimes \sig_2)\cdot G_X,
   \end{eqnarray*}
   with $G_Z$ similarly preserved.
   For the second claim, note that a tensor product of matrices $A\otimes B$ is the identity if and only if $A$ and $B$ are also identity. Hence each pair $\sig_1,\sig_2\in \aut{C_1}\times \aut{C_2}$ produces a distinct $\sig_1\otimes \sig_2 \in \aut{SHYPS(r)}$.
\end{proof}
Note that the above result naturally generalises to all subsystem hypergraph product codes, constructed from possibly distinct classical codes, but it does not generalise to the \textit{stabilizer} hypergraph product construction (HGP). A similar prescription for lifting classical automorphisms is available for HGP codes \cite[App. A.1]{xu2024} and \cite[App. D]{hong2024}, but only in the restricted instance that each $h_{\sig_i}$ is also a permutation (see \cite[Diagram A.3]{xu2024}, where $h_{\sig_i}$ are denoted $\gamma_0^i$).
By considering the subsystem hypergraph product in this work, we are able to lift \textit{all} $\vert GL_r(2)\vert$ classical simplex code automorphisms.

To determine the logical gate induced by the above automorphisms we examine the permutation action on the logical basis given by \cref{supmat:thm:logical basis for SHP}.

\begin{lem}\label{supmat:lem:perm to logical}
Let $r\geq 3$ and $C(r)$ be the $(n_r,r,d_r)$-simplex code with generator matrix $G$. Furthermore, assume that $\sigma_1,\sigma_2 \in \aut{C}$ with corresponding linear transformations $g_{\sigma_1},g_{\sig_2} \in GL_r(2)$ such that $g_{\sigma_i} G  = G \sig_i$.
    Then $\sig_1 \otimes \sig_2 \in \aut{SHYPS(r)}$ induces the following action on the basis of logical operators
    \[
    L_Z \cdot (\sigma_1 \otimes \sigma_2) = (g_{\sig_1} \otimes g_{\sig_2}^{-T})\cdot L_Z,
    \]
    \[
    L_X \cdot (\sigma_1 \otimes \sigma_2) = (g_{\sig_1}^{-T} \otimes g_{\sig_2})\cdot L_X.
    \]
\begin{proof}
    For ease of presentation, let's consider the action of $\sigma \otimes I$ where $\sig \in \aut{C}$ (the general case follows identically). Firstly observe that 
    \[
    L_Z (\sig \otimes I) = (G\sigma \otimes P) = (g_{\sig}G \otimes P) = (g_{\sig} \otimes I)L_Z.
    \]
    The permutation action on the basis of $X$-logicals is then given by
    \[
    L_X(\sig \otimes I) = (P\sigma \otimes G),
    \]
    where 
    \begin{eqnarray*}
        I_{r^2} = L_X\cdot L_Z^T & = & L_X (\sig \otimes I) \cdot (\sig^T \otimes I)L_Z^T \\
        & = & (P\sigma \otimes G)\cdot (L_Z(\sigma \otimes I))^T \\
        & = & (P\sigma \otimes G)\cdot (G^T \otimes P^T) \cdot (g_{\sigma}^T\otimes I) .
    \end{eqnarray*}
Collecting terms and noting that $A \otimes B = I_{r^2}$ if and only if both components are identity (recall also that $PG^T = I_r$), it follows that $P\sigma G^T = g_{\sig}^{-T}.$ Now $H$ spans $\ker G$ and hence all solutions to the above are of the form 
\[
P\sig = g_{\sig}^{-T}P+AH,
\]
for some $A\in M_{n_r}$.
In summary, there exists $A$ such that
\begin{eqnarray*}
    L_X(\sig \otimes I) &=& P\sigma \otimes G \\
    & = & (g_{\sig}^{-T}P+AH) \otimes G \\
    & = & (g_{\sig}^{-T} \otimes I)(P\otimes G) + (A\otimes I_r)(H \otimes G) \\
    & = & (g_{\sig}^{-T} \otimes I)L_X + (A\otimes I_r)S_X.
\end{eqnarray*}
That is, up to stabilizers (the exact stabilizer determined by $A$), $\sigma \otimes I$ has the desired action on $L_X$.
\end{proof}
    
\end{lem}
As a logical Clifford operator is determined (up to a phase) by its action on the basis of logical Paulis, Lemma \ref{supmat:lem:perm to logical} demonstrates that the qubit permutation $\sig_1\otimes \sig_2$ induces the logical CNOT circuit corresponding to $g_{\sig_1}^{-T} \otimes g_{\sig_2} \in GL_{r^2}(2)$.

\section{CNOT operators in SHYPS codes}\label{supmat:sec: cnot gate characterisation}

In this section, the collection of automorphisms 
\[
\{\sig_1 \otimes \sig_2 \mid \sig_i \in \aut{C(r)} \}\subset\aut{SHYPS(r)}
\]
serves as the foundation for generating all logical CNOT operators in the SHYPS codes. 

First, following \cite{Grassl2013}, we show how all logical cross-block CNOT operators (with all controls in a first code block, and all targets in a second code block) may be attained by sequences of physical depth-1 circuits, that interleave the transversal CNOT operator with code automorphisms. In-block CNOT operators are then achieved by use of an auxiliary code block (see \cref{supmat:subsec:full cnot group gen}). To demonstrate efficient compilation, we develop substantial linear algebra machinery for certain matrix decomposition problems -- we expect these methods to be useful for logical gate compilation in many other code families.

A summary of depth bounds for a range of logical CNOT operators of SHYPS codes is available in \cref{supmat:tab:cnot_depth}.

\textbf{Notation:} We adopt the notation discussed in Example~\ref{supmat:ex:symplectic-rep-cnots}, where CNOT circuits on $b$ blocks of $k=r^2$ logical qubits or $n=n_r^2$ physical qubits, are described by invertible matrices in $GL_{br^2}(2)$ or $GL_{bn_r^2}(2)$, respectively. The collection of all (not necessarily invertible) $n\times n$ binary matrices is denoted by $M_n(2)$ and $E_{i,j}\in M_n(2)$ denotes the matrix of all zeros except the $(i,j)$ entry equal to 1. The collection of diagonal matrices are denoted $\text{Diag}_n(2) \subset M_n(2)$.

\begin{lem} \label{supmat:lem:depth 1 cnot circuits}
    Let $g_1,g_2 \in GL_{r}(2)$. Then the logical CNOT circuits $\left( \begin{matrix}
       I & g_1\otimes g_2 \\ 0 & I 
    \end{matrix} \right), 
    \left( 
    \begin{matrix}
       I & 0 \\ g_1\otimes g_2 & I 
    \end{matrix} 
    \right)
    \in GL_{2r^2}(2)$ on code $SHYPS(r)$ may be implemented by a depth-1 physical CNOT circuit.
\end{lem}
\begin{proof}
    First recall that (like all CSS codes) the transversal $CNOT^{\otimes n_r^2}$ implements $\overline{CNOT}^{\otimes r^2}$ between two code blocks of $SHYPS(r)$. Independently, it follows from Lemmas \ref{supmat:lem:aut group of simplex code} and \ref{supmat:lem:perm to logical} that there exist $\sig_i \in S_{n_r}$, such that the physical CNOT circuit given by $\mat{I & 0 \\ 0 & \sig_1\otimes \sig_2} \in GL_{2n_r^2}(2)$ implements $\overline{\mat{I & 0 \\ 0 & g_1\otimes g_2}}\in GL_{2r^2}(2).$ Hence
    \[
    \mat{I & 0 \\ 0 & \sig_1^{-1}\otimes \sig_2^{-1}} \cdot \mat{I & I \\ 0 & I} \cdot \mat{I & 0 \\ 0 & \sig_1\otimes \sig_2}  = \mat{I & \sig_1\otimes \sig_2 \\ 0  & I}, 
    \]
    implements
    \[
     \overline{\mat{I & 0 \\ 0 & g_1^{-1}\otimes g_2^{-1}}} \cdot \overline{\mat{I & I \\ 0 & I}} \cdot \overline{\mat{I & 0 \\ 0 & g_1\otimes g_2}}  = \overline{\mat{I & g_1\otimes g_2 \\ 0 & I}}. 
    \]
    But recall that $\pi = \sig_1\otimes \sig_2$ is a permutation matrix, and in particular has row and column weights equal to one. Hence
    \[
 \mat{I & \pi \\ 0  & I}
    \]
    represents the depth-1 physical CNOT circuit $\prod_i CNOT_{i,\pi^{-1}(i)+n_r^2}$. The circuit diagram illustrating this conjugated CNOT operator is given by \cref{supmat:fig:generalised-transversal-cnot}.
    
    \begin{figure}
\begin{tikzpicture}
\begin{yquantgroup}
  \registers{
         qubit {} q[2];
      }
   \circuit{
      slash q;
      hspace {3mm} q;
      box {$\pi$} q[1] | q[0];
      }
   \equals
   \circuit{
      slash q;
      hspace {3mm} q;
      box {$\pi^{-1}$} q[1];
      cnot q[1] | q[0];
      box {$\pi$} q[1];
   } 
\end{yquantgroup}
\end{tikzpicture}
        \caption{Physical implementation of a cross-block CNOT operator utilising $\pi \in \aut{\calC}$.}
        \label{supmat:fig:generalised-transversal-cnot}
    \end{figure}

The transposed logical CNOT circuit follows identically by exchanging the operations on the two code blocks.
\end{proof}

We now establish some notation for our generators, and drop the overline notation on the understanding that all operators are logical operators unless specified otherwise.

\begin{defn}
    Denote the above collection of logical CNOT operators induced from classical automorphisms by
    \begin{align}
    \mathcal{A}&:=\{ \mat{I & g_1\otimes g_2 \\ 0 & I} \,:\,g_i \in GL_r(2)\} \text{ and }\\ \mathcal{A}^T&:=\{ \mat{I & 0 \\ g_1\otimes g_2 & I} \,:\,g_i \in GL_r(2)\}.
    \end{align}
\end{defn}

\subsection{Generating cross-block CNOT operators }\label{supmat:subsec:full cnot group gen}
The CNOT operators $\calA$ and $\calA^T$ are natural generalisations of ${CNOT}^{\otimes n_r^2}$ in that they have a transversal implementation on a pair of code blocks, and are thus inherently fault-tolerant. It is therefore highly desirable to use these circuits as generators for a larger class of CNOT operators. In this section we will see that they in fact suffice to efficiently generate all CNOT circuits between an arbitrary number of code blocks; a result that is highly unexpected in the context of generic quantum codes, and is derived from the particularly symmetric nature of the classical simplex codes.

Before proceeding, we state our first main result, which concerns cross-block CNOT operators in SHYPS codes

\begin{thm}\label{supmat:thm:cross-block-cnot-operators}
    Let  $A \in M_{r^2}(2)$. Then arbitrary cross-block CNOT operators given by \[\left( \begin{matrix}
       I & A \\ 0 & I 
    \end{matrix} \right), 
    \left( 
    \begin{matrix}
       I & 0 \\ A & I 
    \end{matrix} 
    \right)
    \in M_{2r^2}(2)\] are implemented fault-tolerantly in $SHYPS(r)$ using at most $r^2+r+4$ generators from $\calA$ and $\calA^T$.
\end{thm}

In particular, the implementation of cross-block CNOTs scales with the number of logical qubits $k=r^2$. In fact, by considering the size of the groups involved we see that this scaling is optimal up to a constant factor: as $\vert M_{r^2}(2) \vert =2^{r^4}$ and $\vert \calA \vert = \vert GL_r(2) \vert^2 = O(2^{2r^2})$,
    \[
    M_{r^2}(2) = \calA^l \implies \vert M_{r^2}(2) \vert \leq \vert \calA \vert^l,
    \]
    and hence $l \geq r^2/2.$
    
 \cref{supmat:thm:cross-block-cnot-operators} is the foundation for much of the eficient, fault-tolerant computation in the SHYPS codes; it is from these operators that we build arbitrary CNOT operators, as well multi-block diagonal Clifford gates (in conjunction with a fold-transversal Hadmard gate -- see \cref{supmat:sec: Hadamard}).

 To prove \cref{supmat:thm:cross-block-cnot-operators}, first observe that the composition of cross-block CNOT operators in $\lrang{\calA}$ (and similarly in $\lrang{\calA^T}$) behaves like addition of the off-diagonal blocks
 \[
 \mat{I & g_1\otimes g_2 \\ 0 & I}\cdot \mat{I & h_1\otimes h_2 \\ 0 & I} = \mat{I & g_1\otimes g_2 + h_1\otimes h_2 \\ 0 & I}. 
 \]
Hence our Clifford generation problem becomes a question of showing that tensor products of the form $g_1 \otimes g_2$ efficiently generate the full matrix algebra, under addition. In summary, \cref{supmat:thm:cross-block-cnot-operators} is an immediate corollary of the following result

\begin{thm} \label{supmat:thm:worst-case-cross-block-cnot}
    Let $A\in M_{r^2}(2)$. Then there exist $w$ pairs of invertible matrices $g_{i_1},g_{i_2} \in GL_r(2)$, for some $w\leq r^2+r+4$ such that $A = \sum _{i_1,i_2}^w g_{i_1}\otimes g_{i_2}$.
\end{thm}

The remainder of this section is devoted to the proof of \cref{supmat:thm:worst-case-cross-block-cnot} and so it is useful to begin with some remarks on our approach: First observe that the main challenge is that the matrices $g_{i_1}$ etc. must be invertible. We first demonstrate that a decomposition of at most $r^2$ terms is easy without this invertibility condition, and then proceed to adapt this to invertible tensor products, while incurring minimal additional overhead. To this end we develop a number of matrix decomposition results that consider spanning sets in $GL_r(2)$, as well as the effect of adding specific invertible matrices such as permutations.

To track these results, we introduce the following \emph{weight function}.
\begin{defn}
    Let $A \in M_{r^2}(2)$. The \emph{weight} of $A$, denoted $w(A)$, is the minimal number of pairs $(g_{i_1}, g_{i_2}) \in GL_r(2)^2$ such that $A=\sum_{i=1}^{w(A)} g_{i_1} \otimes g_{i_2}$.
\end{defn}

First let's decompose our arbitrary matrix into a non-invertible tensor product
\begin{lem}\label{supmat:lem:non-invertible-tensor-product}
Let $A \in M_{r^2}(2)$. Then there exist $t\leq r^2$ matrix pairs $M_{i},N_{i} \in M_r(2)$ such that $A = \sum M_i \otimes N_i$.
\end{lem}
\begin{proof}
    
Let $A_{i,j} \in M_{r}(2)$ be the $(i,j)$th block of $A$. Then 
\begin{equation} 
    A = \sum_{i,j = 1}^r E_{i,j} \otimes A_{i,j} \label{supmat:eq:tensor-rank-decomposition}
\end{equation}
is a decomposition of the required form.
\end{proof}

Note that the minimal $t$ required can often be much lower than $r^2$, and such a minimal decomposition is typically referred to as the \emph{tensor rank decomposition}.

The task now is to convert an expression of the form (\ref{supmat:eq:tensor-rank-decomposition}) into a similar expression comprising only invertible matrices.

\begin{lem}\label{supmat:lem:sum of two}
    Let $M \in M_k(2)$, ($k\geq 2$) then $M$ is a sum of at most two elements in $GL_{k}(2)$.
\end{lem}
 \begin{proof}
     First let's check that the claim holds for the identity matrix. 
     \[
         I_2 = \left( \begin{matrix}
        1 & 1 \\ 1&0
    \end{matrix}\right) + \left( \begin{matrix}
        0 & 1 \\ 1&1
    \end{matrix}\right) \text{ and }
    I_3 = \left( \begin{matrix}
        1 & 1 & 1 \\ 0 & 1 & 1 \\ 1&0 & 1
    \end{matrix}\right) + \left( \begin{matrix}
        0 & 1 & 1 \\ 0&0&1 \\ 1&0&0
    \end{matrix}\right)
     \]
     are appropriate decompositions for $k=2,3$. Then clearly larger $I_k$ can be decomposed as blocks of 2 or 3 and treated as block sums of the above.

     Now for arbitrary $M \in M_k(2)$ of rank $l\leq k$, there exist invertible matrices $g,g' \in GL_k(2)$ corresponding to row and column operations respectively, such that 
     \[
     gMg' = \left( \begin{matrix}
         I_l & \\ & 0_{k-l}
     \end{matrix} \right),
     \]
     and zeros elsewhere via Gaussian elimination. But we know that there exist $X,Y \in GL_l(2)$ such that $X+Y = I_l$ and hence
     \[
     M = (g)^{-1}\left( \begin{matrix}
         X & \\ & I_{k-l}
     \end{matrix} \right)(g')^{-1} + (g)^{-1}\left( \begin{matrix}
         Y & \\ & I_{k-l}
     \end{matrix} \right)(g')^{-1}
     \]
     is a decomposition as a sum of two matrices in $GL_k(2).$ NB: if $M$ has rank 1 then simply take 
     \[ \left( \begin{matrix}
         1 & 0 \\ 0 & 0
     \end{matrix} \right) =
     \left( \begin{matrix}
        1 & 1 \\ 1&0
    \end{matrix}\right) + \left( \begin{matrix}
        0 & 1 \\ 1&0
    \end{matrix}\right)
     \]
     and proceed similarly.
 \end{proof}

 As the tensor product is distributive over addition, each non-invertible matrix may be split in two using \cref{supmat:lem:sum of two} to yield 
 
\begin{cor}
    Let $A\in M_{r^2}(2)$. Then $w(A) \leq 4r^2.$
\end{cor}

    Although this establishes a bound $w(A) \in O(r^2)$, we'll see that the constant factor 4 can be greatly improved.
    
\begin{prop}\label{supmat:prop:matrix-weight}
    Let $A \in M_{r^2}(2)$ have a tensor decomposition of rank $t$, i.e., there exists $M_i,N_i \in M_r(2)$ such that $A=\sum_{i=1}^tM_i \otimes N_i$. Then 
    \[
    w(A) \leq \min(4t,2t+8,t+r+6,r^2+r+4).
    \]
\end{prop}

The first step in proving \cref{supmat:prop:matrix-weight} is to generalise \cref{supmat:lem:sum of two} to vector spaces and spanning sets of invertible matrices. This also requires an understanding of the proportion of binary matrices that are invertible.

\begin{lem} \label{supmat:lem:proportion of invertible matrices}
    Let $r\geq 1$. Then
    \[
    |GL_r(2)|/|M_r(2)|>\frac{1}{4}.
    \]
\end{lem}
\begin{proof}
    First observe from standard formulae that
    \[
    |GL_r(2)|/|M_r(2)| = \prod_{i=1}^r (1-2^{-i}),
    \]
     and hence the Lemma clearly holds for $r=1$.
    For $r\geq 2$, we'll prove the slightly stronger statement
    \[
    |GL_r(2)|/|M_r(2)|>\frac{1}{4}+{2^{-(r+3/2)}},
    \]
    via induction. The statement is easily checked for $r=2$ and assuming the result holds for some $l\geq 2$:
    \begin{eqnarray*}
        |GL_{l+1}(2)|/|M_{l+1}(2)| & = & |GL_{l}(2)|/|M_{l}(2)|\cdot (1-2^{-(l+1)}) \\
        & > & (\frac{1}{4}+{2^{-(l+3/2)}})\cdot (1-2^{-(l+1)}) \\
        & = & \frac{1}{4}+2^{-(l+3/2)} - 2^{-(l+3)} - 2^{-(2l+5/2)} \\
        & = & [\frac{1}{4}+2^{-(l+1+3/2)}] \\ & + & [2^{-(l+1+3/2)}- 2^{-(l+3)} - 2^{-(2l+5/2)}].
    \end{eqnarray*}
        So the result follows for $l+1$, provided 
        \[
        2^{-(l+5/2)}- 2^{-(l+3)} - 2^{-(2l+5/2)} > 0.
        \]
        But this holds if and only if
        \[
        2^l-2^{l-1/2}-1=2^{l-1}(2-\sqrt{2})-1>0,
        \]
        which is indeed true for $l\geq 2$.
\end{proof}

\begin{lem}\label{supmat:lem:invertible-spanning-set1}
    Let $V\leq M_r(2)$ be a vector subspace of dimension $d$. Then there exists $T \subset GL_r(2)$ such that $\vert T \vert \leq \min (2d,r^2,d+2)$ and $V \leq \langle T\rangle$.
\end{lem}
	\begin{proof}

     The first two entries in the bound $\min (2d,r^2,d+2)$ follow immediately from \cref{supmat:lem:sum of two} and the fact that $\lrang{GL_r(2)}=M_r(2)$ which has dimension $r^2$. It therefore remains to show the final $d+2$ bound.
     Well if $V$ contains an invertible matrix $g_1$ then $V =\langle g_1, V' \rangle$ for a vector space $V'$ of strictly smaller dimension. Hence we assume without loss of generality that $V\leq M_r(2)\backslash GL_r(2)$ consists solely of non-invertible matrices. Let's denote a basis of $V$ by $A_1,\dots,A_d$.
     
     As the cosets $M_r(2)/V$ tile the space $M_r(2)$, by Lemma \ref{supmat:lem:proportion of invertible matrices} there exists $M\in M_r(2)$ such that the proportion of invertible matrices in the coset $M+V$ is greater than $1/4$. I.e., if we denote these invertible elements by $G\subseteq M+V$, then $\vert G \vert > 2^{d-2}$. Now as each element of $G$ has the form $M+\sum_i \alpha_iA_i$, it follows that even-weight linear combinations of elements in $G$ are non-invertible. Hence
     \[
     \vert \langle G \rangle \vert \geq 2\vert G \vert >2^{d-1},
     \]
     and thus $\langle G \rangle$ has dimension at at least $d$.
     So there exist linearly independent $g_1,\dots,g_d \in G$ and the subspace of even weight linear combinations $V' \leq \langle g_1,\dots,g_d \rangle$ is a $d-1$ dimensional subspace $V'< V$. Finally, taking any $A \in V\backslash V'$ and applying Lemma \ref{supmat:lem:sum of two} to yield a sum of two invertibles $A=g_{d+1}+g_{d+2}$, it follows that
     \[
     V=\langle V',A \rangle \leq \langle g_1,\dots,g_{d+2} \rangle.
     \]
     \end{proof}

     Next we prove some useful Lemmas that study the effect of adding diagonal matrices and permutation matrices.
\begin{lem}
    \label{supmat:lem:diaginvert}
    For all $A\in M_r(2)$ there exists a diagonal $D\in M_r(2)$ such that $A+D \in GL_r(2)$.
\end{lem}
\begin{proof}
    We proceed by induction on $r$. The case $r=1$ is clear, and suppose the Lemma holds holds for $l$. Then for $A \in M_{l+1}(2)$ we may write
    \[
        A = \left(\begin{array}{c|c}
            A' & v \\
            \hline
            u^T & a
        \end{array}\right)
    \]
    for some $A'\in M_{l}(2)$, $u,v\in\gf^l$, and $a\in\gf$. By the induction hypothesis there exists diagonal $D'\in M_l(2)$ such that $A'+D' \in GL_l(2)$. Then for diagonal
    \[
        D = \left(\begin{array}{c|c}
            D' & 0 \\
            \hline
            0 & a + u^T (A' + D')^{-1} v + 1
        \end{array}\right)
    \]
    we have
    \begin{eqnarray*}
        A + D &= &\left(\begin{array}{c|c}
            A' + D' & v \\
            \hline
            u^T & u^T (A' + D')^{-1} v + 1
        \end{array}\right) \\
        &= &\left(\begin{array}{c|c}
            I & 0 \\
            \hline
            u^T (A' + D')^{-1} & 1
        \end{array}\right)\cdot\left(\begin{array}{c|c}
            A'+D' & 0 \\
            \hline
            0 & 1
        \end{array}\right) \\ & \cdot &\left(\begin{array}{c|c}
            I & (A' + D')^{-1}v \\
            \hline
            0 & 1
        \end{array}\right).
    \end{eqnarray*}
    As each matrix in the product decomposition above is invertible, $A+D \in GL_{l+1}(2)$.
\end{proof}

\begin{cor}
    \label{supmat:cor:perminvert}
    Let $A\in M_r(2)$ and $P \in GL_r(2)$ be a permutation matrix. Then there exists $O\in M_r(2)$ with nonzero entries supported on the nonzero entries of $P$ such that $A+O \in GL_r(2)$.
\end{cor}
\begin{proof}
    Apply \cref{supmat:lem:diaginvert} to $P ^{-1}A$ to find a diagonal $D$ such that $P^{-1}A + D$ is invertible and then define $O = P D$.
\end{proof}

\begin{lem}
    \label{supmat:lem:diagperm}
    Let $D\in M_r(2)$ be a diagonal matrix. Then either $D\in GL_r(2)$, or $ D +P\in GL_r(2)$ for all $r$-cycle permutations $P$.
\end{lem}
\begin{proof}
    Clearly $D$ is invertible if and only if $D = I$, so we restrict to the case where $D\neq I$. First note that as $P$ is an $r$-cycle, the full set of indices $1,\dots, r$ is contained in the single orbit of $P$. I.e., given any starting point $i$, the list $i, ~P(i), ~P^2(i), \dots, ~P^{r-1}(i)$, is a re-ordering of $1,\dots, r$.
    Consequently, we may index columns of $D+P$ by the numbers $P^l(i)$.
    Now since $D\neq I$, there exists some $i$ such that the $i$-th diagonal entry of $D$, is zero. Choosing this as the starting point of our column indexing, we see that the columns of $D+P$ are given by
    \[
        (D e_i^T + e^T_{P(i)} = e^T_{P(i)}, De^T_{P(i)} + e^T_{P^2(i)}, \dots, D e^T_{P^{r-1}(i)} + e^T_{P^{r}(i)}).
    \]
    By induction, we show the span of the first $m$ elements of this sequence is given by
    \[
        \langle e^T_{P(i)},e^T_{P^2(i)}, \dots, e^T_{P^m(i)} \rangle.
    \]
    The base case is trivial by assumption. For the inductive step, observe that the $l+1$ entry of the sequence is either $e^T_{P^{l+1}(i)}$ or $e^T_{P^{l}(i)} + e^T_{P^{l+1}(i)}$, neither of which are in $\langle e^T_{P(i)},e^T_{P^2(i)}, \dots, e^T_{P^l(i)}\rangle$ since $P$ is an $r$-cycle, and either of which when added to the generating set yields a vector space given by $\langle e^T_{P(i)},e^T_{P^2(i)}, \dots, e^T_{P^{l+1}(i)}\rangle$. Thus the span of the columns is
    \[
        \lrang{e^T_{P(i)},e^T_{P^2(i)}, \dots, e^T_{P^r(i)}} = \lrang{e^T_{1},e^T_{2}, \dots, e^T_{r}} = \gf^r.
    \]
    As the columns of $D+P$ span $\gf^r$, $D+P$ is invertible as desired.
\end{proof}

\begin{cor}
    \label{supmat:cor:permperm}
    Let $P \in GL_r(2)$ be a permutation, and $D \in M_r(2)$ be diagonal. Then either $PD \in GL_r(2)$ or $PD + PQ \in GL_r(2)$ for all $r$-cycle permutations $Q\in GL_r(2)$.
\end{cor}
\begin{proof}
    Apply \cref{supmat:lem:diagperm} to $D$, then multiply by $P$, which preserves invertibility.
\end{proof}

With these results in place, we are ready to move onto the proof of \cref{supmat:prop:matrix-weight}
\begin{proof} (Proof of \cref{supmat:prop:matrix-weight}).
    By assumption, we can write
    \[
        A = \sum_{i=1}^t M_i \otimes N_i
    \]
    for some $M_i, N_i\in M_r(2)$. By \cref{supmat:lem:invertible-spanning-set1}, we can find a set of $q$ elements $B \subseteq GL_r(2)$ that generate all $M_i$, for $q \leq \min(2t, r^2,t+2)$. Expressing each $M_i$ as a linear combination of elements of $B$, we have
    \[
        A = \sum_{i=1}^t\left(\sum_{j=1}^q\mu_{ji}B_j \right)\otimes N_i = \sum_{j=1}^q B_j\otimes\left(\sum_{i=1}^t\mu_{ji} N_i\right).
    \]
    By \cref{supmat:lem:diaginvert}, there exists some $C_j\in GL_r(2)$ invertible and some $D_j$ diagonal so that
    \[
        C_j + D_j = \sum_{i=1}^t\mu_{ji} N_i
    \]
    for all $1\leq j \leq q$. The vector space spanned by the set $\{D_j\}_{j=1}^q$ has dimension $p\leq \min(q,r)$, and after expanding each $D_j$ in a basis $E$ for $\lrang{D_j}$, we can rewrite $A$ as
    \begin{align*}
        A &= \sum_{j=1}^q B_j\otimes C_j + \sum_{j=1}^q B_j\otimes D_j\\
        &= \sum_{j=1}^q B_j\otimes C_j +  \sum_{j=1}^q B_j\otimes \left(\sum_{\ell=1}^p \delta_{\ell j}E_\ell\right)\\
        &= \sum_{j=1}^q B_j\otimes C_j + \sum_{\ell=1}^p \left(\sum_{j=1}^q \delta_{\ell j} B_j\right) \otimes E_\ell.
    \end{align*}
    Again, by \cref{supmat:lem:invertible-spanning-set1} we can compute a set of $m$ elements $F$ from $ GL_r(2)$ that contains $\lrang{\sum_{j=1}^q \delta_{\ell j} B_j~:~1\leq \ell\leq p}$. Expressing elements from the span as $F$-linear combinations, we have
    \begin{align*}
        A &= \sum_{j=1}^q B_j\otimes C_j + \sum_{\ell=1}^p \left(\sum_{a=1}^m \beta_{a\ell} F_a\right) \otimes E_\ell\\
        &= \sum_{j=1}^q B_j\otimes C_j + \sum_{a=1}^m  F_a\otimes \left(\sum_{\ell=1}^p \beta_{a\ell} E_\ell\right).
    \end{align*}
    Next, by \cref{supmat:lem:diagperm}, every $\sum_{\ell=1}^p \beta_{a\ell} E_\ell$ is such that adding $\rho_a P$ for some $\rho_a\in\gf$ and $P$ a fixed $r$-cycle makes the matrix invertible, and so we have
    \begin{eqnarray*}
        A &= & \sum_{j=1}^q B_j\otimes C_j + \sum_{a=1}^m  F_a\otimes \left(\rho_a P + \sum_{\ell=1}^p \beta_{a\ell} E_\ell\right) \\ &+& \sum_{a=1}^m  F_a\otimes(\rho_a P)\\
        &= &\sum_{j=1}^q B_j\otimes C_j + \sum_{a=1}^m  F_a\otimes \left(\rho_a P +\sum_{\ell=1}^p \beta_{a\ell} E_\ell\right) \\ & + & \left(\sum_{a=1}^m  \rho_a F_a\right)\otimes P
    \end{eqnarray*}
    Finally, by \cref{supmat:lem:sum of two} we know that $\sum_{a=1}^m  \rho_a F_a$ is the sum of at most two invertibles. This permits us to write for $G_1,G_2\in\{0\}\cup GL_r(2)$
    \begin{eqnarray*}
        A &= & \sum_{j=1}^q B_j\otimes C_j + \sum_{a=1}^m  F_a\otimes \left(\rho_a P +\sum_{\ell=1}^p \beta_{a\ell} E_\ell\right) \\
        &+& G_1\otimes P + G_2\otimes P
    \end{eqnarray*}
    where by construction every lone or bracketed term is an element of $ GL_r(2)$. Therefore, we conclude $w(A)\leq q + m + 2$, where we recall
    \begin{align*}
        q \leq \min(2t,r^2,t+2) \\
        m \leq \min(t+2,r)+2.
    \end{align*}
    
    There are now four relevant regimes to bound $w(A)$:

    $\mathbf{(t \leq 3)}$: rather than follow the process above, each $M_i,N_i$ can be decomposed as a sum of two invertible matrices using \cref{supmat:lem:sum of two} to yield $w(A)\leq 4t$.
    
    $\mathbf{(3 < t \leq r-2)}$: here we take $q$ to be bounded by $t+2$ and $m$ bounded by $t+4$ to give $w(A)\leq 2t+8$.
    
    $\mathbf{(r-2 < t \leq r^2 - 2)}$: here we take $q$ to be bounded by $t+2$ and $m$ bounded by $r+2$ to give $w(A)\leq t+r+6$.
    
    $\mathbf{(t > r^2 - 2)}$: finally here we take $q$ bounded by $r^2$ and $m\leq r+2$ to yield $w(A)\leq r^2+r+4$.
\end{proof}
Summarising the worst case of \cref{supmat:prop:matrix-weight}, \cref{supmat:thm:worst-case-cross-block-cnot} is now immediate, and we repeat it below for convenience
\begin{thm} 
    Let $A\in M_{r^2}(2)$. Then $w(A)\leq r^2 + r + 4$.
\end{thm}
For applications in decomposing arbitary Clifford operators it is useful to consider the more specific case of upper-triangular matrices:
\begin{lem}
Let $A \in M_{r^2}(2)$ be any invertible upper-triangular matrix. Then $w(A) \leq r(r+1)/2+6$. 
\end{lem}
\begin{proof}
    Following similarly to the proof of \cref{supmat:prop:matrix-weight}, we can always write any such $A$ as the sum
    \[
        A = \sum_{i = 1}^{\frac{r(r-1)}{2}}S_i\otimes N_i + \sum_{i = 1}^{r}E_{i,i}\otimes U_i
    \]
    for $S_i$ strictly upper triangular matrices, $N_i$ arbitrary elements of $ M_r(2)$, $E_{i,i}$ the usual weight one diagonal matrices, and $U_i$ an invertible upper triangular matrix. Computing a generating set $G$ of invertible matrices for the space spanned by $N_i$ with at most $q\leq\frac{r(r-1)}{2}+2$ elements and collecting terms, we have
    \begin{eqnarray*}
        A &=& \sum_{j = 1}^{q}\left(\sum_{i=1}^{\frac{r(r-1)}{2}}\nu_{ji} S_i\right)\otimes G_j + \sum_{i = 1}^{r}E_{i,i}\otimes U_i\\
        &=& I\otimes\left(\sum_{j = 1}^{q} G_j\right) + \sum_{j = 1}^{q}\left(I + \sum_{i=1}^{\frac{r(r+1)}{2}}\nu_{ji} S_i\right)\otimes G_j \\ &+& \sum_{i = 1}^{r}E_{i,i}\otimes U_i.
    \end{eqnarray*}
    Applying \cref{supmat:lem:sum of two} to the first term, and noting that each $I + \sum_{i=1}^{\frac{r(r+1)}{2}}\nu_{ji} S_i \in GL_r(2)$ yields
    \[
    w(A)\leq q+2+w(\sum_{i = 1}^{r}E_{i,i}\otimes U_i).
    \]
    But then $E_{i,i}+C$ is invertible for any $r$-cycle $C$ and hence the  final term is bounded by
    \[
\sum_{i=1}^rw(E_{i,i}+C,U_i)+w(C,\sum_{i=1}^r U_i) \leq r+2.
    \]
    The result now follows.
\end{proof}

\subsection{Arbitrary CNOT operators}\label{supmat:subsec:arbitrary-cnot-operators}
\cref{supmat:thm:cross-block-cnot-operators} establishes efficient fault-tolerant implementations of cross-block CNOT circuits. It thus remains to show how our lifted automorphisms may implement CNOT circuits within a single code block. We remark that this does follow from \cite[Thm. 4]{Grassl2013}, however by relying not on generating transvections, but rather an auxiliary-block based trick, we minimise additional overhead.

\begin{lem} \label{supmat:lem:in-block CNOT}
    Any logical CNOT circuit on a code block of $SHYPS(r)$ (i.e., any element of $GL_{r^2}(2)$) may be generated by $r^2+r+4$ depth-1 CNOT circuits from ${ \calA \cup \calA^T}$, using a single auxiliary code block.
\end{lem}
\begin{proof}
    Take an arbitrary $C\in GL_{r^2}(2)$. The CNOT circuit $C$
     can be executed using a scheme based on the well-known quantum teleportation circuit:
    \[
        \begin{tikzpicture}
        \begin{yquantgroup}
        \registers{
             qubit {$\ket\psi$} a;
             qubit {$\ket{0}$} b;
        }
        \circuit{
            box {$C^{-1}$} b | a;
            [direct control]
            measure a;
            z b | a;
            discard a;
            output  {$C\ket\psi$} b;
        }
        \end{yquantgroup}
        \end{tikzpicture}
    \]
    So at the cost of introducing an additional auxiliary code block in the $\ket{0}$ state (which may be prepared offline in constant depth) and teleporting our state $\ket\psi$, $C$ is applied within the code block for the cost of a single element from $\lrang{\calA \cup \calA^T}$. The lemma then follows directly from \cref{supmat:thm:cross-block-cnot-operators}.

\end{proof}

Note that in the remainder of this section, we consider arbitrary CNOT operators, modulo a possible permutation of the logical qubits. These will be accounted for later when compiling an arbitrary Clifford operator in \cref{supmat:sec:SHSCompSummary}, and in fact SWAP circuits may be implemented more efficiently than generic CNOT operators (see \cref{supmat:sec: Hadamard} and \cref{supmat:tab:had_swap_depth}). In particular, logical permutations on $b$ code blocks are implementable in $O(r^2)$ depth, a cost that crucially doesn't scale with the number of code blocks (see \cref{supmat:thm:cross-block-permutations}).

    \begin{table*}[h!]
    \centering
    \setlength\tabcolsep{5mm}
    {\renewcommand{\arraystretch}{1.2}

    \begin{tabular}{c | c}
        \textbf{Logical Gate} & \textbf{Depth bound} \\ \hline 
        In-block CNOT operator $g_1 \otimes g_2$ for $g_1, g_2 \in GL_r(2)$ & 0 (qubit relabelling) \\
        Cross-block CNOT operator $\begin{pmatrix}
            I & g_1 \otimes g_2 \\ 0 & I
        \end{pmatrix} \in GL_{2r^2}(2)$ for $g_1, g_2 \in GL_r(2)$ & 1 \\
        Arbitrary cross-block CNOT operator $\begin{pmatrix}
            I & M \\ 0 & I
        \end{pmatrix}  \in GL_{2r^2}(2)$ for $M \in M_{r^2}(2)$ & $r^2+r+4$ \\
        Depth-1 cross-block CNOT operator on two code blocks & $7r+15$ \\
        Arbitrary in-block CNOT operator $A \in GL_{r^2}(2)$ & $r^2+r+4$\\
        $CNOT_{i,j}$ for qubits $i,j$ in distinct or non-distinct code blocks & $4$ \\
        Arbitrary $2$-block CNOT operator (modulo logical permutation) & $3(r^2+r+4)$ \\
        Arbitrary $b$-block upper-triangular CNOT operator & $(b-1)(r^2+r+4) + r(r + 1)/2 + 6$ \\
        Arbitrary $b$-block CNOT operator (modulo logical permutation) & $(2b-1)(r^2+r+4)$ 
    \end{tabular}}
    \caption{Depth bounds for logical CNOT operators in $SHYPS(r)$. Recall that a CNOT operator on $b$ blocks of $r^2$ logical qubits is determined by matrices in $GL_{br^2}(2)$ - we give these matrices above where appropriate.}
    \label{supmat:tab:cnot_depth}
\end{table*}

\begin{lem} \label{supmat:lem: 2 codeblocks}
    Up to logical permutation, any logical CNOT circuit across two code blocks of $SHYPS(r)$ (i.e., any element of $GL_{2r^2}(2)$) may be generated by $4r^2+4r+16$ depth-1 CNOT circuits from ${ \calA \cup \calA^T}$, and executed in depth $3r^2+3r+12$, using at most 2 additional auxiliary code blocks.
\end{lem}
\begin{proof}
    Let $X \in GL_{2r^2}(2)$ and assume that $X$ has $PLU$ decomposition
    \[
    X = P \mat{C_L & 0 \\ B_L & C'_L} \mat{C_U & A_U \\ 0 & C'_U}.
    \]
    As all diagonal blocks in $L$ and $U$ are invertible, there exist $A'_L,B'_U\in M_{r^2}(2)$ such that
    \[
    X = P \mat{I & 0 \\ B'_L & I} \mat{C_L C_U & 0 \\ 0 & C'_L C'_U} \mat{I & A'_U \\ 0 & I}.
    \]
    Here $P$ is a cross block permutation that as stated in the Lemma, we need not consider (in practice it will be accounted for as part of a holistic Clifford decomposition, or simply tracked in software).
    Of the remaining three factors, the outer two matrices are clearly in $\lrang{\calA^T}$ and $\lrang{\calA}$, respectively. Hence it remains to implement 
    \begin{align}
    \mat{C_L C_U & 0 \\ 0 & C'_L C'_U},\notag
    \end{align}
    i.e., arbitrary invertible CNOT circuits $C = C_LC_U$ and $C' = C'_LC'_U$, within the individual code blocks.
    It follows from \cref{supmat:lem:in-block CNOT} that these can be implemented for the cost of a single element from $\lrang{\calA \cup \calA^T}$, each using a single auxiliary code block.

We conclude that $X$ is implemented by a circuit consisting of at most 4 logical cross-block CNOT circuits, and hence $4(r^2+r+4)$ depth-1 circuits from $\calA \cup \calA^T$. Furthermore, $C$ and $C'$ can be performed in parallel since they are applied to different code blocks. Hence, $X$ can be implemented in a depth no greater than $3(r^2+r+4)$.
\end{proof}

We conclude this section by generalising the above to an arbitrary number of blocks. For ease of presentation in the proof, and to avoid overly complicated compiling formulae we restrict our attention to $b=2^a$ blocks - the general case follows similarly. We state the following result in terms of rounds of cross-block CNOT circuits in $\lrang{\calA \cup \calA^T}$

\begin{lem}\label{supmat:lem:b blocks}
    Let $b=2^a$ and $X = GL_{br^2}(2)$ be an arbitary CNOT circuit across $b$ code blocks. Up to possible logical permutation, $X$ is implemented by $b(b+1)$ cross-block CNOT circuits in $\lrang{\calA \cup  \calA^T}$, performed over $2b - 1$ rounds.
\end{lem}
\begin{proof}
    We start again with a $PLU$ decomposition for $X$,
    \[
    X = P \mat{C_L & 0 \\ B_L & C'_L} \mat{C_U & A_U \\ 0 & C'_U},
    \]
    where each block in $L,U$ themselves consist of $b/2$ sub-blocks, and the permutation $P$ may be ignored.
    Isolating $U$, we have that
    \[
    U = \mat{C_U & A_U \\ 0 & C'_U} = \mat{C_U & 0 \\ 0 & C'_U} \mat{I & C^{-1}_U A_U \\ 0 & I}.
    \]
    But $C^{-1}A_U$ consists of $b/2 \times b/2$ blocks $M_{i,j}\in M_{r^2}(2)$ and hence
    \begin{align}
    \mat{I & C^{-1}_U A_U \\ 0 & I} &= \prod_{j=1}^{b/2} \mat{I & \sum_{i}^{b/2}E_{i,\sigma^j(i)}\otimes M_{i,\sigma^j(i)} \\ 0 & I}  \notag\\
    &= \prod_{i,j=1}^{b/2} \mat{I & E_{i,\sigma^j(i)}\otimes M_{i,\sigma^j(i)} \\ 0 & I},\notag
    \end{align}
    where $\sigma = (1,2,\dots,b/2)$ is the length $b/2$ cyclic shift of indices. For example if $b=4$, $\sig = (1,2)$ and
    \begin{eqnarray*}
            \mat{I & C^{-1}_U A_U \\ 0 & I} & = & \left[\mat{I & E_{1,1}\otimes M_{1,1} \\ 0 & I} \cdot \mat{I & E_{2,2}\otimes M_{2,2} \\ 0 & I} \right] \\
            &\cdot & \left[\mat{I & E_{1,2}\otimes M_{1,2} \\ 0 & I} \cdot \mat{I & E_{2,1}\otimes M_{2,1} \\ 0 & I} \right].
    \end{eqnarray*}

    Observe that each square bracketed term, corresponding to a fixed power $\sigma^j$, consists of $b/2$ block to block operators in $\lrang{\calA}$ that may be performed in parallel. Hence the off-diagonal matrix in the decomposition of $U$ requires $b^2/4$ operators in $\lrang{\calA}$, that with parallelisation can be achieved in $b/2$ rounds.
    
    To cost the diagonal term in $U$ we proceed by induction: assume that each (upper triangular) block $C_U,C'_U \in GL_{(b/2)r^2}(2)$ may be implemented by $f(b/2)$ generators in $\calA$, over $t(b/2)$ rounds (both are performed in parallel). Then
    \[
    t(b) = t(b/2) + b/2 \text{ and } f(b)=2f(b/2)+b^2/4.
    \]
    Solving these recurrence relations using initial conditions $t(2)=2$ and $f(2) = 3$ yields
    \[
    t(b) = b \text{ and } f(b)=\frac{b}{2}(b+1).
    \]
    Recall that these costings are for $U$ only, however the costings for $L$ are identical.
    Lastly, we note that the $b$ in-block CNOT operators appearing in $L$ can be combined with those in $U$ and executed together in a single round. 
    Thus given cross-block operators in $\lrang{\calA \cup \calA^T}$, any $b$-block CNOT circuit $X \in GL_{br^2}(2)$ can be implemented in $2b-1$ rounds, for a total $b(b+1)$ cross-block operators.
\end{proof}
\begin{cor}\label{supmat:cor:b-block-cnot}
    Let $b=2^a$ and $X = GL_{br^2}(2)$ be an arbitrary logical CNOT circuit across $b$ code blocks of $SHYPS(r)$. Then up to a possible logical permutation, $X$ is implementable in depth at most $(2b-1)(r^2+r+4)$, using depth-1 CNOT circuits in $ \calA \cup \calA^T$.
\end{cor}
\begin{proof}
    This is immediate from \cref{supmat:thm:cross-block-cnot-operators} and \cref{supmat:lem:b blocks}.
\end{proof}

Lastly, let's consider the case of single CNOT operators
\begin{lem}\label{supmat:lem:single-cnot}
    Any single logical CNOT operator between or within code blocks is implementable fault-tolerantly in depth at most 4. 
\end{lem}
\begin{proof}
    A single CNOT gate across code blocks that connects the $(i,j)$th and $(l,k)$th qubits corresponds to the matrix in $GL_{2r^2}(2)$ with off diagonal block $E_{i,k}\otimes E_{j,l}$. The result then follows by decomposing each term in the tensor product into a sum of two invertible matrices using \cref{supmat:lem:sum of two}. 
    
    For an in-block gate, we instead decompose the matrix $I+E_{a,b}\otimes E_{c,d} \in GL_{r^2}(2)$ where at least one of $a \neq b$, or $c \neq d$ holds, and then apply the scheme introduced in the proof of \cref{supmat:lem:in-block CNOT}.
    In particular, observe that
    \begin{eqnarray*}
        I+E_{a,b}\otimes E_{c,d} &=& I\otimes (I+E_{c,d}) + (I+E_{a,b})\otimes E_{c,d} \\
        & = & (I+E_{a,b})\otimes I + E_{a,b}\otimes (I+E_{c,d}) 
    \end{eqnarray*}
One of the two expressions above will always contain at most two non-invertible terms, and hence two applications of \cref{supmat:lem:sum of two} yields $w(I+E_{a,b}\otimes E_{c,d}) \leq 4$. Thus the corresponding in-block CNOT operator has depth at most 4.

We remark that the most common instance will be that both $a \neq b$ \emph{and} $c \neq d$, in which case there is only one non-invertible term, yielding depth 3.
    \end{proof}
    
\begin{rem*}
    \textbf{(Space Costs)} Throughout this section we have been primarily concerned with minimising the depth of CNOT operator implementations. This directly translates to minimising the number of logical cycles required in SHYPS codes and (assuming a fixed time cost for syndrome extraction), the total time cost. Of lesser importance has been tracking the space overhead of these operator implementations but this can be easily derived: cross-block logical generators in $\calA \cup \calA^T$ require 0 additional qubits, whereas \cref{supmat:lem:in-block CNOT} demonstrates that applying an in-block CNOT circuit requires a single auxiliary code block. Hence, assuming that in-block operators are performed in parallel (as in the proof of \cref{supmat:lem:b blocks}), a logical CNOT operator on $b$ code blocks incurs a space overhead of at most $bn$ qubits. Additionally, in the large $b$ limit, \cite[Thm. 4]{Grassl2013} implies that we can get away with 0 additional code blocks without changing the leading order terms for time overheads.
\end{rem*}

\section{Diagonal operators in SHYPS codes}\label{supmat:sec: diagonal gate characterisation}
The previous section characterizes the depth-1 CNOT circuits that arise as lifts of classical automorphisms, and furthermore, how such circuits generate the full group of CNOT operators across multiple blocks of the SHYPS codes. We now turn our attention to diagonal Clifford gates, i.e., those that correspond to diagonal unitary matrices with respect to the computational basis. In particular we describe a collection of depth-1 logical diagonal Clifford operators that leverage automorphisms of the classical simplex code $C(r)$, and demonstrate how these efficiently generate all diagonal Cliffords.
Recall from \cref{supmat:subsec:review-of-paulis-and-cliffords} that up to phase, such Clifford operators are generated by $S$ and CZ gates and form an abelian subgroup of $\calC_n$. 

Before going through the details on diagonal operators, we note that all results below can be trivially converted to results on $X$-diagonal operators, i.e., Clifford operators that are diagonal unitary matrices when considered in the \emph{Hadamard-rotated} computational basis. This subgroup of Clifford operators is obtained by conjugating diagonal operators by the all-qubit Hadamard operator. 
The fact that results on diagonal operators carry over to $X$-diagonal operators can be understood by considering that conjugating a diagonal operator by the all-qubit Hadamard operator transforms its symplectic representation by moving the off-diagonal block from the top-right quadrant to the bottom-left quadrant. 

\subsection{Lifting classical automorphisms}
Recall that diagonal Clifford operators are represented modulo Paulis by symplectic matrices
\begin{align}\label{supmat:eq:diagCliffSymplecticForm}
    \{ \mat{I & B \\ 0 & I} \in Sp_{2k}(2) \mid B \in SYM_k(2) \},
\end{align}
where $SYM_k(2)\subset M_k(2)$ denotes the subspace of symmetric matrices.

As the diagonal subgroup is abelian, efficiently generating all operators thus becomes a question of decomposing arbitrary symmetric matrices as short sums from a particular collection of generators. We construct these generating symmetric matrices from classical automorphisms such that the corresponding diagonal operators have favorable depth and fault-tolerance properties.

Before proceeding we first define a distinguished permutation on our qubit array; in the language of \cite{Breuckmann2024}, this permutation is a \emph{$ZX$-duality}. Intuitively, the map $\tau$ exchanges the vector spaces defining $X$- and $Z$-gauge operators, and will be used to account for the fact that conjugating $X$-Paulis by diagonal Cliffords may produce $Z$-Paulis.
\begin{defn}
    Given $a^2$ qubits (physical or logical) arranged in an $a \times a$ array, let $\tau_a \in S_{a^2}$ be the permutation that exchanges qubits across the diagonal, i.e., $\tau_a$ is an involution exchanging qubits $(i,j) \leftrightarrow (j,i)$. In particular,  $\tau_a = \tau_a^{-1} = \tau_a^T$. 

    As we typically index qubits by the tensor product basis $e_i^T\otimes e_j^T$, we have equivalently in vector form that $\tau_a(e_i^T \otimes e_j^T)=e_j^T \otimes e_i^T$. This extends naturally to matrices acting on this basis:
    \[
    \tau_a (A \otimes B) \tau_a = B\otimes A.
    \]
\end{defn}

We're now ready to introduce our generating set of logical diagonal Clifford operators in the SHYPS codes: a collection of so-called \emph{phase-type} fold-transversal operators \cite{Breuckmann2024}.
\begin{lem} \label{supmat:lem:depth 1 diagonal lifts}
Let $r\geq 3$ and $C(r)$ be the $(n_r,r,2^{r-1})$-simplex code with generator matrix $G$. Furthermore, assume that $\sigma \in \aut{C(r)}$ with corresponding linear transformation $g_{\sigma} \in GL_r(2)$ such that $g_{\sigma} G = G \sig$.

    Then $\sig$ lifts to a physical diagonal Clifford operator
    
    \begin{align}
    \label{supmat:eq:physicalDiag}
    U(\sig):= \mat{I & (\sig \otimes \sig^T)\cdot \tau_{n_r} \\ 0  & I}
    \end{align}
    that preserves $SHYPS(r)$. Moreover, this depth-1 circuit of diagonal gates has logical action given by 
    \[
    \overline{U}(g_\sig):= \mat{I & (g_{\sig}^{-T} \otimes g_\sig^{-1}) \cdot \tau_r \\ 0 & I}.
    \]
\end{lem}
\begin{proof}
    Firstly we need to establish that $B = (\sig \otimes \sig^T)\cdot \tau_{n_r}$ is symmetric for $U(\sig)$ to define a valid physical diagonal Clifford. As the transpose and inverse of a permutation matrix are equal, it suffices to check that $B$ is self-inverse:
    \begin{eqnarray*}
        B^2 & = & (\sig \otimes \sig^T)\tau_{n_r} (\sig \otimes \sig^T)\tau_{n_r} \\
        & = & (\sig \otimes \sig^T)(\sig^T \otimes \sig) \\
        & = & (\sig\sig^T \otimes \sig^T\sig) \\
        & = & I
    \end{eqnarray*}
    Next observe that as a product of permutation matrices, $B$ has row/column weights equal to one and thus corresponds to a depth-1 physical circuit.
    As $U$ commutes with $Z$-gauge operators, to confirm that $U$ acts as a logical gate of $SHYPS(r)$ we need only check the action on $X$-gauge operators, described by $G_X$. Well
    \begin{eqnarray*}
     G_X \cdot B & = & (H \otimes I_{n_r})\cdot (\sig \otimes \sig^T)\cdot \tau_{n_r}\\
     & = & (H\sig \otimes \sig^T) \cdot \tau_{n_r} \\
     & = & (h_{\sig} \otimes \sig^T)(H \otimes I) \cdot \tau_{n_r} \\
     & = & (h_{\sig} \otimes \sig^T)\cdot \tau_{n_r} \cdot (I \otimes H), 
    \end{eqnarray*}
where we've used that $\aut{C(r)} =\aut{C(r)^\perp}$  and hence there exists $h_{\sig}$ such that $H\sig = h_{\sig}H$.
    But 
    \[
    (h_{\sig} \otimes \sig^T)\cdot \tau_{n_r} \cdot (I \otimes H) = (h_{\sig} \otimes \sig^T)\cdot \tau_{n_r} \cdot G_Z
    \]
    and hence the action of $U(\sig)$ is given by
    \[
    \left[G_X,0\right]\cdot U = \left[G_X, (h_{\sig} \otimes \sig)\cdot \tau_{n_r} \cdot G_Z \right].
    \]
    We conclude that conjugation by the physical Clifford $U$ preserves stabilizers and hence performs a valid logical Clifford operator on $SHYPS(r)$. Note that by the nature of the symplectic group, this logical Clifford is implemented up to a Pauli correction. As there is always a Pauli operator with the appropriate (anti)commutation relations with the stabilizers and logical operators of the code to fix any logical/stabilizer action sign issues, we can ignore this subtlety \cite{Sayginel2024}.
    The logical action of 
    $\overline{U}$ is then determined akin to the proof of Lemma \ref{supmat:lem:perm to logical}:
    \begin{align*}
     L_X \cdot B & = (P \otimes G)\cdot (\sig \otimes \sig^T)\cdot \tau_{n_r}\\
     & = (P\sig \otimes G\sig^T) \cdot \tau_{n_r} \\
     & = (g^{-T} \otimes g^{-1})(P \otimes G) \cdot \tau_{n_r} \\
     & = (g^{-T} \otimes g^{-1})\cdot \tau_r \cdot (G \otimes P) \\
     & = (g^{-T} \otimes g^{-1})\cdot \tau_r \cdot L_Z.\qedhere
    \end{align*}
\end{proof}

As we are focused primarily on the logical action and not the particular automorphism being used, we shall typically relabel $g^{-T} \mapsto g$ for ease of presentation. Summarising the above we have
\begin{cor} \label{supmat:cor:diagonal generators}
    Let $g \in GL_r(2).$ The logical Clifford operator $\overline{U} = \mat{I & (g\otimes g^T)\cdot \tau_r \\ 0 & I}$ on code $SHYPS(r)$ is implementable by a depth-1 physical diagonal Clifford circuit.
\end{cor}

\begin{exmp}
    Let $r=3$ and $G$ be the simplex code generator matrix from (\ref{supmat:eq:r=3 gen matrix}). Observe that
    \begin{eqnarray*}
         \mat{1 & 1 & 0 \\
    0 & 1 & 0 \\
    0 & 0 & 1
    }G &
    = &
    \mat{1 & 1 & 0 \\
    0 & 1 & 0 \\
    0 & 0 & 1
    }
    \mat{
        1 & 0 & 1 & 1 & 1 & 0 & 0 \\
        0 & 1 & 0 & 1 & 1 & 1 & 0 \\
        0 & 0 & 1 & 0 & 1 & 1 & 1 
    } \\
    & = & 
     \mat{
        1 & 0 & 1 & 1 & 1 & 0 & 0 \\
        1 & 1 & 1 & 0 & 0 & 1 & 0 \\
        0 & 0 & 1 & 0 & 1 & 1 & 1 
    } \\
    & = & G\cdot(1,4)(3,5).
    \end{eqnarray*}
   Hence 
   \begin{eqnarray*}
       U & = &  \mat{ I & \left[(1,4)(3,5) \otimes (1,4)(3,5)\right]\cdot \tau_7 \\ 0 & I} \\
   \end{eqnarray*}
   implements
   \[
   \overline{U} = \mat{I &  \mat{1 & 0 & 0 \\
    1 & 1 & 0 \\
    0 & 0 & 1
    } \otimes  \mat{1 & 1 & 0 \\
    0 & 1 & 0 \\
    0 & 0 & 1
    }\tau_3 \\ 0 & I }.
   \]

\end{exmp}

\begin{rem*}
Recall that the presence of $S$ and CZ gates in diagonal Clifford operators (\ref{supmat:eq:diagCliffSymplecticForm}) is determined by the diagonal and off-diagonal entries of $B$, respectively. Then specifically for our generators, diagonal entries are determined by the fixed points of the permutation matrix $\sig \otimes \sig^T \cdot \tau_{n_r}$. Here we see that
    \begin{eqnarray*}
        \sig \otimes \sig^T \cdot \tau_{n_r}(e_i^T\otimes e_j^T) & = & \sig \otimes \sig^T(e_j^T \otimes e_i^T) \\
        & = & e^T_{\sig(j)}\otimes e^T_{\sig^T(i)}
    \end{eqnarray*}
    equals $e^T_i \otimes e^T_j$ if and only if $\sig(j)=i$.
    Hence $\sig \otimes \sig^T \cdot \tau_{n_r}$ has exactly $n_r=\sqrt{n}$ fixed points $\{e^T_{\sig(j)} \otimes e^T_{j}\}_{j=1}^{n_r}$, and the corresponding physical diagonal Clifford generator contains $n_r$ $S$ gates, and $(n_r^2-n_r)/2$ CZ gates.
\end{rem*}

Corollary \ref{supmat:cor:diagonal generators} establishes a collection of logical diagonal Cliffords that are implemented in SHYPS codes by depth-1 physical circuits and we denote these by
\[
\calB:=\{\mat{I & (g \otimes g^T) \cdot \tau_r \\ 0 & I}\,:\,g\in GL_r(2) \}.
\]

Unlike the CNOT circuits of \cref{supmat:lem:depth 1 cnot circuits} these operators are not implemented strictly transversally. Furthermore, as they contain entangling CZ gates within a code block, one may be concerned that $X$-faults occurring mid-circuit spread to a high-weight $Z$-error pattern on different qubits. However from a fault-tolerance perspective, this error spread is only of concern if the additional Paulis increase the support of a nontrivial logical operator. Otherwise the final error pattern, and the collection of circuit faults that caused it, are equally correctable by the code's error correction protocol. 

This notion of fault-tolerance is quantified by the \emph{circuit distance} $d_{circ}$ which is the minimum number of faults required in a noisy quantum circuit to produce a logical error, without triggering any syndromes. It's clear that $d_{circ}$ is bounded above by the code distance, and if $d_{circ}=d$ we call the circuit \emph{distance preserving}.

As the following Lemma demonstrates, our diagonal operators $\calB$ meet this distance preserving criterion.

\begin{lem} \label{supmat:lem:DiagGenType1Distance}
    Let $\overline{U}\in \calB$. Then there exists a physical implementation with circuit distance $d_{circ}=d$.
\end{lem}
\begin{proof}
By Lemma \ref{supmat:lem:depth 1 diagonal lifts}, there exists $\sig\in S_{n_r}$ such that $\overline{U}$ is implemented by $U(\sig \otimes \sig^T)$ where the off-diagonal entries of $\sig \otimes \sig^T$ determine the presence of physical CZ gates. As discussed above, CZ gates spread $Y/Z$ faults but by the structure of the logical Paulis in $SHYPS(r)$ (see \cref{supmat:thm:logical basis for SHP}), a single $X$-fault spreads to a weight two $Z$-Pauli in the support of a logical operator if and only if the CZ connects qubits in the same row or column. Hence to show that $d_{circ}=d$ it suffices to show that this is not the case.

Well firstly consider two distinct qubits $e_i\otimes e_j$ and $e_i \otimes e_k$, $j \neq k$ in some fixed row $i$. Then
\begin{eqnarray*}
    (\sig \otimes \sig^T)\tau_{{n_r}_{(i,j),(i,k)}} & 
    = & (e_i \otimes e_j) \cdot (\sig \otimes \sig^T)\tau_{n_r}\cdot (e_i \otimes e_k)^T \\
    & = & (e_i \sig e^T_k) \otimes (e_j \sig^T e^T_i) \\
    & = & (e_i \sig e^T_k) \otimes (e_i \sig e^T_j)^T.
\end{eqnarray*}
But $(e_i \sig e^T_k)=1$ if and only if $\sig(k)=i$ and this cannot also hold for $k\neq j$. Hence $(\sig \otimes \sig^T)\tau_{{n_r}_{(i,j),(i,k)}}=0$ for $j\neq k$, confirming that $U(\sig \otimes \sig^T)$ contains no CZ gates joining qubits within a row. The result for columns follows identically.
\end{proof}

In summary, the diagonal operators $\overline{U}(g\otimes g^T)$ are depth-1, distance-preserving (i.e., fault-tolerant) logical gates of $SHYPS(r)$. It's natural therefore to proceed as in Section \ref{supmat:subsec:full cnot group gen} and consider the subgroup of logical diagonal gates they generate.

\subsection{Generating in-block operators }\label{supmat:subsec:DiagGensType1Proof}
In this section we show that the fault-tolerant logical gates $\calB$ generate all diagonal Clifford operators of the SHYPS codes. Furthermore, the compilation of operators using this generating set scales asymptotically optimally. For exact costings of the overhead required to implement a range of diagonal operators we refer the reader to \cref{supmat:tab:diagonal_depth}, but first we focus on proving the following result.

\begin{thm}\label{supmat:thm:DiagGensType1Generate}
    Every in-block logical diagonal Clifford operator (modulo Paulis) in $SHYPS(r)$ is implementable by a sequence of at most $r^2+5r+2$ generators from $\calB$ for $r\geq 4$ and at most $r^2+8r+2$ generators for $r=3$.
\end{thm}

As aforementioned, the diagonal Clifford operators (modulo Paulis) of the $SHYPS(r)$ code are in one-to-one correspondence with symmetric matrices $SYM_{r^2}(2)$. Hence \cref{supmat:thm:DiagGensType1Generate} is an immediate corollary of the following result.

\begin{thm} \label{supmat:thm:diagonal-decomp-invertible}
        Let $S \in SYM_{r^2}(2)$. Then there exist $A_1,\dots, A_p \in GL_r(2)$ for some $p\leq r^2+5r+2$ ($r^2+8r+2$ when $r=3$) such that
    \[
    S = \sum_{i=1}^p (A_i\otimes A_i^T)\tau.
    \]
\end{thm}

Similar to the work in \cref{supmat:sec: cnot gate characterisation}, the strategy for proving Theorem \ref{supmat:thm:diagonal-decomp-invertible} is to first show an analogous (and typically stronger) result over (possibly) non-invertible matrices. We then adapt such a decomposition to invertible matrices, while minimising any additional overhead incurred.

\begin{lem}\label{supmat:lem:diagonal decomp}
    Let $S \in SYM_{r^2}(2)$. Then there exist $M_1,\dots, M_a \in M_r(2)$ for some $a\leq r^2+1$, such that
    \[
    S = \sum_{i=1}^a (M_i\otimes M_i^T)\tau.
    \]
\end{lem}
The proof of \cref{supmat:lem:diagonal decomp} uses the following notions of matrix reshaping
\begin{defn}
    \textbf{(Flatten)} Let $M = (m_{i,j})\in M_{r^2}(2)$. Then the \emph{flatten function} $fl: M_r(2) \longrightarrow \gf^{1\times r^2}$ is defined 
    \[
    fl: \begin{pmatrix}
        m_{1,1} & \dots & m_{1,r} \\

        \vdots & & \vdots \\
        m_{r,1} & \dots & m_{r,r} \\
    \end{pmatrix} \mapsto 
    \begin{pmatrix}
        m_{1,1} & \dots & m_{1,r} & m_{2,1} & \dots & m_{r,r}
    \end{pmatrix}.
    \]
    By considering the action of $fl(M)$ on basis vectors $e_i^T \otimes e_j^T$ and $e^T_j \otimes e^T_i$, we see immediately that $fl(M^T)=fl(M)\cdot \tau$.
    
    \textbf{(Reshape)}
    Given $M \in M_{r^2}(2)$, there exist block matrices $M_{i,j}$ such that 
    \[
    M = \sum_{1\leq i,j\leq r} E_{i,j}\otimes M_{i,j}.
    \]
    Then the re-shape function $re: M_{r^2}(2) \longrightarrow M_{r^2}(2)$ is defined as follows 
    \[
     re: M \mapsto  \sum_{1\leq i,j\leq r} fl(E_{i,j})^T\cdot fl(M_{i,j})
    \]
    I.e., the successive rows of $re(M)$ are formed by flattening successive blocks of $M$
\end{defn}
For example
\[
re: \begin{pmatrix}
    1 & 0 & 1 & 1 \\
    1 & 0 & 1 & 0 \\
    1 & 1 & 1 & 1 \\
    0 & 0 & 0 & 0 \\
\end{pmatrix} \mapsto
\begin{pmatrix}
    1 & 0 & 1 & 0 \\
    1 & 1 & 1 & 0 \\
    1 & 1 & 0 & 0 \\
    1 & 1 & 0 & 0 \\
\end{pmatrix}.
\]
The effect of the  reshape map is more evident on how it acts on basis vectors
\[
(e_{i_1} \otimes e_{i_2})re(M)(e^T_{j_1} \otimes e^T_{j_2}) = (e_{i_1} \otimes e_{j_1})M(e_{i_2}^T \otimes e_{j_2}^T).
\]
I.e., it effectively swaps the indices $i_2$ and $j_1$. This is thus linear and also self-inverse.
\begin{proof}
    (Proof of Lemma \ref{supmat:lem:diagonal decomp}) The proof strategy is to use a binary Cholesky decomposition \cite{Lempel1975} that decomposes symmetric matrices as sums of outer products of vectors with themselves. To account for $\tau$ we don't apply this directly to the desired $S$ but rather a transformed matrix $S' = re(S\cdot\tau)\cdot\tau$ - undoing this reshaping later correspondingly transforms outer products into summands of the desired tensor product form.

    First observe that
    \begin{eqnarray*}
        (e_{i_1} \otimes e_{i_2})S'(e_{j_1}^T \otimes e_{j_2}^T) & = & (e_{i_1} \otimes e_{i_2})re(S\cdot \tau)\cdot \tau(e_{j_1}^T \otimes e_{j_2}^T) \\
        & =  & (e_{i_1} \otimes e_{j_2})S(e_{j_1}^T \otimes e_{i_2}^T) \\
        & = & (e_{j_1} \otimes e_{i_2})S(e_{i_1}^T \otimes e_{j_2}^T) \\
        & = & (e_{j_1} \otimes e_{j_2})S'(e_{i_1}^T \otimes e_{i_2}^T), \\
    \end{eqnarray*}
    where we have unpacked the definitions of $\tau, re$ and used that $S$ is symmetric. Hence, by definition $S'$ is also symmetric.

    Now we apply the work of \cite{Lempel1975} to decompose $S'$ via a binary analogue of the well-known Cholesky decomposition. Let $a=\rank(S')+(\prod (S_{i,i}+1) \mod 2)$ (i.e., the rank of $S'$, with an additional plus one if all diagonal entries are zero, so in particular $a \leq r^2+1$). Then there exists $L \in \gf^{r^2 \times a}$ such that $S' = LL^T$ \cite[Thm. 1]{Lempel1975}.  

    Next observe that the matrix product $LL^T$ may be rewritten as a sum of outerproducts of the columns $l_k$ of $L$ with themselves
    \[
    S' = LL^T = (l_1 \dots l_a)(l_1 \dots l_a)^T = \sum_{k=1}^a l_kl_k^T.
    \]
    Considering a single summand of the form $ll^T$, there exists $M = (m_{i,j})\in M_r(2)$ such that $l^T = fl(M)$ and hence
    \[
    ll^T = \begin{pmatrix}
        m_{1,1} \\ \vdots \\ m_{r,r}
    \end{pmatrix}(m_{1,1},\dots,m_{r,r}) = \sum_{i,j}m_{i,j} fl(E_{i,j})^T fl(M).
    \]
    It then follows that 
    \begin{eqnarray*}
        re^{-1}(ll^T\cdot \tau) & = & \sum_{i,j}m_{i,j} \cdot re^{-1} \left(fl(E_{i,j})^T fl(M)\cdot \tau \right) \\ 
        & = & \sum_{i,j}m_{i,j} \cdot re^{-1} \left(fl(E_{i,j})^T fl(M^T) \right) \\ 
         & = & \sum_{i,j}m_{i,j} E_{i,j} \otimes M^T \\ 
         & = & M \otimes M^T.
    \end{eqnarray*}
    Finally, letting $fl(M_k)=l_k^T$ for each column $l_k$, we have by linearity that
    \begin{eqnarray*}
        S & = & re^{-1}(S'\cdot \tau) \cdot \tau \\
        & = & \sum_{k=1}^a (M_k \otimes M_k^T)\tau.
    \end{eqnarray*}
\end{proof}

\begin{exmp}
    Consider
    \[
    S = \begin{pmatrix}
        0 & 0 & 0 & 0 \\
        0 & 0 & 1 & 1 \\
        0 & 1 & 0 & 0 \\
        0 & 1 & 0 & 1
    \end{pmatrix} \mapsto S' = re(S\cdot \tau)\cdot \tau = \begin{pmatrix}
        0 & 0 & 0 & 1 \\
        0 & 0 & 0 & 1 \\
        0 & 0 & 0 & 0 \\
        1 & 1 & 0 & 1
    \end{pmatrix}.
    \]
    Applying the constructive algorithm in \cite[Sec. 3]{Lempel1975} we produce a minimal binary Cholesky decomposition
    \[
    S' = \begin{pmatrix}
        1 & 1 \\
        1 & 1 \\
        0 & 0 \\
        0 & 1 \\
    \end{pmatrix}
    \begin{pmatrix}
        1 & 1 & 0 & 0 \\
        1 & 1 & 0 & 1 \\
    \end{pmatrix}.
    \]
    Isolating the summand corresponding to the second column, we have
    \begin{eqnarray*}
      re^{-1}(l_2l_2^T \tau) & = & re^{-1} \left(\begin{pmatrix}
        1 & 1 & 0 & 1 \\
        1 & 1 & 0 & 1 \\
        0 & 0 & 0 & 0 \\
        1 & 1 & 0 & 1 \\
    \end{pmatrix} \tau \right) \\ &
     = & 
    \begin{pmatrix}
        1 & 0 & 1& 0 \\
        1 & 1 & 1 & 1 \\
        0 & 0 & 1 & 0 \\
        0 & 0 & 1 & 1 \\
    \end{pmatrix} \\ & = & \begin{pmatrix}
        1 & 1 \\ 0 & 1
    \end{pmatrix} 
    \otimes
    \begin{pmatrix}
        1 & 1 \\ 0 & 1
    \end{pmatrix}^T.  
    \end{eqnarray*}
    
    Handling the first column $l_1$ similarly then yields
    \[
    S = \begin{pmatrix}
        1 & 1 \\ 0 & 0
    \end{pmatrix} 
    \otimes
    \begin{pmatrix}
        1 & 1 \\ 0 & 0
    \end{pmatrix}^T \cdot \tau+ \begin{pmatrix}
        1 & 1 \\ 0 & 1
    \end{pmatrix} 
    \otimes
    \begin{pmatrix}
        1 & 1 \\ 0 & 1
    \end{pmatrix}^T \cdot \tau.
    \]
\end{exmp}

With Lemma \ref{supmat:lem:diagonal decomp} established, we now investigate methods for converting a sum of non-invertible $M_i \otimes M_i^T$ into a similar decomposition comprising only invertible matrices.

To this end, it is useful to set up some notation and collect a few preliminary lemmas.
\begin{defn}
    Let $A,B \in M_r(2)$. We define
    \[
    S(A): = (A\otimes A^T)\tau
    \]
    \begin{eqnarray*}
            D(A,B) &:=& S(A+B) +S(A)+S(B) \\ &=& (A \otimes B^T + B \otimes A^T) \tau
    \end{eqnarray*}
\end{defn}
In this notation, the goal is to find a minimal number of $A_i \in GL_r(2)$ such that $S = \sum_{i} S(A_i)$. Similarly to \cref{supmat:sec: cnot gate characterisation}, it useful to introduce a weight function to track this
 \begin{defn}
       Let $S \in SYM_{r^2}(2)$. The \emph{weight} of $S$, denoted $w(S)$, is the minimal number of $A_i \in GL_r(2)$ such that $S = \sum_{i}^{w(S)} S(A_i)$.
 \end{defn}
 We remark that terms of the form $D(A,B)$ are the natural ``cross-terms'' incurred when splitting non-invertible matrices into sums of invertibles - dealing with these is the main remaining challenge.
 
 The following facts are immediate from the definitions
 
 \begin{lem} \label{supmat:lem:trick1}
     Let $A,B,C_i \in M_r(2)$. Then the following hold
  \begin{enumerate}[label=(\alph*)]
  \item $D(A,B) = D(B,A)$
  \item $D(A,B) = D(A,A+B)$
  \item $D(A,\sum_i C_i) =  \sum_i D(A,C_i). $
  \end{enumerate}
 \end{lem}
 
\begin{lem} \label{supmat:lem:trick2}
    Let $A,B \in M_r(2)$ and assume there exists $C\in GL_r(2)$ such that \[A+C,B+C, A+B+C \in GL_r(2).\] Then
    \[
    D(A,B)= S(A+C)+S(B+C)+S(A+B+C)+S(C).
    \]
    In particular $w(D(A,B)) \leq 4.$
\end{lem}

\begin{lem}\label{supmat:lem:trick3}
    Let $A,B \in M_r(2)$ and $P,Q \in GL_r(2)$. Then 
    \[
    (P \otimes Q)S(A)(P^T \otimes Q^T) = S(PAQ^T)
    \]
    \[
    (P\otimes Q)D(A,B)(P^T \otimes Q^T) = D(PAQ^T, PBQ^T)
    \]
    Hence for $M \in M_{r^2}(2)$, we have that \[w(M) = w( (P \otimes Q)M(P^T \otimes Q^T)).\] In particular, $w(D(A,B))=w(D(PAQ^T,PBQ^T))$.
\end{lem}
\begin{proof} (Proof of Theorem \ref{supmat:thm:diagonal-decomp-invertible})
By Lemma \ref{supmat:lem:diagonal decomp} there exists a decomposition

\begin{equation}\label{supmat:eq:singular-decomp}
    S = \sum^{a}_{i=1} (M_i \otimes M_i^T) \tau,
\end{equation}  
for some $a \leq r^2+1$.
We first consider the addition of diagonal matrices $D_i$ to each component $M_i$ - as seen in \cref{supmat:sec: cnot gate characterisation}, this can be advantageous as the space spanned by the $D_i$ has dimension at most $r$. So, applying \cref{supmat:lem:diaginvert}, there exist invertible matrices $A_i$ and diagonal $D_i$ such that $M_i = A_i + D_i$. It follows that
\begin{eqnarray}
    S & = & \sum_{i=1}^a S(M_i) \\
      & = & \sum_{i=1}^a S(A_i) + S(D_i) + D(A_i, D_i). \label{supmat:eq:expression1}
\end{eqnarray}
Note that each diagonal $D_i$ may itself be written in the standard diagonal basis $E_{j,j} \in M_r(2)$ to yield
\begin{equation}
    S(D_i) = \sum_{j=1}^r d_{i,j} \left(S(E_{j,j})+ D(E_{j,j},\sum_{k>j}d_{i,k}E_{k,k})\right), \label{supmat:eq:expression2}
\end{equation}

for some $d_{i,j}\in \gf$.
Then substituting each decomposition (\ref{supmat:eq:expression2}) into (\ref{supmat:eq:expression1}), expanding each $D(A_i,D_i)$ in the digaonal basis, and collecting terms by repeated use of Lemma \ref{supmat:lem:trick1} (a) and (c) yields
\begin{equation} \label{supmat:eq:expression3}
    S = \sum_{i=1}^a S(A_i) + \sum_{j=1}^r e_{j}S(E_{j,j}) + \sum_{j=1}^r f_{j}D(E_{j,j},B_j),
\end{equation}
for some $e_{i},f_i\in \gf$ and $B_i \in M_r(2)$. Note in particular that the latter two sums contain at most $r$ terms each, and that the $B_j$ are not necessarily invertible. Of course each $E_{j,j} \in M_r(2) \backslash GL_r(2)$ also, but applying Lemma \ref{supmat:lem:diagperm}, there exists a single $r$-cycle permutation matrix $\sigma$ such that each $E_{i,i}+\sigma \in GL_r(2)$. By substituting 
\[
S(E_{i,i}) = S(E_{i,i}+\sigma)+S(\sigma)+D(E_{i,i},\sigma)
\] into (\ref{supmat:eq:expression3}) and again collecting like terms by applying Lemma \ref{supmat:lem:trick1} yields
\begin{eqnarray}\label{supmat:eq:expression4}
    S & = & \sum_{i=1}^a S(A_i)+ \sum_{j=1}^re_jS(E_{j,j}+\sigma) \\
    & + & S(\sigma)\cdot(\sum_j e_j)+\sum_{j=1}^r f_j'D(E_{j,j},B'_j).
\end{eqnarray}

To summarise, up to the addition of $\sum_i^r f_i'D(E_{i,i},B'_i)$, we have that $w(S) \leq r^2+r+2$. It thus remains to handle this final term. Moreover, as there exist permutation matrices $P,Q \in GL_{r}(2)$ such that $E_{j,j} = PE_{1,1,}Q^T$, we may restrict our attention to a single summand $D(E_{1,1},B_1)$ by Lemma \ref{supmat:lem:trick3}.

Now again using Lemma \ref{supmat:lem:trick3}, observe that if $P,Q\in GL_r(2)$ are such that $PE_{1,1}Q^T = E_{1,1}$, then $w(D(E_{1,1},B_1))= w(D(E_{1,1},PB_1Q^T))$. But this holds for a large range of $P$ and $Q$:
\[
\begin{pmatrix}
    1 & p^T \\
    \mathbf{0} & P'
\end{pmatrix} \cdot E_{1,1} \cdot  
\begin{pmatrix}
    1 & \mathbf{0}^T \\
    q & Q'
\end{pmatrix} = E_{1,1},
\]
for all vectors $p,q\in \gf^{r-1}$ and $P',Q' \in GL_{r-1}(2)$.
Left and right multiplication by such $P$ and $Q^T$  allows us to perform a modified Gaussian elimination on $B_1$, reducing it to the form
\[
B_1' = \begin{pmatrix}
    X & \\
    & D 
\end{pmatrix} = 
\begin{pmatrix}
    a & b & 0 & 0 & \\
    c & d & 0 & 0 &\\
    0 & 0 & e & 0 & \\
    0 & 0 & 0 & f & \\
    & & & & D \\
\end{pmatrix},
\]
where $D \in \diag_{r-4}(2)$ is diagonal.

Now if $r=4,5$ we check computationally that for all $B_1'$ there exists $C \in GL_r(2)$ such that $C,C+B_1', C+E_{1,1}, C+B_1'+E_{1,1}$ are invertible.
Thus by Lemma \ref{supmat:lem:trick2}, $w(D(E_{1,1},B_1'))\leq 4$, and by the arguments above, all $w(D(E_{j,j},B_j))\leq 4$.

If $r\geq 6$ then first note by the $r=4$ computations, there exists an invertible $C \in GL_4(2)$ such that $X+C,E_{1,1}+C$ and $X+E_{1,1}+C$ are all invertible. To handle the lower block of dimension $r-4\geq 2$, if $D$ is not full rank, then we choose $C'$ to be any $(r-4)$-cycle permutation matrix and $D+C'\in GL_{r-4}(2)$ by Lemma \ref{supmat:lem:diagperm}. If $D=I_{r-4}$, then it suffices to choose any invertible $C'$ that fixes only the zero vector (and such a matrix always exists for $r-4 \geq 2$). 

Hence letting $C'' = \begin{pmatrix}
    C & \\ & C'
\end{pmatrix}$, we have that $C''$
\[
E_{1,1}+C'' = \begin{pmatrix}
    E_{1,1}+C & \\ & C'
\end{pmatrix}, \] \[
B_1'+C'' = \begin{pmatrix}
    X+C & \\ & D+C' \\
    \end{pmatrix},\]\[
    B_1'+E_{1,1}+C'' = \begin{pmatrix}
    X+E_{1,1}+C & \\ & D+C' \\
    \end{pmatrix}
\]
are all invertible. Thus as above, we have by Lemma \ref{supmat:lem:trick2}, $w(D(E_{1,1},B_1'))\leq 4$, and similarly, all $w(D(E_{j,j},B_j))\leq 4$.
In conclusion, when $r\geq 4$ we collect terms in (\ref{supmat:eq:expression4}), to see that there exist at most $r^2+5r+2$ matrices $A_i \in GL_r(2)$ such that $S = \sum_{i} S(A_i)$.

It remains to consider the case $r=3$. Here we note that provided 
\begin{equation}\label{supmat:eq:counterexamples}
    B_1' \neq \begin{pmatrix}
    0 & 1 & 0 \\ 1 & 0 & 0 \\ 0 & 0 & 1
\end{pmatrix},
\,\, \begin{pmatrix}
    1 & 1 & 0 \\ 1 & 0 & 0 \\ 0 & 0 & 1
\end{pmatrix},
\end{equation}
we may proceed as above, as there exists a choice of $C \in GL_3(2)$ that satisfies the conditions of \cref{supmat:lem:trick2}. However if $B_1'$ is one of the two exclusions in (\ref{supmat:eq:counterexamples}) then no such $C$ exists, and we in fact check computationally that $w(D(E_{1,1},B_1')) = 7$. The $r$ possible terms of this form thus incurs an additional cost of up to $7r$, giving an overall bound of $r^2+8r+2$.
\end{proof}

In summary, we have demonstrated that all logical diagonal Clifford operators of the $SHYPS(r)$ codes may be implemented fault-tolerantly, using generators of (physical) depth 1. Moreover, the total depth required for an arbitrary diagonal operator is bounded above by $r^2+5r+2$ (for $r\geq 4$). Comparing the number of generators \[
\vert \calB \vert = \vert GL_r(2) \vert \sim 2^{r^2} \] to the total number of diagonal operators 
\[\vert SYM_{r^2}(2) \vert = 2^{r^2(r^2+1)/2}, \] we see that our achieved compilation in $O(r^2)$ steps is optimal up to a constant factor. Moreover as $k=r^2$ is the number of logical qubits, this mirrors the case of $k$ un-encoded qubits: given a symmetric matrix $S \in SYM_k(2)$, the depth of the (un-encoded) diagonal Clifford operator is determined by solving an associated graph colouring problem (ignoring diagonal entries, treat $S$ as an adjacency matrix). Colouring the complete graph on $k$ vertices requires $k+1$ colours, and correspondingly, there exist diagonal operators that require depth $k+1$. Hence encoding using the SHYPS codes incurs a negligible depth penalty.

\subsection{Compiling specific operators}
\cref{supmat:thm:DiagGensType1Generate} demonstrates efficient compiling of arbitrary (in-block) diagonal Clifford operators in the SHYPS codes. However, costings of single one- and two-qubit logical gates are also of interest, and we examine these below. For a full breakdown of logical diagonal Clifford costs, we refer the reader to \cref{supmat:tab:diagonal_depth}

Let's first consider $S$ circuits.
\begin{lem}\label{supmat:lem:single S}
    Let $S_{(i,j)}$ be a single qubit logical $S$ operator corresponding to the $(i,j)$-position in the $r \times r$ logical qubit array of $SHYPS(r)$. Then the following holds
    \begin{enumerate}
        \item There exists an implementation of $S_{(i,j)}$ using $6$ ($9$ when $r=3$) generators from $\calB$.
        \item There exists a depth-1 operator in $\calB$ that performs $S_{(i,j)}$ (plus additional diagonal gates, in general)
    \end{enumerate}
\end{lem}
\begin{proof}
    Recall  that the operator $S_{(i,j)}$ corresponds to the symmetric matrix $B$ with single nonzero entry in the $(i,j)$th diagonal position. In fact, $B=(E_{i,j}\otimes E_{i,j}^T)\tau_r=:S(E_{i,j})$. By \cref{supmat:lem:sum of two} there exists $X \in GL_r(2)$ such that $E_{i,j}+X \in GL_r(2)$ and this yields 
    \[
    S(E_{i,j}) = S(E_{i,j}+X)+S(X)+D(E_{i,j},X).
    \]
    Observing that there exist permutations matrices $P,Q$ such that $PE_{i,j}Q^T=E_{1,1}$, we follow the proof of \cref{supmat:thm:diagonal-decomp-invertible}) to show that $w(D(E_{1,1},X)) \leq 4$ and hence $w(B)\leq 6$, provided $r\geq 4$. If $r=3$ then the associated $X$ above may produce one of the exceptions listed in (\ref{supmat:eq:counterexamples}), yielding $w(B)\leq 7+2 =9$ instead.

    For the second statement, it suffices to choose any $g \in GL_r(2)$ such that $(g\otimes g^T)\cdot \tau_r$ has a nonzero entry in the $(i,j)$th diagonal position. Well
    \begin{eqnarray*}
        (e_i \otimes e_j)\cdot ((g\otimes g^T)\cdot \tau_r)\cdot (e_i \otimes e_j)^T & = & e_ige_j^T\otimes e_jg^Te_i^T \\
        & = & e_ige_j^T\otimes (e_ige_j^T)^T \\
        & = &  (g)_{i,j}.
    \end{eqnarray*}
    So it suffices to choose any invertible $g \in GL_r(2)$ with nonzero $(i,j)$th entry. A simple choice is the permutation matrix corresponding to the 2 cycle $(i,j)$ - note that such a choice minimises the weight of $(g\otimes g^T)\cdot \tau_r$ and thus has a `minimal' logical action (containing $S_{(i,j)}$).
\end{proof}

\begin{cor}\label{supmat:cor:invertible S circuit}
    Let $S^V$ be a logical $S$ circuit, where the $(a,b)$th entry of $V\in M_r(2)$ indicates an $S$ gate performed on the $(a,b)$th qubit. If $V$ is invertible then there exists a depth-1 generator that implements a logical diagonal Clifford operator containing $S^V$.
\end{cor}

\begin{table*}[h!]
    \centering
    \setlength\tabcolsep{5mm}
    {\renewcommand{\arraystretch}{1.2}

    \begin{tabular}{c|c}
        \textbf{Logical Gate} & \textbf{Depth bound} \\ \hline
        In-block diagonal generator $B = (g \otimes g^T)\tau_r$ for $g \in GL_r(2)$ & 1 \\
        Single logical $S$ gate $r \geq 4$ ($r=3$) & 6, $(9)$ \\
        $CZ_{i,j}$ for qubits $i,j$ in distinct or non-distinct code blocks & 4 \\
        Arbitrary in-block diagonal operator $r \geq 4$ ($r=3)$ & $r^2+5r+2$, $(r^2+8r+2)$\\
        Arbitrary cross-block CZ circuit on $b$ blocks & $(b-1)(r^2+r+4)$ \\
        Depth-1 cross-block CZ circuit on two code blocks & $7r+15$ \\
        Arbitrary $b$-block diagonal circuit $r\geq 4$ ($r=3$) & $br^2+r(b+4)+(4b-2)$, ($16b+25$) \\
    \end{tabular}}
    \caption{Depth bounds for logical diagonal operators in $SHYPS(r)$. Recall that a diagonal operator on $b$ blocks of $r^2$ logical qubits is determined (modulo Pauli) by a symmetric matrix $B \in SYM_{br^2}(2)$ - we give these matrices where appropriate. Note that $\tau_r$ is the permutation matrix exchanging qubit labels $(i,j) \longleftrightarrow (j,i)$.}
    \label{supmat:tab:diagonal_depth}
\end{table*}

\begin{table*}[h!]
    \centering
    \setlength\tabcolsep{5mm}
    {\renewcommand{\arraystretch}{1.2}
    \begin{tabular}{c|c}
        \textbf{Logical Gate} & \textbf{Depth bound} \\ \hline
       All qubit Hadamard (modulo logical SWAP) & 4 \\
       All qubit Hadamard & $3r+10$ \\
       Single qubit Hadamard gate $r \geq 4$ ($r=3$) & 8, (11) \\
       Arbitrary Hadamard circuit & $11r + 15$\\
       Arbitrary in-block permutation & $3r+6$ \\
       Arbitrary permutation & $36r^2 + 3r + 6$ 
    \end{tabular}}
    \caption{Depth of logical Hadamard-SWAP operators in $SHYPS(r)$}
    \label{supmat:tab:had_swap_depth}
\end{table*}

We are similarly able to produce low-depth implementations for isolated CZ gates in the SHYPS codes.
\begin{lem}
    Let $(a,b) \neq (c,d)$ be qubit positions in the logical qubit array of $SHYPS(r)$. Then there exists an implementation of $CZ_{(a,b),(c,d)}$ in physical depth 4, using 4 generators from $\calB$.
\end{lem}
\begin{proof}
    Let $S = D(E_{a,d},E_{c,b})$. Clearly S is symmetric and contains a single pair of off-diagonal entries. Furthermore, $(a,b) \neq (c,d)$ implies that
    \begin{eqnarray*}
        (e_a\otimes e_b)S(e_c \otimes e_d)^T & = & e_aE_{a,d}e_d^T \otimes e_bE_{b,c}e_c^T \\
        & = & 1.
    \end{eqnarray*}
    Hence the nonzero entries in $S$ correspond to the logical gate $CZ_{(a,b),(c,d)}$.
    Next observe that there exist row and column permutations $P$ and $Q^T$ such that the nonzero entries of $PE_{a,d}Q^T, PE_{c,b}Q^T$ lie below the diagonal. So in particular, adding the identity $I_r$ to these lower diagonal matrices is guaranteed to be invertible. The result then follows by \cref{supmat:lem:trick2} and \cref{supmat:lem:trick3}.
    \end{proof}
    Further results for low-depth implementations of specific diagonal operators are given in \cref{supmat:sec: Hadamard}. There we shall see that for a fixed set of logical qubits, the ability to perform \emph{any} logical diagonal on said qubits, relates to low-depth implementations of arbitrary Hadamard circuits.

\subsection{Multi-block diagonal Cliffords}
To generate logical diagonal Clifford circuits across multiple blocks of our chosen SHYPS code, we rely on the results of Section \ref{supmat:subsec:full cnot group gen} by Hadamard transforming cross-block CNOT operators. Furthermore, due to the tensor product form of the gauge generators in the SHYPS codes, an all-qubit logical Hadamard is particularly easy to implement.

\begin{lem}\label{supmat:lem:all qubit hadamard}
     In $SHYPS(r)$, the physical Hadamard-SWAP operator
     $H_1\cdots H_{n^2}\tau_{n_r}$ implements the logical gate
     $H_1\cdots H_{r^2}\tau_r$.
\end{lem}
\begin{proof}
    This is clear from the definitions of $G_X,L_X$ etc, and the fact that $(P\otimes G) \tau_{n_r}=\tau_r(G \otimes P).$
\end{proof}

In essence, the all-qubit logical Hadamard is implemented by an all-qubit physical Hadamard. As the following Lemma for cross-block CZ operators demonstrates, the additional SWAP circuits $\tau_n$ and $\tau_r$ can be easily accounted for by adjusting the automorphisms in our depth-1 CNOT generators. Thus cross-block CZ circuits may be implemented analogously to cross-block CNOT operators, by a low-depth sequence of transversal CZ circuits.

\begin{lem} \label{supmat:lem:2-block-cz-cost} Let $A \in M_{r^2}(2)$ and 
    \[
    U_{1,2}(A) = \prod_{i,j=1}^r CZ_{i,r^2+j}^{a_{i,j}}.
    \]
    be the corresponding cross-block circuit of logical CZ gates. Then $U_{1,2}(A)$ is implemented by a sequence of at most $r^2+r+4$ transversal physical CZ circuits.
    \end{lem}
\begin{proof}
First observe that  $U$ has symplectic matrix representation 
    \begin{eqnarray*}
          \left(
    \begin{array}{cc|cc}
        I & 0  & 0 & A  \\
        0 & I  & A^T & 0  \\ \hline
        0 & 0  & I & 0  \\
        0 & 0  & 0 & I  \\
            \end{array}
    \right)
    &=&    \left(
    \begin{array}{cc|cc}
        I & 0 & 0 & 0  \\
        0 & 0  & 0 & I  \\   \hline
        0 & 0  & I & 0  \\
        0 & I  & 0 & 0  \\
    \end{array}
    \right)
        \left(
    \begin{array}{cc|cc}
        I & A  & 0 & 0  \\
        0 & I  & 0 & 0  \\ \hline
        0 & 0  & I & 0  \\
        0 & 0  & A^T & I  \\
    \end{array}
    \right)\\ & \cdot &
    \left(
    \begin{array}{cc|cc}
        I & 0 & 0 & 0  \\
        0 & 0  & 0 & I  \\   \hline
        0 & 0  & I & 0  \\
        0 & I  & 0 & 0  \\
    \end{array}
    \right).
    \end{eqnarray*}
    Inserting two instances of $\tau_r^2=1$ then yields   \begin{multline}
    U_{1,2}(A) =
    \left(
    \begin{array}{cc|cc}
        I & 0 & 0 & 0  \\
        0 & 0  & 0 & I  \\   \hline
        0 & 0  & I & 0  \\
        0 & I  & 0 & 0  \\
    \end{array}
    \right)
    \left(
    \begin{array}{cc|cc}
        I & 0 & 0 & 0  \\
        0 & \tau_r  & 0 & 0  \\   \hline
        0 & 0  & I & 0  \\
        0 & 0  & 0 & \tau_r  \\
    \end{array}
    \right)
    \\ \cdot
    \left(
    \begin{array}{cc|cc}
        I &  A\cdot \tau_r  & 0 & 0  \\
        0 & I  & 0 & 0  \\ \hline
        0 & 0  & I & 0  \\
        0 & 0  & (A\cdot \tau_r)^T & I  \\
    \end{array}
    \right)
    \left(
    \begin{array}{cc|cc}
        I & 0 & 0 & 0  \\
        0 & \tau_r  & 0 & 0  \\   \hline
        0 & 0  & I & 0  \\
        0 & 0  & 0 & \tau_r  \\
    \end{array}
    \right)
    \\ \cdot
    \left(
    \begin{array}{cc|cc}
        I & 0 & 0 & 0  \\
        0 & 0  & 0 & I  \\   \hline
        0 & 0  & I & 0  \\
        0 & I  & 0 & 0  \\
    \end{array}
    \right).
    \end{multline}
So this is exactly a cross-block CNOT circuit given by off-diagonal matrix $A\cdot \tau_r$, conjugated by the logical Hadamard-SWAP operator $H^{\otimes r^2}\tau_r$ on the second code block. In circuit form:
    \[
\begin{tikzpicture}
\begin{yquantgroup}
  \registers{
         qubit {} q[2];
      }
      \circuit{
      slash q[0];
      slash q[1];
      h q[1];
      box {$\tau_r$} q[1];
      box {$A \tau_r$} q[1] | q[0];
      box {$\tau_r$} q[1];
      h q[1];
   }
\end{yquantgroup}
\end{tikzpicture}
\]
Combining \cref{supmat:thm:worst-case-cross-block-cnot} and
\cref{supmat:lem:all qubit hadamard}, there exists $a\leq r^2+r+4$ automorphisms $\pi_i = \sig_{i_1}\otimes \sig_{i_2}\in \aut{SHYPS(r)}$ such that $U_{1,2}(A)$ is implemented by the sequence of conjugated transversal CNOT operators 
\begin{eqnarray}
     (I_{n_r^2} \otimes H^{\otimes n_r^2}\tau_{n_r}) 
 & \cdot & \prod_{i=1}^a \prod_{j=1}^{n_r^2} CNOT_{j,\pi_i^{-1}(j)+n_r^2} \\ &\cdot & ( I_{n_r^2} \otimes \tau_{n_r} H^{\otimes n_r^2}), \label{supmat:eq:conjugated-cnot}
\end{eqnarray}

Now, conjugating by $\tau_{n_r}$ permutes the targets of the CNOT gates, while the Hadamard action transforms every physical CNOT into a physical CZ gate. Hence letting $\rho_i= (\sig_{i_1} \otimes \sig_{i_2})\tau_{n_r}$ we have that (\ref{supmat:eq:conjugated-cnot}) equals
\[
\prod_{i=1}^a \left( \prod_{j=1}^{n_r^2} CZ_{j,\rho_i^{-1}(j)+n_r^2} \right).
\]
In particular, each bracketed term above is a logical gate of $SHS(r)$ that is evidently depth-1 and transversal (and therefore fault-tolerant).
\end{proof} 
The following corollary is immediate by following the proof above, but inserting the improved depth bound \cref{supmat:lem:single-cnot} for isolated CNOT gates.
\begin{cor}
    Any single logical CZ operator between two SHYPS code blocks is implementable fault-tolerantly in depth at most 4.
\end{cor}

Generalising the above operator to $U_{i,j}(A_{i,j})$, it follows that an arbitrary $b$-block diagonal Clifford consisting of only cross-block CZ gates may be written
\[
\prod_{1\leq i < j \leq b} U_{i,j}(A_{i,j})= 
\left(
    \begin{array}{cccc|cccc}
        I &  &  &  & 0 & A_{1,1} & \cdots & A_{1,b} \\
         & I &  &   & A_{1,1}^T & 0 & \cdots & A_{2,b} \\
       &  & \ddots &   & \vdots & \vdots & \ddots & \vdots \\
         &  &  & I  & A_{1,b}^T & A_{2,b}^T & \cdots & 0 \\ \hline
         &  &  &  &I &  &  &  \\
         &  &  &   &  & I &  &  \\
         &  &  &   &  &  & \ddots & \\
         &  &  &  &  & &  & I \\
    \end{array}
    \right).
\]
We now cost this operator in a similar manner to \cref{supmat:sec: cnot gate characterisation}.

\begin{lem} \label{supmat:lem:multi-block-cz}
	For an even positive integer $b$, the diagonal operator $U = \prod_{1\leq i < j \leq b} U_{i,j}(A_{i,j})$ on $b$ blocks requires at most $b/2\cdot(b-1)$ block-to-block operators $U_{i,j}(A)$, implemented over at most $b-1$ rounds.
\end{lem}
\begin{proof}
	The number of operators required is clear and so we need only prove the time constraint. 
	Consider a graph with $b$ vertices and an edge between any pair of vertices $i,j$ (with $i<j$) for which $A_{i,j} \neq 0$. The chromatic index of this graph then corresponds to the number of rounds diagonal block-to-block operators required to implement $U$.
	Note that the chromatic index of this graph is upper bounded by that of a complete graph with $b$ vertices. The result then follows immediately from the fact that a complete graph with $b$ vertices has chromatic index $b-1$ for even $b$.
\end{proof}
Note that for odd $b$, one needs $b$ rounds rather than $b-1$, as the chromatic index of the complete graph with an odd number of vertices is $b$. The number of required rounds of odd $b$ is thus as if we were considering the next even number of even code blocks (i.e., $b+1$), but then leaving one unused.

The following Corollary is immediate by combining \cref{supmat:lem:multi-block-cz} and \cref{supmat:lem:2-block-cz-cost}
\begin{cor}\label{supmat:cor:cross-block-diagonal-depth}
For $b$ even, the diagonal operator $U = \prod_{1\leq i < j \leq b} U_{i,j}(A_{i,j})$ on $b$ blocks is implemented in depth at most $(b-1)(r^2+r+4)$.
\end{cor}

To complete this section on diagonal Clifford operators, we combine our work on in-block and cross-block operators in \cref{supmat:thm:overall-diagonal-cost}. Note that the results of this section are also presented in \cref{supmat:tab:diagonal_depth}, while \cref{supmat:sec: Hadamard} contains additional work on diagonal operators that are useful for compiling Hadamard circuits.

\begin{thm}\label{supmat:thm:overall-diagonal-cost}
   Any logical diagonal Clifford operator (modulo Paulis) on $b$ blocks of $SHYPS(r)$, is implemented in depth at most 
   \[br^2+(b+4)r+(4b-2), \,\, r\geq 4, \]
   \[16b+25 = br^2+(b+7)r+(4b-2), \,\, r=3,\]
   
   for $b$ even, and
   \[(b+1)r^2+(b+5)r+(4b+2), \,\,r\geq 4, \]
   \[16b+41 = (b+1)r^2+(b+8)r+(4b+2), \,\,r=3, \]
   for $b$ odd.
   
   Furthermore, such implementations require no auxiliary qubits, and thus have zero space cost.
\end{thm}
\begin{proof}
    As the subgroup of diagonal Clifford operators is abelian, we may implement an arbitrary operator as a sequence of parallel in-block gates, followed by any necessary cross-block gates. We therefore accrue an initial circuit of depth $r^2+5r+2$ (or $r^2+8r+2$ when $r=3)$ for in-block gates by \cref{supmat:thm:DiagGensType1Generate}, regardless of block count. This is then followed by additional depth of at most $(b-1)(r^2+r+4)$ for even $b$ and $b(r^2+r+4)$ for odd $b$ by \cref{supmat:cor:cross-block-diagonal-depth}.
    
    It's clear that implementing diagonal operators in SHYPS codes incurs zero space overhead: in-block diagonal operators are implemented directly using logical generators from $\calB$, whereas cross-block logical CZ gates mirrors the case of cross-block CNOT operators seen in \cref{supmat:sec: cnot gate characterisation}.
\end{proof}

\section{Hadamard-SWAP operators in SHYPS codes}\label{supmat:sec: Hadamard}
To complete the analysis of logical Clifford implementations in the SHYPS codes, we need lastly to characterize the Hadamard gates. However as \cref{supmat:lem:all qubit hadamard} suggests, it is natural to consider Hadamard circuits in conjunction with SWAP circuits as together these operators generate a subgroup of $Sp_{2k}(2)$ called the \emph{signed-symmetric group} or \emph{hyperoctahedral group}.

We find that logical permutations within a single SHYPS code block present significant depth savings compared to generic CNOT circuits ($O(r)$ versus $O(r^2)$). Moreover, in \cref{supmat:subsec:cross-block-perms} we show that the depth of multi-block permutations does not grow with the number of blocks $b$. Lastly, we examine arbitrary Hadamard circuits in \cref{supmat:subsec:hadamard-costs}. Here we combine our work on in-block permutations with certain depth-1 diagonal circuits, to implement Hadamard circuits in depth $O(r)$.

\subsection{In-block permutations}

First recall \cref{supmat:lem:depth 1 cnot circuits} that for any choice of $g_1,g_2 \in GL_{r}(2)$, the logical CNOT circuit $\left( \begin{matrix}
       I & g_1\otimes g_2 \\ 0 & I 
    \end{matrix} \right) 
    \in GL_{2r^2}(2)$ may be implemented by a depth-1 physical CNOT circuit. Moreover, these depth-1 circuits generate the full algebra of upper-right cross block CNOT operations
    \[ \{ \mat{I & A \\ 0 & I} \,:\,A \in M_{r^2}(2)\}.\]
To achieve the full algebra, we generically require $O(r^2)$ depth-1 generators but let's restrict our attention to certain permutations in $GL_{r^2}(2)$: For $\sig_1,\dots,\sig_r \in S_r$ we define
\begin{eqnarray*}
    A(\sig_i) :=\mat{ 
\sigma_1 & & & \\
 & \sigma_2 & &\\
 & & \ddots & \\
 & & & \sigma_r \\
}&=&
\sum_{i=1}^r E_{i,i} \otimes \sig_i,
\end{eqnarray*}
and then for an additional $\upsilon \in S_r,$ consider the product 

\begin{equation} \label{supmat:eq:generating-in-block-perm}
    A(\sig_i)\cdot (\upsilon \otimes I) \in GL_{r^2}(2).
\end{equation}

Now clearly these elements are permutations of the $r^2$ logical qubits, but moreover they form a subgroup isomorphic to $S_r^r \rtimes S_r \leq S_{r^2}$ - this can be observed by checking that
\begin{eqnarray*}
(\upsilon^{-1} \otimes I)(\sum_i E_{i,i}\otimes \sig_i)(\upsilon \otimes I) &=& \sum_i E_{\upsilon^{-1}(i),\upsilon^{-1}(i)}\otimes \sig_i \\
&=& \sum_i E_{i,i}\otimes \sig_{\upsilon(i)}.
\end{eqnarray*}
To understand exactly how this subgroup acts, first consider the action of $A(\sig_i)$:

\begin{eqnarray*}
    (\sum_{j,k} \alpha_{j,k} e_j \otimes e_k)A(\sig_i) & = & (\sum_{j,k} \alpha_{j,k} e_j \otimes e_k)\sum_{i=1}^r E_{i,i} \otimes \sig_i \\
    & = & \sum_{i=1}^r \sum_{k} \alpha_{i,k} e_i \otimes e_k\sig_i \\
    & = & \sum_{i=1}^r \sum_{k} \alpha_{i,k} e_i \otimes e_{\sig^{-1}_i(k)}.
\end{eqnarray*}

In particular for each chosen $\sig_i$ in $A$, there is a corresponding permutation action on the basis vectors $e_i \otimes e_1,\dots,e_i\otimes e_r$. In qubit terms, $A(\sigma_i)$ performs a permutation $\sig_i$ along each row of the $r\times r$ grid. Similarly, it's easy to check that the action of $(\upsilon \otimes I)$ on basis vectors induces a permutation of the rows of qubits.

Expanding the product (\ref{supmat:eq:generating-in-block-perm}) we see that
\begin{equation}\label{supmat:eq:generating-in-block-perm2}
    A(\sig_i)\cdot (\upsilon \otimes I) = \sum_{i=1}^r E_{i,\upsilon^{-1}(i)} \otimes \sig_i.
\end{equation}
Since $E_{i,\upsilon^{-1}(i)} \cdot \upsilon^{-1} = E_{i,i}$, it follows from \cref{supmat:lem:diagperm} that $E_{i,\upsilon^{-1}(i)} \cdot \upsilon^{-1} + \rho \in GL_r(2)$ for any $r$-cycle $\rho$ and hence $E_{i,\upsilon^{-1}(i)} + \rho\upsilon \in GL_r(2)$.
Collecting terms in \cref{supmat:eq:generating-in-block-perm2} then yields
\begin{eqnarray*}
    A(\sig_i)\cdot (\upsilon \otimes I) & = & \sum_{i} (E_{i,\upsilon^{-1}(i)}+\rho\upsilon) \otimes \sig_i \\
    & + &\rho\upsilon \otimes (\sum_{i}\sig_i).
\end{eqnarray*}
Finally, applying \cref{supmat:lem:sum of two} to the (possibly) non-invertible permutation sum in the second term provides a decomposition of $A(\sig_i)\cdot (\upsilon \otimes I)$ into at most $r+2$ matrices $g_1 \otimes g_2\in Gl_r(2)^2$. It follows that the cross-block CNOT circuit $\mat{I & A(\sig_i)\cdot (\upsilon \otimes I) \\ 0 & I}$ and the in-block CNOT circuit $\mat{A(\sig_i)\cdot (\upsilon \otimes I) & 0 \\ 0 & I}$ are implemented by $r+2$ depth-1 generators, respectively. Here for the in-block operator we have used the cross-block trick from \cref{supmat:lem:in-block CNOT}. 

Now in the analysis above, we could have instead chosen matrices that act on columns rather than rows, i.e., of the form: 
\[
B(\sig_i)\cdot (I \otimes \upsilon) :=
(\sum_{i=1}^r \sig_i \otimes E_{i,i})(I \otimes \upsilon).
\]
By the same arguments, these generate a subgroup isomorphic to $S_r^r \rtimes S_r$, implementable using at most $r+2$ depth-1 physical generators. Observe however that these two subgroups, called say $K_{\mathrm{row}}$ and $K_{\mathrm{col}}$, are distinct. They are also maximal subgroups of $S_{r^2}$ by the famous Onan-Scott Theorem \cite{aschbacher1985maximal}, and hence $S_{r^2}=\lrang{K_{\mathrm{row}},K_{\mathrm{col}}}$. I.e., these two particular collection of logical SWAP circuits, generate all permutations on a code block of $r^2$ logical qubits. In fact, one can show that $S_{r^2} = K_{\mathrm{row}} \cdot K_{\mathrm{col}} \cdot K_{\mathrm{row}}$ \footnote{A full discussion of this question is available on Math Stackexchange: Steve D (https://math.stackexchange.com/users/265452/steve-d), Diameter of $S_{n^2}$ with respect to two copies of $S_{n} \wr S_{n}$, URL (version: 2024-06-02): https://math.stackexchange.com/q/4926419}, yielding the following result:
\begin{prop}\label{supmat:prop:swaps in O(r)}
     Any logical permutation on a single $SHYPS(r)$ code block is implemented fault-tolerantly, using CNOT generators $\calA \cup \calA^T$, in depth at most $3(r+2)$.
\end{prop}

In-block permutations can be used to derive a useful result on cross-block CNOT circuits: \emph{depth-1} cross-block CNOT circuits can be implemented in $O(r)$ depth, rather than $O(r^2)$ as required for general cross-block CNOT circuits. 
\begin{cor}\label{supmat:lem:depth one cross-block CNOT}
    Any logical depth-1 cross-block CNOT circuit across two $SHS(r)$ code blocks can be implemented in depth no greater than $7r+15$, using CNOT generators $\calA \cup \calA^T$.
\end{cor}
\begin{proof}
    Take an arbitrary logical depth on cross-block CNOT circuit 
    \[ X = 
    \begin{pmatrix}
        I & A \\ 0 & I
    \end{pmatrix} \in GL_{2r^2}(2).
    \]
    Since $X$ is depth-1, the matrix $A \in M_{r^2}(2)$ has at most a single nonzero entry in each row and column. Therefore, there exists a permutation $P\in S_{r^2}$ such that $A P \in \diag_{r^2}(2)$. 
    We rewrite $X$ as
    \begin{equation}\label{supmat:eq:conjugate by perms}
        X = 
        \begin{pmatrix}
        I & 0 \\ 0 & P
        \end{pmatrix}
        \begin{pmatrix}
            I & AP \\ 0 & I
        \end{pmatrix}
        \begin{pmatrix}
            I & 0 \\ 0 & P^{-1}
        \end{pmatrix},
    \end{equation}
    and bound its depth by bounding the cost of each term.
    
    By \cref{supmat:prop:swaps in O(r)}, it follows that the in-block permutations $P$ and $P^{-1}$ can each be implemented in depth no greater than $3r+6$. 

    It remains to cost the cross-block CNOT operator given by the second term.
    We decompose the diagonal matrix $D:=AP$ as $D = \sum_{i=1}^r E_{i,i} \otimes D_i$ for some diagonal matrices $D_i \in \diag_r(2)$.
    For an $r$-cycle permutation $\sigma$, we define matrices $\rho_i = \sigma$ iff $D_i \neq I$ and $\rho_i = 0$ otherwise.
    \begin{align*}
        AP &= \sum_{i=1}^r E_{i,i} \otimes D_i \\
        &= \sum_{i=1}^r (E_{i,i} + \sigma) \otimes (D_i + \rho_i) 
        + \sum_{i=1}^r E_{i,i} \otimes \rho_i \\
        & \quad + \sigma \otimes \sum_{i=1}^r D_i
         + \sigma \otimes \sum_{i=1}^r \rho_i\\
        &= \sum_{i=1}^r (E_{i,i} + \sigma) \otimes (D_i + \rho_i)\\
        &\quad  + \big(\sum_{i \mid D_i \neq I} E_{i,i} + \sigma \big) \otimes \sigma
          + \sigma \otimes \big(\sum_{i=1}^r D_i + \sigma \big)\\
         & \quad + \sum_{i \mid D_i \neq I} \sigma \otimes \sigma
    \end{align*}
    By \cref{supmat:lem:diagperm}, it then immediately follows that $w(D) \leq r+3$ for any $D \in \diag_r(2)$. The cross-block CNOT operator in \eqref{supmat:eq:conjugate by perms} can thus be implemented using no more than $r+3$ depth-1 generators from $\calA$.
        
    We conclude that $X$ has a physical implementation with a depth no greater  $7r+15$. 
\end{proof}

A similar result exists for depth-1 cross-block CZ circuits.
\begin{cor}\label{supmat:lem:depth one cross-block CZ}
    Any logical depth-1 cross-block CZ circuit across two $SHYPS(r)$ code blocks can be implemented in depth no greater than $7r+15$.
\end{cor}
\begin{proof}
    Take any logical depth-1 cross-block CZ circuit $U$.
    First observe that $U$ has symplectic matrix representation 
    \begin{equation}
     \left(
    \begin{array}{cc|cc}
        I & 0  & 0 & A  \\
        0 & I  & A^T & 0  \\ \hline
        0 & 0  & I & 0  \\
        0 & 0  & 0 & I  \\
            \end{array}
    \right),
    \end{equation}
    where the matrix $A$ has at most one nonzero entry in each row or column.
    Similar to the strategy used in the proof of \cref{supmat:lem:2-block-cz-cost}, we first rewrite $U$ as a cross-block CNOT operator, conjugated by $H^{\otimes r^2}\tau_r$ on the second code block:
     \[
        \begin{tikzpicture}
        \begin{yquantgroup}
          \registers{
                 qubit {} q[2];
              }
              \circuit{
              slash q[0];
              slash q[1];
              h q[1];
              box {$\tau_r$} q[1];
              box {$A \tau_r$} q[1] | q[0];
              box {$\tau_r$} q[1];
              h q[1];
           }
        \end{yquantgroup}
        \end{tikzpicture}
    \]
    Note that since $A$ has rows and columns of weight at most one, so has $A\tau_r$. 
    We decompose the cross-block CNOT operator in the circuit above as we did in the proof of \cref{supmat:lem:depth one cross-block CNOT}:
    \[
        \begin{tikzpicture}
        \begin{yquantgroup}
          \registers{
                 qubit {} q[2];
              }
              \circuit{
              slash q[0];
              slash q[1];
              h q[1];
              box {$\tau_r$} q[1];
              box {$P$} q[1];
              box {$D$} q[1] | q[0];
              box {$P^T$} q[1];
              box {$\tau_r$} q[1];
              h q[1];
           }
        \end{yquantgroup}
        \end{tikzpicture}
    \]
    where $D := A\tau_r P$ and $P\in S_{r^2}$ is chosen such that $D \in \diag_{r^2}(2)$.
    Conjugating $P$ by $H^{\otimes r^2}\tau_r$ yields a different permutation $P'$, and hence we can rewrite the previous circuit as
    \[
        \begin{tikzpicture}
        \begin{yquantgroup}
          \registers{
                 qubit {} q[2];
              }
              \circuit{
              slash q[0];
              slash q[1];
              box {$P'$} q[1];
              h q[1];
              box {$\tau_r$} q[1];
              box {$D$} q[1] | q[0];
              box {$\tau_r$} q[1];
              h q[1];
              box {$P'^T$} q[1];
           }
        \end{yquantgroup}
        \end{tikzpicture}
    \]
    As shown in the proof of \cref{supmat:lem:depth one cross-block CNOT}, the cross-block CNOT operator in the circuit above can be implemented using at most $r+3$ depth-1 generators from $\calA$. 

    Following the same reasoning as in the proof of \cref{supmat:lem:2-block-cz-cost}, we then find that $U$ can be implemented using two logical in-block permutations and up to $r+3$ transversal physical CZ circuits. 
    Since, by \cref{supmat:prop:swaps in O(r)}, the in-block permutations $P'$ and $P'^T$ each require depth at most $2\cdot(3r+6)$.    
    We thus find that $U$ can be implemented in depth no greater than $7r+15$. 
\end{proof}

\subsection{Multi-block permutations}\label{supmat:subsec:cross-block-perms}
Next, we turn to logical multi-block permutations. 
One might initially expect that the depth of their physical implementation would scale with the total number of logical qubits (i.e., $br^2$ for $b$ code blocks of $SHYPS(r)$) like it does for, e.g., multi-block CZ circuits.
However, we will show below that general permutations can be done much more efficiently. In particular, they are implementable in a depth that does not scale with the number of code blocks, but simply with the number of logical qubits per block. 

The strategy to efficiently implement a general permutation is to decompose it as the product of two involutions (i.e., two depth-1 SWAP circuits). The logical SWAP gates in these involutions may then be scheduled efficiently using an edge-coloring algorithm, and the resulting number of required rounds does not scale with the number of code blocks but rather with the number of qubits per block. Since individual SWAP gates can be performed in depth $O(1)$, it then follows that multi-block permutations can be performed in detph $O(r^2)$. 

\begin{thm} \label{supmat:thm:cross-block-permutations}
		 A logical permutation operator on $b$ code blocks of $SHYPS(r)$ can be implemented in depth $36r^2 + 3r + 6$.
	\end{thm}
	\begin{proof}
		Take any permutation operator $P$ acting on $br^2$ qubits. We first write $P$ as the product of two involutions: $P = A\cdot B$. Each of these involutions can in turn be written as a product of an involution consisting of in-block SWAP operators, and one consisting of cross-block SWAP operators, denoted by $A_{\mathrm{IB}}$ and $A_{\mathrm{CB}}$, respectively.
		Note that in-block permutations normalize cross-block involutions. Hence, we find $P=A_{\mathrm{CB}}A_{\mathrm{IB}}B_{\mathrm{IB}}B_{\mathrm{CB}}$.
		
		It follows form \ref{supmat:prop:swaps in O(r)} that $A_{\mathrm{IB}}B_{\mathrm{IB}}$ can be implemented in depth $3r+6$
		It remains to determine the depth required to implement the cross-block components $A_{\mathrm{CB}}$ and $B_{\mathrm{CB}}$.
		
		Consider a multigraph $\mathcal{G}_A$ with $b$ vertices, corresponding to the $b$ code blocks, and an edge between two vertices for each two-cycle in $A_{\mathrm{CB}}$ swapping qubits between the corresponding code blocks (i.e., if $m$ qubits are swapped between a pair of code blocks, the corresponding vertices are connected by $m$ edges). 
		A theorem by Shannon \cite{shannon1949theorem} states that a proper edge coloring of a multigraph $\mathcal{G}$ with maximum vertex degree $\Delta(\mathcal{G})$ requires at most $3/2 \Delta(\mathcal{G})$ colors.
		Since $A$ is an involution, it immediately follows that $\Delta(\mathcal{G}_A) \leq r^2$. Therefore, $A_{\mathrm{CB}}$ can be implemented in at most $3/2 r^2$ rounds, each containing at most a single SWAP operator per code block. By \cref{supmat:lem:single-cnot} a single logical CNOT operator between code blocks is implementable in depth at most 4, and hence a single SWAP operator between code blocks is implementable in depth at most $12$. It follows that $A_{\mathrm{CB}}$ is implementable in depth at most $18 r^2$, and the same applies to $B_{\mathrm{CB}}$.
		
		We conclude that $P$ is implementable in depth at most $36r^2 + 3r + 6$.
	\end{proof}

\subsection{Hadamard circuits}\label{supmat:subsec:hadamard-costs}
We now turn to implementing arbitrary Hadamard operators in the SHYPS codes. 
Recall that SHYPS codes have a fold-transversal Hadamard operator $H^{\otimes n}\tau_{n_r}$, which implements the logical Hadamard-SWAP operator $H^{\otimes r^2} \tau_r$ (see \cref{supmat:lem:all qubit hadamard}).
We may then apply the result on in-block permutations from \cref{supmat:prop:swaps in O(r)} to $\tau_r$, to cost the isolated all-qubit logical Hadamard.
\begin{cor} \label{supmat:cor:transversal hadamard O(r) depth}
    The all-qubit logical Hadamard operator $H^{r^2}$ on $SHYPS(r)$ is implemented in depth at most $3r+10$.
\end{cor}
\begin{proof}
    By Lemma \ref{supmat:lem:all qubit hadamard}, $H^{\otimes r^2}$ is implemented by depth-1 all-qubit physical Hadamard $H^{\otimes n}$ gates, composed with $\tau_{n_r}$ and $\tau_r$. As $\tau_{n_r}$ is a physical SWAP circuit, it is implementable in a depth-3 CNOT circuit. Whereas $\tau_r \in S_{r^2}$ requires depth at most $3r+6$ by \cref{supmat:prop:swaps in O(r)}.
\end{proof}

As for the fault-tolerance of this implementation, we recall from the discussion in \cref{supmat:lem:DiagGenType1Distance} that $\tau_n$ only exchanges qubits lying in different rows and columns of the physical qubit array. Thus a single gate failure during implementation will not lead to a weight two Pauli in the support of a single logical operator, and the circuit is therefore distance preserving. This guarantees the fault-tolerance of $H^{\otimes r^2}$. We note that for qubit architectures with high-connectivity, such a swap may be implemented in practice by simple qubit relabelling. This however leads to only modest constant savings in depth.

Combined with CNOT operators and diagonal gates, the all-qubit Hadamard is sufficient to generate all Clifford operators. However for practical applications it is desirable to more accurately cost $H^V$ for $V\in M_r(2)$, where as before, $V_{i,j}=1$ indicates a Hadamard gate on qubit $(i,j)$. 

To implement such arbitrary Hadamard circuits, the following circuit identities are useful
\begin{lem}\label{supmat:lem:single block H}
    Let $V\in M_r(2)$. Then the arbitrary Hadamard circuit $H^V$ is implemented by repeated application of $H^{\otimes r^2}$ and $S^{V}$.
\end{lem}
\begin{proof}
    It's easy to check that $(S\cdot H)^3 = (1+\i)/\sqrt{2}\cdot I_2$, i.e., the identity up to global phase. Hence because $H$ has order 2,
    \[
    (S^{V}\cdot H^{\otimes r^2})^2S^{V} = H^V.
    \]
\end{proof}
\begin{lem}
    Let $V \in M_r(2)$ with $V_{i,j}=1$ for an even number of entries. Then the arbitrary Hadamard circuit $H^V$ is implemented by repeated application of $H^{\otimes r^2}$, any depth-1 circuit of CZ gates with support strictly on qubits $\{(i,j) \mid V_{i,j}=1 \}$, and a qubit permutation.
\end{lem}
\begin{proof}
    It is readily verified that the following circuit identity holds:
    \begin{equation}\label{supmat:eq:CZ H SWAP}
        (CZ_{1,2} \cdot (H \otimes H))^3 = SWAP_{1,2}.
    \end{equation}
    Since $H$ has order two, one then finds
    \[
        \bigg( \big( \prod_{\{a,b\} \in \mathcal{P}} CZ_{a,b} \big) \cdot H^{\otimes r^2}  \bigg)^2 \cdot \big( \prod_{\{a,b\} \in \mathcal{P}} CZ_{a,b} \cdot SWAP_{a,b}\big) = H^V,
    \]
    where $\mathcal{P}$ is any partition of the set $\{(i,j) \mid V_{i,j}=1 \}$ into subsets of size 2. The lemma now follows from the fact that $ \prod_{\{a,b\} \in \mathcal{P}} CZ_{a,b} $ is a depth-1 circuit.
\end{proof}
These two lemmas can be combined into the following statement:
\begin{cor}\label{supmat:cor:Hadamard using CZ}
    Let $V \in M_r(2)$. Then $H^V$ is implemented by repeated application of $H^{\otimes r^2}$, any depth-1 circuit of diagonal gates with support strictly on qubits $\{(i,j) \mid V_{i,j}=1 \}$, and a qubit permutation.
\end{cor}
	
We now introduce a particular set of depth-1 (logical) diagonal Clifford circuits on $r^2$ qubits:

\[
    \Xi := \{ \begin{pmatrix} I & D \tau_r \\ 0 & I	\end{pmatrix} \in Sp_{2r^2} \mid D \in \diag_{r^2}(2),\, D \cdot \tau_r \in SYM_{r^2}(2)\}.
\]
For convenience, we denote the set of diagonal matrices appearing in the definition above by 
\[ 
    \xi := \{ D \in \diag_{r^2}(2) \mid D \cdot \tau_r \in SYM_{r^2}(2)\}.
\]

\begin{lem}\label{supmat:lem:support Xi}
    For integer $0\leq s \leq r^2$, the set $\Xi$ contains an operator with support on exactly $s$ qubits.
\end{lem}
\begin{proof}
    Recall that the right action of $\tau_r$ on any matrix in $M_{r^2}(2)$ is to swap the columns $(i,j)$ and $(j,i)$ for $1\leq i \neq j \leq r$, while leaving columns $(i,i)$ unchanged. 
    Hence, the following set forms a basis of $\xi$:
    \begin{multline*}
        \{E_{(i,i),(i,i)} \mid 1\leq i \leq r\} \\
        \cup \{E_{(i,j),(i,j)} + E_{(j,i),(j,i)} \mid 1\leq i\neq j \leq r \}.
    \end{multline*}
    These basis elements of $\xi$ correspond to diagonal operators $S_{(i,i)}$ and $CZ_{(i,j),(j,i)}$ in $\Xi$, respectively. The two subsets contain $r$ and $(r^2-r)/2$ elements, respectively.
    The lemma now follows by observing that any integer $0\leq s \leq r^2$ can be written as a sum $ a + 2b$ for some $0\leq a \leq r$ and $0 \leq b \leq (r^2-r)/2 $.
\end{proof}

Using circuits from $\Xi$ in conjunction with Corollary \ref{supmat:cor:Hadamard using CZ} allows us to execute a Hadamard circuit on any number of qubits inside a $SHYPS(r)$ code block, but we are not free to choose \emph{which} qubits. Conjugating these circuits by a permutation operator, however, allows for the implementation of any arbitrary Hadamard circuit $H^V$ for $V\in M_r(2)$.
Given that both the all-qubit logical Hadamard operator and all logical permutations within a $SHYPS(r)$ code block can be implemented in depth $O(r)$, a depth-$O(r)$ implementation of all logical depth-1 diagonal operators in the aforementioned set would guarantee that all logical Hadamard circuits are implementable in depth $O(r)$ as well. 

We first prove a small lemma on a generating set for $\xi$.
\begin{lem}\label{supmat:lem:generating xi with diagonals}
    The set of matrices $\xi$ is generated (under addition) by matrices of the form
    \[
        D \otimes D \quad \text{ where } D \in \diag_r(2).
    \]
\end{lem}
\begin{proof}
    One can readily verify that $(D \otimes D)\tau_r$ is symmetric for any diagonal matrix $D \in  M_r(2) $. 
    Since the set $\xi$ is closed under addition, it suffices to show that we can generate all elements of some basis.
    Recall that the following set forms a basis of $\xi$:
    \begin{multline*}
        \{E_{(i,i),(i,i)} \mid 1\leq i \leq r\} \\
        \cup \{E_{(i,j),(i,j)} + E_{(j,i),(j,i)} \mid 1\leq i\neq j \leq r \}.
    \end{multline*}
    The lemma follows by observing that $E_{(i,i),(i,i)} = E_{i,i} \otimes E_{i,i}$, 
    and $E_{(i,j),(i,j)} + E_{(j,i),(j,i)}  = (E_{i,i} + E_{j,j}) \otimes (E_{i,i} + E_{j,j}) + E_{i,i} \otimes E_{i,i} +  E_{j,j} \otimes  E_{j,j}$. 
\end{proof}

\begin{thm}\label{supmat:thm:weight xi}
    Let $D \in \xi$. Then there exist $A_1,\dots, A_p \in GL_r(2)$ for some $p\leq 5r+1$, such that
    \[
        D = \sum_{i=1}^p (A_i\otimes A_i^T).
    \]
\end{thm}

\begin{proof}
    Using \cref{supmat:lem:generating xi with diagonals}, we first decompose $D\in\xi$ as
    \begin{equation} \label{supmat:eq:diagonal-tensor-product}
        D = \sum_{i=1}^a D_i \otimes D_i
    \end{equation}
        
    for $D_i \in \diag_r(2)$ and $a \leq 2^r-1$.
    One can then proceed in a manner similar to the proof of \cref{supmat:thm:diagonal-decomp-invertible}, starting from equation \ref{supmat:eq:expression1} (note that we would only retain the second term in that expression).
    In particular, we decompose each diagonal matrix $D_i$ in the basis $\{E_{j,j}\mid 1\leq j \leq r\}$:
    \[
        D = \sum_{i=1}^a \sum_{j=1}^r d_{i,j} E_{j,j} \otimes \sum_{k=1}^r d_{i,k} E_{k,k},
    \]
    and then collect the terms as follows
    \[
        D = \sum_{j=1}^r e_j E_{j,j} \otimes E_{j,j} + \sum_{j=1}^r f_j (E_{j,j} \otimes B_j + B_j \otimes E_{j,j}),
    \]
    for some $e_j, f_j\in \gf$ and $B_j \in \diag_r(2)$. Provided $r > 3$, one can then invoke \cref{supmat:lem:diagperm}, \cref{supmat:lem:trick1} and \cref{supmat:lem:trick3} as in the proof of \cref{supmat:thm:diagonal-decomp-invertible} to obtain that $w(D \tau)\leq 5r+1$. When $r=3$ we check computationally that $w(D\tau) \leq 12$ for all $D$ of the form (\ref{supmat:eq:diagonal-tensor-product}), completing the proof.
\end{proof}
The bound on the weight of the diagonal matrices in $\xi$ specified in \cref{supmat:thm:weight xi} implies the following statement when combined with \cref{supmat:cor:diagonal generators}:
\begin{cor}\label{supmat:cor:depth Xi}
    The logical depth-1 diagonal circuits in $\Xi$ on code $SHYPS(r)$ may be implemented by a physical circuit of depth at most $5r+1$.
\end{cor}

Finally, we can combine the result above with those on in-block logical permutations and the circuit identity in \cref{supmat:cor:Hadamard using CZ} to obtain a $O(r)$ depth scaling for arbitrary Hadamard circuits.
\begin{thm}\label{supmat:thm:depth arbitrary Hadamard}
    Let $V \in M_r(2)$. Then the arbitrary logical Hadamard circuit $H^V$ on code $SHYPS(r)$ may be implemented, up to a logical Pauli correction, by a physical circuit of depth no greater than $11r + 15$.
\end{thm}
\begin{proof}
    Using \cref{supmat:cor:Hadamard using CZ}, one can write 
	\[
		H^V = S_V \cdot H^{\otimes r^2} \cdot S_V \cdot H^{\otimes r^2} \cdot S_V \cdot \Lambda_V\,,
	\]
	were $S_V$ is a logical depth-1 diagonal circuit with support strictly on qubits $\{(i,j) \mid V_{i,j}=1 \}$, and $\Lambda_V$ is a depth-1 logical SWAP circuit (with support only within $V$).

	We first focus on the special case $V \in SYM_r(2)$, where we can choose $S_V\in \Xi$. Furthermore, up to a possible Pauli-$Z$ correction, we may replace the first and last instance of $S_V$ by $\tau_r$.
    In particular, one has 
    \[
		H^V = \tau_r \cdot H^{\otimes r^2} \cdot S_V \cdot H^{\otimes r^2} \cdot \tau_r \cdot P_Z \cdot \Lambda_V\,,
	\]
    where $P_Z = \prod_{i\mid V_{i,i}=0} Z_{(i,i)}$.
    The depth of $S_V$ and $\Lambda_V$ follows from \cref{supmat:cor:depth Xi}, and \cref{supmat:prop:swaps in O(r)}, respectively. Note that $H^{\otimes r^2} \cdot S_V \cdot H^{\otimes r^2}$ is an $X$-diagonal operator which can be executed in the same depth as $S_V$. Hence, we find that up to a a logical Pauli correction the total depth is upper bounded by $1 + (5r+1) + 1 + 3r+6 = 8r + 9 $.

    To obtain the cost for $V \in M_r(2) \backslash SYM_r(2)$, all we need to do is find a $V'\in SYM_r(2)$ with the same number of nonzero entries. $H^{V}$ and $H^{V'}$ then have the same weight, and can be mapped onto one another by conjugation with an appropriate logical qubit permutation, this is indeed always possible by \cref{supmat:lem:support Xi}.
    In particular, for any such $V$, there exists a $V'\in SYM_r(2)$ and a permutation $P\in S_{r^2}$ such that $fl(V) = fl(V')\cdot P$. One then has $H^{V} = P H^{V'} P^{-1}$.
    We may combine the $\Lambda_{V'}$ and $P^{-1}$ into a single permutation, which results in total depth (up to a logical Pauli correctoin) of $11r + 15$.
\end{proof}

We end this section by noting that a single logical Hadamard operator can be executed at a constant cost. To see this, recall that (up to a global phase) $ S_i\cdot ( H_i \cdot S_i \cdot H_i) \cdot S_i = H_i$. Note that we can replace the initial and final $S_i$ operators by some logical depth-1 diagonal operator $D_i$ that contains $S_i$\footnote{Note that this may require a Pauli $Z$ correction.}. 
We may then use \cref{supmat:lem:single S} to determine the total cost, noting that $H_i \cdot S_i \cdot H_i$, as an $X$-diagonal operator, can be implemented with the same depth as $S_i$.
\begin{cor}
	A single logical qubit Hadamard is implementable in depth $8$ ($11$ when $r = 3$).
\end{cor}

\section{SHYPS compiling summary}\label{supmat:sec:SHSCompSummary}
In this section we present two novel Clifford decompositions, which we then use to synthesize arbitrary Clifford operators in terms of the logical generators presented in this paper. 
Contrary to other known Clifford decompositions, those introduced below do not contain any Hadamard gates. Instead, the first one contains exclusively $X-$ and $Z$- diagonal gates, while the second one contains $X-$ and $Z$- diagonal gates and a CNOT circuit.
For codes where arbitrary Hadamard circuits are more expensive than depth-1 diagonal gates, as is the case for SHYPS codes, these decompositions can be advantageous for minimizing the total depth when synthesizing a Clifford circuit. To highlight the significance of the worst-case logical Clifford depth in SHYPS obtained with these decompositions, we conclude this section with a brief comparison of the space-time volume of SHYPS against that of surface codes implementing transversal gates or lattice surgery.

\subsection{Clifford decompositions}

\begin{thm}\label{supmat:thm:clifford decomposition_1}
	Any Clifford operator $C \in \mathcal{C}_n / \mathcal{P}_n$, can be written as the product
	\begin{equation}\label{supmat:eq:decomp_1}
		C = DZ \cdot DX \cdot DZ' \cdot DX' \cdot DZ(1)\,,
	\end{equation}
	where 
	\begin{itemize}
		\item $DZ$ and $DZ'$ are $Z$-diagonal operators,
		\item $DX$ and $DX'$ are $X$-diagonal operators, 
		\item $DZ(1)$ is a depth-1 $Z$-diagonal operator.
	\end{itemize}
\end{thm}
\begin{rem*}
	One can obtain three related decompositions by applying \eqref{supmat:eq:decomp_1} to $C^{-1}$, $H^{\otimes n} C H^{\otimes n}$, and $H^{\otimes n} C^{-1} H^{\otimes n}$, resulting in decompositions of the form $C = DZ(1) \cdot DX \cdot DZ' \cdot DX' \cdot DZ $, $C = DX \cdot DZ \cdot DX' \cdot DZ' \cdot DX(1) $, and $C = DX(1) \cdot DZ \cdot DX' \cdot DZ' \cdot DX $, respectively.
\end{rem*}

\begin{thm}\label{supmat:thm:clifford decomposition_2}
	Any Clifford operator $C \in \mathcal{C}_n / \mathcal{P}_n$, can be written as the product
	\begin{equation}\label{supmat:eq:decomp_2}
		C = DZ \cdot CX \cdot DX \cdot DZ(1)\,,
	\end{equation}
	where 
	\begin{itemize}
		\item $DZ$ is a $Z$-diagonal operator,
		\item $CX$ is a CNOT operator,
		\item $DX$ is an $X$-diagonal operator, 
		\item $DZ(1)$ is a depth-1 $Z$-diagonal operator.
	\end{itemize}
\end{thm}
\begin{rem*}
	We can again obtain additional related decompositions by applying \eqref{supmat:eq:decomp_2} to  $C^{-1}$, $H^{\otimes n} C H^{\otimes n}$, and $H^{\otimes n} C^{-1} H^{\otimes n}$. In addition, since CNOT normalizes both the group of $Z$-diagonal operators and that of $X$-diagonal operators, one can change the order of the consituents by moving $CX$ above through either $DZ$ or $DX$. 
	In total 12 different decompositions of this kind can be obtained by combining these two methods.
\end{rem*}

Before proving both theorems above, we must prove a number of supporting lemmas. 
We first define the following subgroups of $Sp_{2n}(2)$:
\begin{align*}
	\mathcal{Z}_n = \left\{\begin{bmatrix} I_n & M \\ 0 & I_n\end{bmatrix}~;~ M\in M_n(2) \,,\, M=M^T\right\},\\
	\mathcal{X}_n = \left\{\begin{bmatrix} I_n & 0 \\ M & I_n\end{bmatrix}~;~ M\in M_n(2) \,,\, M=M^T\right\}.
\end{align*}
Note that these are precisely the symplectic representations of $Z$- and $X$-diagonal Clifford operators, respectively. 
We also define the following subsets:
\begin{align*}
	\mathcal{Z}_n(1) &= \left\{\begin{bmatrix} I_n & M \\ 0 & I_n\end{bmatrix}~;~M=M^T~\text{with}~\text{weight}(M)\leq 1\right\},\\
	\mathcal{X}_n(1) &= \left\{\begin{bmatrix} I_n & 0 \\ M & I_n\end{bmatrix}~;~M=M^T~\text{with}~\text{weight}(M)\leq 1\right\};
\end{align*}
which correspond to \emph{depth-1} $Z$- and $X$-diagonal Clifford operators, respectively.

\begin{lem}\label{supmat:lem:make fist quadrant invertible}
	Let $\chi = \begin{pmatrix} A & B\\C & D\end{pmatrix} \in Sp_{2n}(2)$. Then there exists $\zeta \in \mathcal{Z}_n(1)$ such that $\zeta \cdot \chi$ has invertible top-left quadrant.
\end{lem}
\begin{proof}
	Denoting the rank of $A\in M_n(2)$ by $k$,
	we may right-multiply $\chi$ with a symplectic matrix of the form $\begin{pmatrix} U & 0\\0 & U^{-T}\end{pmatrix}$ to perform Gaussian elimination on the first $n$ columns, yielding 
	\[
	\chi \cdot \begin{pmatrix} U & 0 \\0 & U^{-T}\end{pmatrix}
	= 
	\begin{pmatrix} 
		\begin{matrix} A_1 & 0\; \\ A_2 & 0\;  \end{matrix} & B'\; \\
		\begin{matrix} C_1 & C_3 \\ C_2 & C_4  \end{matrix}  & D'\;
	\end{pmatrix},
	\]
	where $A_1 \in M_k(2)$ and $\begin{pmatrix} A_1 \\ A_2 \end{pmatrix}$ has rank $k$.
	Note that since $\begin{pmatrix} A \\ C \end{pmatrix}$ has rank $n$, $\begin{pmatrix} C_3 \\ C_4 \end{pmatrix}$ necessarily has rank $n-k$.
	We can thus perform further Gaussian elimination on the latter $n-k$ columns. In particular, there exists a $V\in GL_{n-k}$ and a permutation $P\in S_n$ such that
	\[
	\begin{pmatrix} C_3 \\ C_4 \end{pmatrix} \cdot V =: P \cdot \begin{pmatrix} C'_3 \\ I_{n-k} \end{pmatrix}.
	\]
	For $U'= \begin{pmatrix} I_k & 0 \\0 & V \end{pmatrix} U$, we then have
	\[
	\begin{pmatrix} P^{-1} & 0 \\0 & P^{-1}\end{pmatrix} \cdot \chi \cdot \begin{pmatrix} U' & 0 \\0 & U'^{-T}\end{pmatrix}
	= 
	\begin{pmatrix} 
		\begin{matrix} A'_1 & 0\; \\ A'_2 & 0\;  \end{matrix} & B''\; \\
		\begin{matrix} C'_1 & C'_3 \\ C'_2 & I_{n-k}  \end{matrix}  & D''\;
	\end{pmatrix}.
	\]
	The symplectic condition for the matrix above implies that $C'^T_3 A'_1 = A'_2$, which in turn implies that $A'_1$ has rank $k$ because $\rank \begin{pmatrix} A'_1 \\ A'_2 \end{pmatrix} = k$.
	It follows that the matrix $\begin{pmatrix} A'_1 & 0 \\ A'_2 + C'_2 & I_{n-k} \end{pmatrix}$ is full rank, and therefore the matrix
	\[
	\begin{pmatrix} I & \diag(0,I_{n-k})\\ 0 & I\end{pmatrix} \begin{pmatrix} P^{-1} & 0 \\0 & P^{-1}\end{pmatrix} \cdot \chi \cdot \begin{pmatrix} U' & 0 \\0 & U'^{-T}\end{pmatrix}
	\]
	has an invertible top-left quadrant. 
	Since multiplying this matrix from the left with $\diag(P,P)$ and from the right with $\diag(U'^{-1}, U'^T)$ does not change that, we conclude that 
	\[ 
	\begin{pmatrix} I & K\\0 & I\end{pmatrix} \cdot \chi
	\]
	with $K = P \cdot \diag(0,I_{n-k}) \cdot P^{-1}$ has an invertible top-left quadrant. $K$ is symmetric and has columns with at most a single nonzero entry, hence $\zeta := \begin{pmatrix} I & K\\0 & I\end{pmatrix} \in \mathcal{Z}_n(1) $. 
\end{proof}
Note that the symmetric matrix $K$ found in the proof above is always diagonal and therefore the symplectic matrix $\zeta$ in the lemma corresponds to a circuit of S gates. However, for any involution $J \in S_n$ obeying $J K J = K$, the matrix $K'= K J$ is a valid solution too. 
With this choice for the off-diagonal block, $\zeta$ is the symplectic representation of a depth-1 circuit containing CZ gates. 
Specifically, replacing $K$ by $K J$ corresponds to taking pairs (determined by the two-cycles that constitute the involution permutation $J$) of S gates appearing in $\zeta$ and replace those by a CZ gate between the affected qubits.  
Before moving on to the second lemma, we first prove the following technical result:
\begin{prop}
	\label{supmat:prop:prodsym}
	Every square matrix over a field $K$ is the product of two symmetric matrices over $K$.
\end{prop}
The proof of this proposition relies on the fact that every square matrix $M$ over a field $K$ admits and is similar to its rational canonical form $\Lambda$. Both $\Lambda$ and the similarity transformation are themselves matrices over $K$. Additionally, two matrices are similar if and only if they admit the same rational canonical form.
\begin{proof}
	Let $M$ be a square matrix over $K$. $M$ admits a rational canonical form $\Lambda$ such that
	\[
	M = S \Lambda S^{-1}
	\]
	with both $S$ and $\Lambda$ square matrices over $K$ and $S$ invertible. We write $\Lambda$ as
	\[
	\Lambda = \begin{bmatrix}
		C_1 & & &\\
		& C_2 & &\\
		& & \ddots &\\
		& & & C_k
	\end{bmatrix}
	\]
	for each $C_i$ a companion matrix for the invariant factor $f_i$ of $M$. Explicitly, for
	\[
	f_i = c_0 + c_1 X + \cdots + c_{d-1} X^{d-1} + X^{d}
	\]
	we have
	\[
	C_i = \begin{bmatrix}
		& & & -c_0\\
		1 & & & -c_1\\
		& \ddots & & \vdots\\
		& & 1 & -c_{d-1}
	\end{bmatrix}.
	\]
	Note that by explicit computation,
	\[
	C_i \cdot \begin{bmatrix}
		c_1 & c_2 & \dots & c_{d-1} & 1\\
		c_2 & c_3 & \iddots & 1 &\\
		\vdots & \iddots & \iddots  &&\\
		c_{d-1} & 1 & && \\
		1 & & & &
	\end{bmatrix} = \begin{bmatrix}
		-c_0 &  & & & \\
		& c_2 & \dots & c_{d-1} & 1\\
		& \vdots & \iddots  & \iddots &\\
		& c_{d-1} & \iddots && \\
		& 1 & & &
	\end{bmatrix}.
	\]
	As the matrix right-multiplying $C_i$ is always invertible, we must have
	\begin{align}
		C_i = \begin{bmatrix}
			-c_0 &  & & & \\
			& c_2 & \dots & c_{d-1} & 1\\
			& \vdots & \iddots  & \iddots &\\
			& c_{d-1} & \iddots && \\
			& 1 & & &
		\end{bmatrix}\cdot \begin{bmatrix}
			c_1 & c_2 & \dots & c_{d-1} & 1\\
			c_2 & c_3 & \iddots & 1 &\\
			\vdots & \iddots & \iddots  &&\\
			c_{d-1} & 1 & && \\
			1 & & & &
		\end{bmatrix}^{-1}.
		\label{eq:prodoftwo}
	\end{align}
	For any symmetric invertible matrix $A$
	\[
	A = A^T \iff A^{-1} = (A^T)^{-1} \iff A^{-1} = (A^{-1})^T
	\]
	so that $A^{-1}$ is also symmetric. Therefore, \cref{eq:prodoftwo} is a decomposition of $C_i$ as a product of two symmetric matrices. Labelling these as $U_i,V_i$ so that $C_i = U_i V_i$, we have a decomposition of $\Lambda$ as a product of two symmetric matrices by
	\[
	\Lambda = U V=  \begin{bmatrix}
		U_1 & & \\
		& \ddots & \\
		& & U_k
	\end{bmatrix}\begin{bmatrix}
		V_1 & & \\
		& \ddots & \\
		& & V_k
	\end{bmatrix}.
	\]
	Finally, note that
	\begin{multline*}
	    M = S \Lambda S^{-1} = S U V S^{-1} = S U S^T (S^T)^{-1} V S^{-1} \\
     = (S U S^T)((S^{-1})^T V S^{-1})
	\end{multline*}
	where the two terms in parentheses are explicitly symmetric matrices.
\end{proof}

We are now equipped to prove the following lemma:
\begin{lem}
	\label{supmat:lem:XZXZ}
	Let $\chi = \begin{pmatrix} A & B\\C & D\end{pmatrix} \in Sp_{2n}(2)$ with invertible top-left quadrant $A\in GL_n(2)$. Then there exist symplectic matrices $\alpha, \gamma \in \mathcal{X}_n$ and $\beta, \delta \in \mathcal{Z}_n$ such that $\chi = \alpha \cdot \beta \cdot \gamma \cdot \delta$.
\end{lem}

\begin{proof}
	Substituting the explicit definitions for $ \mathcal{X}_n$ and $ \mathcal{Z}_n$, one finds that  $\alpha \cdot \beta \cdot \gamma \cdot \delta$ takes the form
	\[
	\begin{bmatrix}
		I_n & 0\\
		L & I_n
	\end{bmatrix}\begin{bmatrix}
		I_n & M\\
		0 & I_n
	\end{bmatrix}\begin{bmatrix}
		I_n & 0\\
		N & I_n
	\end{bmatrix}\begin{bmatrix}
		I_n & P\\
		0 & I_n
	\end{bmatrix} 
	= \begin{bmatrix}
		A' & B'\\
		C' & D'
	\end{bmatrix}
	\]
	where
	\begin{align*}
		A' &= I_n + M N\\
		B' &= M + P + MNP\\
		C' &= L + N + LMN\\
		D' &= I_n + LM + NP + LP + LMNP
	\end{align*}
	and $L,M,N,P$ are symmetric $n\times n$ binary matrices. 
	Take any $\chi$ with a block form as stated in the theorem, we must then proof that $L,M,N,P$ exist such that $A=A'$, $B=B'$, $C=C'$, and $D=D'$.

	By \cref{supmat:prop:prodsym}, we can always find symmetric matrices $M,N  \in M_n(2)$ such that
	\[
	M N = A + I_n
	\]
	and so we can always satisfy the upper left quadrant constraint. Since $A$ is invertible, we can solve for $L$ and $P$ as
	\begin{align*}
		L (I_n + M N) + N = C & \implies L A + N = C\\
		&\implies L = (C + N)A^{-1}\\
		(I_n + M N)P + M = B & \implies A P + M = B \\
		&\implies P = A^{-1}(B + M).
	\end{align*}
	Since $M,N$ are currently fixed, we do not necessarily satisfy our condition that $L,P$ are symmetric. Explicitly checking $L$, we have
	\[
	L^T = A^{-T}(C+N)^T.
	\]
	Since $\chi$ is a symplectic matrix, we have
	\[
	\begin{bmatrix}
		A & B\\
		C & D
	\end{bmatrix}^T \begin{bmatrix}
		0 & I_n\\
		I_n & 0
	\end{bmatrix}\begin{bmatrix}
		A & B\\
		C & D
	\end{bmatrix}=
	\begin{bmatrix}
		0 & I_n\\
		I_n & 0
	\end{bmatrix}
	\]
	which in turn implies
	\[
	A^T C + C^T A = 0 \implies A^{-T} C^T = C A^{-1}.
	\]
	Expanding the expression for $L^T$ and noting that $N$ has already been fixed as symmetric, we thus have
	\begin{align*}
		L^T &= A^{-T} C^T + A^{-T} N^T\\
		&= C A^{-1} + A^{-T} N.
	\end{align*}
	Comparing against the expanded expression for $L$, we see that $L = L^T$ if and only if $N A^{-1} = A^{-T} N$. But $N A^{-1} = A^{-T} N$ if and only if $A^T N = N A$. Explicitly checking, we have
	\begin{multline*}
	    A^T N = (I_n + M N)^T N = (I_n + N^T M^T)N \\
     = N + N M N = N(I_n + M N) = N A
	\end{multline*}	
	and we conclude that $L$ is symmetric as desired given any $M,N$ pair that satisfied the first constraint. An analogous sequence of arguments proves $P$ is also explicitly symmetric when computed with the given expression. Finally, the expressions for $L,P$ can be substituted into our check for the final quadrant to yield
	\begin{align*}
		&I_n + LM + NP + LP + LMNP\\
		&= I_n + (C + N)A^{-1} M + N A^{-1}(B + M)\\
		&\qquad+ (C + N)A^{-1}(I_n + M N)A^{-1}(B + M)\\
		&=I_n + C A^{-1}M + N A^{-1} B + (C + N)A^{-1}(A)A^{-1}(B + M)\\
		&=I_n + C A^{-1}M + N A^{-1} B + (C + N)A^{-1}(B + M)\\
		&=I_n + C A^{-1} B + N A^{-1} M.
	\end{align*}
	Again, due to the symplectic structure of $\chi$, we remark that $C^T B = A^T D + I_n$. Recalling that the symplectic structure implies $C A^{-1} = A^{-T} C^T$, we can simplify further:
	\begin{align*}
		I_n + C A^{-1} B + N A^{-1} M &= I_n + A^{-T} C^T B + N A^{-1} M\\
		&= I_n + A^{-T}(A^T D + I_n) + N A^{-1} M\\
		&= I_n + A^{-T} + D + N A^{-1} M.
	\end{align*}
	As we saw earlier, $N A = A^T N$. But this implies $A^{-T} N = N A^{-1}$ so that we have
	\begin{align*}
		I_n + A^{-T} + D + N A^{-1} M &= I_n + A^{-T} + D + A^{-T} N M\\
		&=I_n + A^{-T} + D + A^{-T} (M N)^T\\
		&=I_n + A^{-T} + D + A^{-T} (I_n + A)^T\\
		&=I_n + A^{-T} + D + A^{-T} + I_n\\
		&=D
	\end{align*}
	as required.
\end{proof}

\cref{supmat:thm:clifford decomposition_1} is a direct consequence of the two lemmas we proved above.

\begin{proof} (Proof of Theorem \ref{supmat:thm:clifford decomposition_1})
    Take an arbitrary Clifford operator $C$ and denote its symplectic representation by $\chi$. It follows directly from \cref{supmat:lem:make fist quadrant invertible} and \cref{supmat:lem:XZXZ} that there exist symplectic matrices $\zeta \in \mathcal{Z}_n(1)$, $\alpha, \gamma \in \mathcal{X}_n$ and $\beta, \delta \in \mathcal{Z}_n$ such that $\chi = \zeta \cdot \alpha \cdot \beta \cdot \gamma \cdot \delta$.
	The symplectic representation of Clifford operators used in this manuscript (see \cref{supmat:subsec:review-of-paulis-and-cliffords}) ,  $\Phi: \mathcal{C}_n / \mathcal{P}_n \rightarrow Sp_{2n}(2)$, is defined to right-act on row-vectors, and therefore the order of multiplication must be inverted, i.e., for $C_1, C_2 \in \mathcal{C}_n$ we have $\Phi(C_1 C_2) = \Phi(C_2)\Phi(C_1)$. 
	Since $\mathcal{X}_n$, $\mathcal{Z}_n$ and  $\mathcal{Z}_n(1)$ are precisely the images of the sets of $X$-diagonal, $Z$-diagonal, and depth-1 $Z$-diagonal gates, respectively, under the symplectic representation, the matrix decomposition $\chi = \zeta \alpha \cdot \beta \cdot \gamma \cdot \delta$ corresponds to the desired decomposition of the Clifford operator $C$. 
\end{proof}

In order to prove the second decomposition, described in \cref{supmat:thm:clifford decomposition_2}, we need to prove one additional lemma:
\begin{lem}\label{supmat:lem:decomposition assuming invertible first quadrant}
	Let $\chi = \begin{pmatrix} A & B\\C & D\end{pmatrix} \in Sp_{2n}(2)$ with invertible top-left quadrant $A\in GL_n(2)$. 
    Then there exist symplectic matrices $\alpha \in \mathcal{X}_n$, $\gamma \in \mathcal{Z}_n$ and 
    $\beta = \begin{pmatrix}F & 0 \\ 0 & F^{-T}\end{pmatrix}$ 
    with $F\in GL_n(2)$ such that $\chi = \alpha \cdot \beta \cdot \gamma$.
\end{lem}
\begin{proof}
	Define matrices $\alpha$, $\beta$ and $\gamma$ as in the lemma above. 
	We compute their product:
	\begin{eqnarray*}
		\alpha \cdot \beta \cdot \gamma & = & 
		\begin{pmatrix}
			I & 0\\ E & I
		\end{pmatrix}
		\cdot 
		\begin{pmatrix}
			F & 0 \\ 0 & F^{-T}
		\end{pmatrix}
		\cdot
		\begin{pmatrix}
			I & G\\ 0 & I
		\end{pmatrix} \\ &
		= &
		\begin{pmatrix}
			F & FG\\
			EF & EFG + F^{-T}
		\end{pmatrix}\,.
	\end{eqnarray*}	
	By choosing $E = C A^{-1}$, $F = A$, and $G = A^{-1}B$, the product above yields
	\[
	\alpha \cdot \beta \cdot \gamma = \begin{pmatrix}
		A & B\\
		C & CA^{-1}B + A^{-T}
	\end{pmatrix}
	\]\,
	To prove the lemma, we must show that $CA^{-1}B + A^{-T} = D $, and that our choices of $E$, $F$, and $G$ result in valid symplectic matrices $\alpha$, $\beta$, and $\gamma$.
	Recall that the matrix $\chi$ is symplectic, i.e., it satisfies the symplectic condition $\chi^T \Omega \chi = \Omega$. This yields the following conditions on the submatrices $A$, $B$, $C$, and $D$:
	\begin{align}
		A^T D + C^T B = I\,, \label{supmat:eq:condition_1}\\
		A^T C + C^T A = 0\,, \label{supmat:eq:condition_2}\\
		B^T D + D^T B = 0\,. \label{supmat:eq:condition_3}
	\end{align}
	Since $A$ is invertible, condition (\ref{supmat:eq:condition_1}) 
	can be restated as 
	\begin{equation}\label{supmat:eq:D}
		D = A^{-T} + A^{-T}C^T B\,.
	\end{equation}
	Similarly, condition (\ref{supmat:eq:condition_2}) can be reformulated as 
	\begin{equation}\label{supmat:eq:CA}
		CA^{-1} = A^{-T}C^T\,.
	\end{equation} 
	Combining these two equations, we find that $D = CA^{-1}B + A^{-T}$, and hence $\alpha \beta \gamma = \chi$.
	
	The symplectic condition for $\alpha$ requires $E = E^T$. For our choice $E = CA^{-1}$, this follows from \cref{supmat:eq:CA}.
	The symplectic condition for $\beta$ requires $F\in GL_n(2)$, which is guaranteed for our choice $F=A$.
	Finally, the symplectic condition for $\gamma$ requires $G = G^T$, which for our choice of $G$ becomes $ A^{-1}B + B^T A^{-T} = 0$. Since $A$ is invertible, this condition is equivalent to $AB^T + B A^{T} = 0$.
	We find the following equalities:
	\begin{align*}
		AB^T &+ B A^{T} \\
		&= AB^T + B A^{T} + AB^{T} \left(A^{-T}C^T + CA^{-1} \right) B A^{T}\\
		&= AB^T\left(I + A^{-T}C^TBA^T\right) +  \left(I + AB^TCA^{-1}\right)B A^{T}\\
		&= AB^T D A^T + AD^TBA^T\\
		&= A\left(B^T D + D^TB \right)A^T\\
		&= 0\,,
	\end{align*}
	where the first equality follows from Eq.~\ref{supmat:eq:CA}, the third equality follows from Eq.~\ref{supmat:eq:D}, and the last equality follows form Eq.~\ref{supmat:eq:condition_3}.
\end{proof}

Similar to the proof of \cref{supmat:thm:clifford decomposition_1}, we prove \cref{supmat:thm:clifford decomposition_2} by combining \cref{supmat:lem:make fist quadrant invertible} with \cref{supmat:lem:decomposition assuming invertible first quadrant}.
\begin{proof} (Proof of Theorem \ref{supmat:thm:clifford decomposition_2})
	Take an arbitrary Clifford operator $C$ and denote its symplectic representation by $\chi$. It follows directly from \cref{supmat:lem:make fist quadrant invertible} and \cref{supmat:lem:decomposition assuming invertible first quadrant} that there exist symplectic matrices $\zeta \in \mathcal{Z}_n(1)$, $\alpha, \in \mathcal{X}_n$, $\beta = \begin{pmatrix} F&0 \\ 0&F^{-T} \end{pmatrix}$ with $F\in GL_n(2)$, and $\gamma \in \mathcal{Z}_n$ such that $\chi = \zeta \alpha \cdot \beta \cdot \gamma$.
	
	Note that $\mathcal{X}_n$, $\mathcal{Z}_n$ and  $\mathcal{Z}_n(1)$ are precisely the images of $X$-diagonal, $Z$-diagonal, and depth-1 $Z$-diagonal gates under the symplectic representation, respectively. Also note that $\beta$ is the symplectic representation of a CNOT circuit.
	Finally, recall that we defined the symplectic representation of Clifford operators,  $\Phi: \mathcal{C}_n / \mathcal{P}_n \rightarrow Sp_{2n}(2)$, to right-act on row-vectors, and therefore the order of multiplication must be inverted, i.e., for $C_1, C_2 \in \mathcal{C}_n$ we have $\Phi(C_1 C_2) = \Phi(C_2)\Phi(C_1)$. 
	The decomposition $\chi = \zeta \alpha \cdot \beta \cdot \gamma$ thus corresponds to the desired decomposition of Clifford operator $C$. This concludes the proof.	
\end{proof}

\subsection{Cost of arbitrary Clifford operators}
The Clifford decompositions introduced above can be used in conjunction with the results for the depth of various types of logical gates, detailed in Sections \ref{supmat:sec: cnot gate characterisation}, \ref{supmat:sec: diagonal gate characterisation} and \ref{supmat:sec: Hadamard}, to determine an upper bound on the depth required for an arbitrary logical Clifford operator. Since the decomposition in \cref{supmat:thm:clifford decomposition_1}, relies entirely on $X$- and $Z$-diagonal operators, it is straightforward to find an upper bound for the total depth. Note that one can always choose $DZ(1)$ to contain exclusively in-block diagonal gates, which means it can always be implemented in a depth no greater than $r^2 +5r + 2$ for $r\geq 4$ ($r^2 + 8r + 2$ for $r=3$). For a Clifford operator acting on $b$ code blocks, the other components can each be implemented in depth $br^2 + (b + 4)r + (4b - 2)$ when $b$ is even, and $(b+1)r^2 + (b + 5)r + (4b + 2)$ when $b$ is odd, for $r\geq 4$ ($16b + 25$ and $16b + 41$, respectively, for $r=3$).
Most importantly, since no logical in-block CNOT gates are required, this method for synthesizing an arbitrary Clifford operator does not require any auxiliary code blocks.
We therefore find the following result, which provides a comprehensive costing of Clifford operator implementations in the SHYPS codes:
\begin{thm}\label{supmat:thm:final-theorem}
    Let $r\geq 4$. Any Clifford operator on $b$ blocks of $SHYPS(r)$ is implemented fault-tolerantly in depth at most
    \begin{align}
        (4b+1)r^2 + (4b+21)r + 16b - 6, \,\,\text{for } b \textbf{ even} \\
        (4b+5)r^2 + (4b+25)r + 16b + 10, \,\,\text{for } b \textbf{ odd}
    \end{align}
    
    Moreover this implementation requires no auxiliary code blocks.
    When $r=3$, the depth is at most $64b+135$ for $b$ even and  $64b+199$ for $b$ odd.
\end{thm}

For completeness, we also compute an upper bound to the depth of an arbitrary logical Clifford operator synthesized through the decomposition detailed in \cref{supmat:thm:clifford decomposition_2}. Since this decomposition may contain in-block CNOT gates, it requires up to $b$ auxiliary code blocks to perform a logical Clifford operator on $b$ code blocks. Despite this space overhead, there might be specific Clifford operators for which this decomposition is advantageous.

Note that since the depth of multi-block CNOT gates was determined up to a logical permutation, we should add the cost of a $b$-block permutation to it. This, however, does not contribute to the leading order of the depth of general CNOT circuits because multi-block permutations have an implementation of depth $O(r^2)$ while $b$ block CNOT circuits require $O(br^2)$. Also note that $DZ(1)$ can always be chosen to be an in-block depth-1 diagonal circuit. Furthermore, it can be chosen to contain no more than $r$ $S$ gates per $SHYPS(r)$ code block, which allows us to use \cref{supmat:cor:depth Xi} along with two in-block permutations. One of these in-block permutations can be included in the general permutation required for the CNOT circuit and can thus be ignored, resulting in a total depth of $8r+7$.


Using the constructions of logical generators presented in previous sections, we then find the following total depth for a Clifford synthesized with this decomposition:
\begin{align*}
	& DZ:  & br^2 + (b + 4)r + 4b - 2\\
	+&\, CX:  & (2b + 35)r^2 + (2b+2)r + 8b + 2\\ 
	+&\, DX: & br^2 + (b + 4)r + 4b - 2\\
	+&\, DZ(1):  & 8r + 7\\
	\hline
	&\textrm{Total Depth:} &  (4b+35)r^2 + (4b+18)r + 16b + 5 
\end{align*}
The costing for multi-block CNOT circuits assumes $b=2^a$ for some integer $a$.
The above costing of diagonal gates assumes that $r\geq 4$. As shown in \cref{supmat:thm:diagonal-decomp-invertible}, in-block diagonal gates may incur an additional depth of up to $3r$ when $r=3$, and incorporating this in the above calculation produces an overall depth upper bound of $64b+392$ in this case.

\subsection{Space-time volume}

Besides the depth (expressed as number of syndrome extraction rounds) required to execute various kinds of logical operators, another valuable metric to compare the performance of error correcting codes is the total space-time volume (expressed in units of ``physical qubits $\times$ syndrome extraction rounds'') required for said logical operations. 
In particular, when compiling logical circuits, there is often an opportunity to trade space for time, i.e., reduce the total depth of the circuit by introducing auxiliary qubits. Therefore the total space-time volume of a given logical operator usually provides a more balanced comparison point than only the depth. 
Below, we compute the space-time volume required to execute an arbitrary $m$-qubit logical Clifford operator in SHYPS codes. We then compare this to the space-time volume required to execute such an operator in rotated surface codes using both transversal operators and lattice surgery.

Performing an arbitrary logical Clifford operator on $m$ qubits encoded with the SHYPS($r$) code requires $\ceil{m/r^2}$ data code blocks, and hence $(2^{r}-1)^2 \ceil{m/r^2}$ physical qubits in total.
The worst-case depth of the circuit is listed in \cref{supmat:thm:final-theorem}. Note that no auxiliary code blocks are required to achieve this depth. 
The exact space-time volume is then readily found by multiplying space and time costs.
To leading order, it scales as $16 k \ceil*{\frac{m}{k}}^2 d^2 \approx 16 m^2d^2/k$, where $k=r^2$ is the number of logical qubits per code block.

In surface codes, one may execute any Clifford operator using either (fold- or permutation-) transversal gates or lattice surgery. 
We compute the space-time volume of the former based on the Clifford decomposition detailed in \ref{supmat:thm:clifford decomposition_1} and the upper bounds on the circuit depth for diagonal operators found in Ref.~\cite{Maslov2022}. 
The resulting space-time volume for a logical $m$-qubit Clifford operator is given by
\[
4\floor*{\frac{m}{2} +0.4993\log_2(m)^2 + 3.0191 \log_2(m) - 10.9139} + 5\,,
\]
where we used that encoding $m$ logical qubits in rotated surface codes takes $md^2$ physical qubits.
Note that to leading order, this scales as $2m^2d^2$. Hence, for sufficiently large distances (and hence a sufficiently large $k$ value for the SHYPS code), SHYPS outperforms surface codes for this metric. 

Finally, we consider rotated surface codes with lattice surgery.
We first note that any $m$-qubit Clifford operator can be performed using $m$ auxiliary qubits and $2m$ Pauli product measurements (see, for instance, Fig.~21 in Ref. \cite{Nickerson.2022}). 
For the space and time cost of performing said Pauli product measurements using lattice surgery, we use the ``fast block'' configuration detailed in Ref.~\cite{litinski2019_game}.
The resulting space requirement is $(4m + 4\ceil{\sqrt{m}})d^2$, and each Pauli product measurement requires $d$ rounds of syndrome extraction, leading to a total space-time cost of
\[
8m^2 d^3 + 8 m\ceil*{\sqrt{m}} d^2\,.
\]

We compare all results for distances 8, 16 and 32 and $10\leq m \leq 360$ in \cref{fig:space-time}. 
As expected, SHYPS outperforms the surface code with both transversal gates and lattice surgery on this metric. Even with $d=8$, a SHYPS code block has a sufficient number of logical qubits to cancel out the worse prefactor in the leading term of the space-time volume compared to surface codes with transversal gates.
Note that for small $m$, surface codes with transversal gates require a smaller space-time volume in some cases. 
This is due to the fact that the resources required to perform a Clifford operator with SHYPS codes increases step-wise with each multiple of $k$ (for a given distance), whereas surface codes can scale gradually since each code block only contains a single logical qubit. Once $m$ is sufficiently large (22,27 and 22 for $d=$8,16,32, respectively), this minor advantage of surface codes is undone by the much lower space-overhead of SHYPS codes. 

While it would be interesting to also include the space-time cost of other QLDPC codes with recent code surgery schemes \cite{Brown.2022, cross2024, williamson2024_lowoverhead, ide2025_faulttolerant, cowtan2025_parallel} in the comparison above, at the time of writing, and to the best of our knowledge, no precise calculations of this kind have been done for QLDPC codes with code surgery schemes elsewhere in the literature. 
A complete costing of all recent code surgery procedures in the literature would require a dedicated study and is beyond the scope of this work.

Finally, we note that the above space-time costs are those for worst-case Clifford operators. In practice, when compiling highly structured quantum algorithms, one will typically use specialized compiling techniques optimized for particular subroutines. 
Therefore, while the space-time volumes computed above provide a useful benchmark to compare different codes, a proper in-depth comparison would entail performing a full quantum resource estimation of known quantum algorithms. This is, however, beyond the scope of this work. 

\begin{figure}
    \centering
    \begin{subfigure}[b]{0.45\textwidth}
        \includegraphics[width=\linewidth]{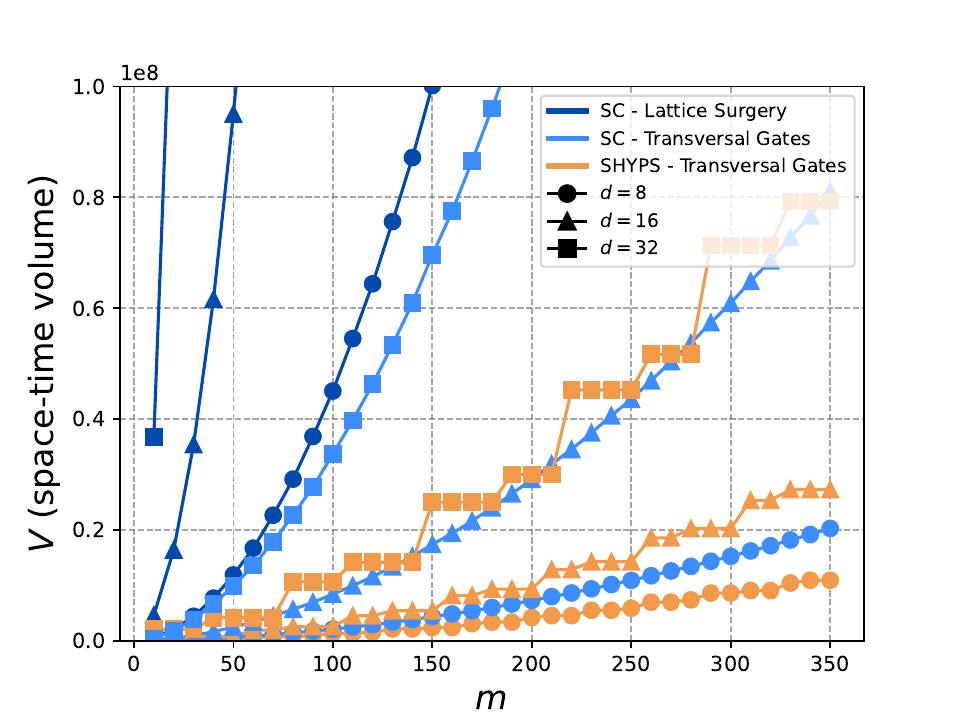}
    \end{subfigure}
    \begin{subfigure}[b]{0.45\textwidth}
        \includegraphics[width=\linewidth]{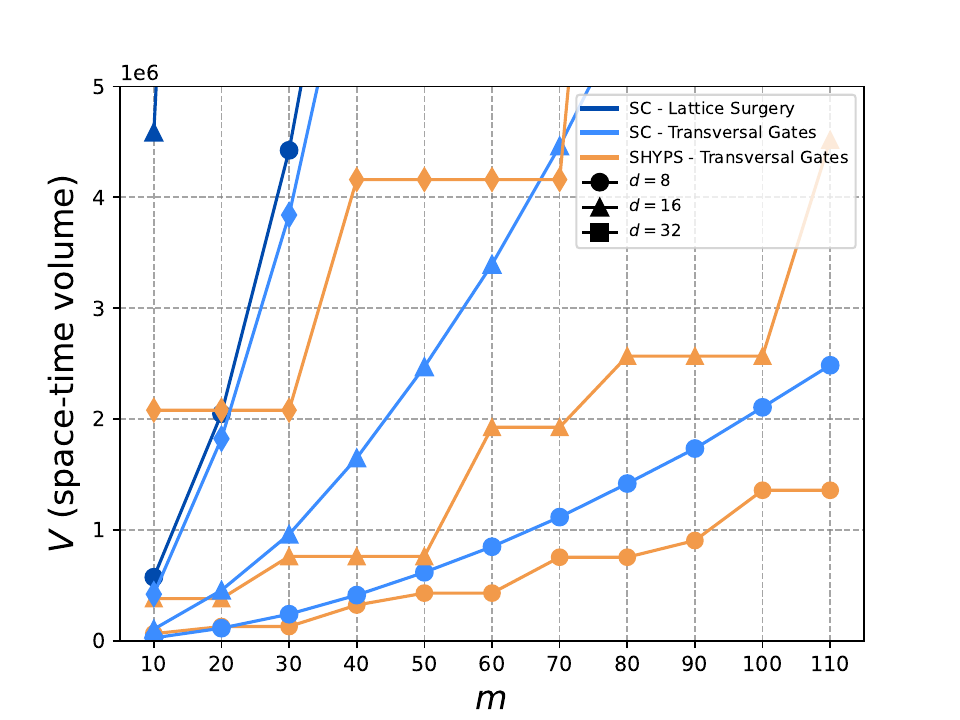}
    \end{subfigure}
    \caption{Space-time volume $V$ of a logical $m$-qubit Clifford operator performed on SHYPS codes, rotated surface codes using transversal operators, and rotated surface codes using lattice surgery for distances 8, 16 and 32.}
    \label{fig:space-time}
\end{figure}


\section{Fault-tolerant demonstration} \label{supmat:fault-tolerant-demonstration}

In this section, we discuss our simulations of quantum memories and Clifford circuits using SHYPS codes. We begin in \cref{supmat:sec:numsims} with a detailed description of the setup used for memory and logic simulations.
Next, in \cref{supmat:sec:syndrome extraction circuits} we discuss fault-tolerant syndrome extraction (SE), as well as SE scheduling options for the SHYPS codes. In \cref{supmat:decoding-details}, we give details of the decoder used in the simulations.

\subsection{Numerical simulations}\label{supmat:sec:numsims}

We use two types of numerical simulations in this paper:
\begin{itemize}
    \item \textbf{Memory simulation} to establish performance of a single code block of the $[49, 9, 4]$ and $[225, 16, 8]$ SHYPS codes as compared to surface codes of analogous scale.
    \item \textbf{Logic simulation} of a random Clifford operation decomposed into efficient logical generators applied across two blocks of the $[49, 9, 4]$ SHYPS code.
\end{itemize}

In what follows, we describe the details of each of these simulations. Although we describe concepts in the context of SHYPS codes, this discussion would also be applicable to simulations of other CSS codes.

\subsubsection{Memory simulations}

Quantum memory simulations are the standard approach to study the circuit-level performance of surface codes, and have recently been extended to analyze QLDPC codes under circuit-level noise \cite{ibm-qmem, closed-branch-decoder, bplsd, Gong2024, riverlane-ac}. Circuits that implement quantum memory experiments follow a specific set of steps:
\begin{enumerate}
    \item \textbf{Transversal Initialization:} Initialize the data qubits in the $Z$ ($X$) basis and perform a single SE round.
    \item \textbf{Syndrome Extraction:} Perform a predefined number of SE rounds. In most cases, $d$ SE rounds are used for a code with distance $d$.
    \item \textbf{Transversal measurement of data qubits:} Measure the data qubits in the $Z$ ($X$) basis.
\end{enumerate}

We use the open-source Clifford simulation package \texttt{Stim} \cite{stim} to build descriptions of these circuits, each annotated with \textit{detectors} and \textit{logical observables}. A detector is the parity of measurement outcome bits in a quantum error correction circuit that is deterministic in the absence of errors, while a logical observable is the linear combination of measurement bits whose outcome corresponds to the measurement of a logical Pauli operator \cite{dems-higgot}. To describe detectors we label the bits produced during SE with the stabilizer index $i$, SE round index $t$, and basis $B \in \{X, Z\}$. Throughout this section, unless explicitly stated, we assume use of the $Z$ basis for initialization and transversal measurement. A similar approach is valid for $X$ basis simulations.

The circuits we use for memory simulation begin with \textbf{transversal initialization}: initializing all data qubits to $\ket{0}$ and performing a single SE round. After this first SE round, all $Z$ stabilizers will be in the $+1$-eigenspace and $X$ stabilizers will be projected randomly to the $+/- 1$ eigenspace. For this reason, at the end of this first SE round we only define detectors based on the $Z$ stabilizer measurement results: $$ D_i^{t=0}(Z) = s_i^{t=0}(Z) \ . $$ 

Next, SE rounds are repeated a predefined number of times. Since the first round of SE forces the $X$ stabilizers to have either a $+1$ or $-1$ eigenvalue, their expected value after the second SE round should no longer be random in the absence of noise, enabling the definition of detectors based on $X$ stabilizer measurements. Note that this state preparation method does not produce a logical all-zero state $\ket{\bar{0}}$ matching the ``all stabilizers in the $+1$-eigenspace" definition, but rather an equivalent version where not all $X$ stabilizers have eigenvalue $+1$. This is purely symbolic though, since by tracking the Pauli frame (easily achieved with \texttt{stim}) we make the $-1$ eigenspace the codespace for those specific $X$ stabilizers that are flipped (equivalent to applying Pauli corrections to flip $-1$ stabilizers into the codespace), and prepare $\ket{\bar{0}}$ in a single-shot manner. We define detectors following subsequent rounds of SE by comparing stabilizer measurement results in between SE rounds: $$D_i^t(Z) = s_i^t(Z) \oplus s_i^{t-1}(Z), \ D_i^t(X) = s_i^t(X) \oplus s_i^{t-1}(X) ,$$ 
where $\oplus$ represents addition modulo 2.

The last step is transversal measurement in the $Z$ basis. In this case, we define detectors by comparing the stabilizer values from the final noisy SE round to the final stabilizer values computed from products of the data qubit measurements. We do not define detectors for $X$ stabilizers as they are unknown following measurement in the $Z$ basis. Logical observables are defined to be the logical $Z$ operators of the SHYPS code. We can write them as 
\begin{equation} \label{eq:observables}
L(Z) = \bigoplus_{m_i \in L_Z} m_i \ ,
\end{equation} where $m_i$ takes values $0$ or $1$ and represents data qubit measurements produced during transversal readout in the $Z$ basis, $L_Z$ represents the chosen basis for the logical $Z$ operators of the SHYPS code and $i \in \{0, \ldots, n-1\}$. 

Once descriptions of memory circuits are constructed, complete with detectors and observables, the decoding problem for memory simulations can be cast within the framework of Detector Error Models (DEM) \cite{dems-higgot, dems-eisert}. DEMs convey information about SE circuits in the form of a detector check matrix $\mathbf{H}_{\text{DCM}}$, a logical observable matrix $L$, and a vector of priors $\mathbf{p}$. The rows and columns of $\mathbf{H}_{\text{DCM}}$ represent detectors and independent error mechanisms in the circuit, respectively. The entry in position $(i, j)$ of $\mathbf{H}_{\text{DCM}}$ will be $1$ iff the $i$-th detector is flipped (recall that detectors are deterministically $0$ in the absence of noise) whenever the $j$-th error occurs and zero otherwise. Similarly, the rows and columns of $L$ represent the $k$ logical observables we are attempting to preserve with our protocol, and the independent error mechanisms in the circuit, respectively. The entry in position $(i, j)$ of $L$ will be $1$ iff the $i$-th logical observable is flipped by error mechanism $j$ and zero otherwise. The vector of priors $\mathbf{p}$ contains the prior error probability for each of the individual error mechanisms in the circuit. 

We use \texttt{Stim} to compute $\mathbf{H}_{\text{DCM}}$, $L$, and $\mathbf{p}$ for memory circuits under standard circuit-level depolarizing noise \cite{circuit-noise-standard}. This noise model assumes that each element in a quantum circuit is independently either ideal or faulty with probability $1-p$ and $p$, respectively, where $p$ is the model parameter called the physical error rate. In the context of our circuits for memory simulations, we have the following faulty operations:

\begin{itemize}
    \item \textbf{State preparation:} With probability $p$, the orthogonal state (e.g., $\ket{0}$ instead of $\ket{1}$) is prepared.
    \item \textbf{Measurement:} With probability $p$, the classical measurement result is flipped (from $0$ to $1$ or vice versa).
    \item \textbf{Single qubit gates:} With probability $p$, apply $X$, $Y$, or $Z$ (the specific Pauli operator is picked uniformly at random). An idle qubit in any time step experiences a noisy $I$ gate.
    \item \textbf{Two qubit gates:} With probability $p$, apply one of the $15$ nontrivial $2$-qubit Pauli operations on the control and the target qubits $\{IX, IY, IZ, XI, \ldots, ZZ\}$ (the specific Pauli operator is picked uniformly at random).
\end{itemize}

Aside from using \texttt{Stim} to compute the triplet $\mathbf{H}_{\text{DCM}}$, $L$, and $\mathbf{p}$, we also use it to simulate our circuits efficiently and produce detector and observable samples over different physical noise realizations. This allows us to formulate the decoding problem for quantum memories: provided with $\mathbf{H}_{\text{DCM}}$, $\mathbf{p}$, and the detector samples, the decoder provides an estimate of the real circuit error, which we denote by $\mathbf{c}$. We assess the accuracy of the error correction protocol by comparing the observable samples provided by \texttt{Stim} to the logical effect of $\mathbf{c}$, computed as $L \cdot \mathbf{c}$. Repeating this procedure over different physical noise realizations allows us to estimate the logical error rate of quantum memories. 

We refer those looking for extensive discussions on detector error models and the circuit-level decoding problem to \cite{dems-higgot, dems-eisert} and \cite{riverlane-ac}, respectively. 

\subsubsection*{Detector considerations}

In recent works, memory simulations have used either strictly $Z$-type or strictly $X$-type detectors \cite{ibm-qmem, closed-branch-decoder, bplsd, Gong2024, riverlane-ac}. More explicitly, $D(X)$ ($D(Z)$) have not been used to decode $Z$ ($X$) basis experiments~\footnote{Note that all stabilizers are still measured, as omitting the $X$ ($Z$) stabilizer measurements in a $Z$ ($X$) memory experiment eliminates the guarantee that the protocol would work to preserve an arbitrary quantum state.}. This is primarily due to a substantial increase in the size of the detector check matrix $H$ when both detector types are used, and the fact that most decoders cannot exploit the traces left by $Y$-errors on both the $X$ and $Z$ stabilizers. To keep comparisons fair, we simulate our memory circuits using only $Z$-type detectors. 

\subsubsection{Simulations of logical operation}

As with memory, we construct descriptions of logical circuits in \texttt{Stim}. The following steps outline a logical circuit:

\begin{enumerate}
    \item \textbf{Transversal Initialization:} Initialize the data qubits in the $Z$ ($X$) basis and perform a round of syndrome extraction.
    \item \textbf{Logical Operations Interleaved with Syndrome Extraction.} The Clifford unitary of interest is synthesized as a depth-$D$ sequence of Clifford generators. Each of these Clifford generators is applied to the circuit followed by a round of stabilizer extraction. For the logic simulation shown in the main text, the circuit simulated comprises of $2 \times 63 = 126$ stabilizer generators interleaved with syndrome extraction. 
    \item \textbf{Transversal measurement of data qubits:} Measure the data qubits in the $Z$ ($X$) basis.
\end{enumerate}

In order to annotate our \texttt{Stim} circuits with deterministic observables using equation \eqref{eq:observables}, we also apply the inverse of the synthesized operator. This is why the total depth of the circuit we used for our Clifford simulation has twice the depth as the sampled Clifford circuit itself. Noise is added to our descriptions of logical circuits using the standard circuit-level noise model. 

It is worth mentioning that it is also possible to perform analysis based on logical measurement results that are random. This was recently done in \cite{correlated-decoding-2Zhou2024}, where the authors override the \texttt{Stim} requirement for deterministic observables by defining \textit{gauge detectors} and then interpreting the random measurement results according to their proposed methodology. Running simulations with non-deterministic observables is important for a full characterization. In fact, it might even be more practical in some ways, as it would eliminate the need for inverting the logical action of the Clifford operators we simulate. As is done in most quantum error correction analyses and simulations, we have focused on the case of deterministic observables for simplicity. We leave the extension to simulations with non-deterministic observables for future work.

\subsubsection*{Detector discovery for logical circuits}

As with memories, the first step in a logical circuit is \textbf{transversal initialization}, so we define detectors in this first stage to be the $Z$ stabilizer measurement outcomes: $$ D_i^{t=0}(Z) = s_i^{t=0}(Z) \ .$$

The next stage in the circuit involves performing multiple rounds of a logical operation followed by SE. Defining detectors in this stage is more nuanced than for quantum memories \footnote{Note that memories are just simulations of trivial logic.} because logical operations preserve the codespace but act non-trivially on the stabilizer generators we infer during SE \footnote{We do not measure the stabilizer generators of the SHYPS code directly. We measure their gauge generators and use those outcomes to compute the stabilizer measurement results.}. For instance, the logical fold-transversal Hadamard gate $H^{\otimes{n}}\tau$ gate 
(see \cref{supmat:sec: Hadamard}) maps $X$ ($Z$) stabilizers in the $(i-1)$-th round to $Z$ ($X$) stabilizers in the $i$-th round. To define valid detectors, we must track the map that each particular logical gate applies on the stabilizer generators. We do so based on a general linear algebra approach that computes the logical action of the operation on the stabilizers and returns the relationship between the stabilizers before and after the logical operation. Alternatively, it is also possible to define detectors in the context of non-trivial logic by using specific update rules for each type of logical gate that may appear in the circuit. In \cite{correlated-decoding-2Zhou2024}, the authors explain how to define detectors following the application of transversal logical $H$, $\text{CNOT}$, and $S$ gates.

\subsection{Syndrome extraction circuits}
\label{supmat:sec:syndrome extraction circuits}

In our simulations we do not measure the stabilizers of a SHYPS code directly to perform SE. Instead, we measure the gauge generators of the code and then aggregate those measurement outcomes accordingly to infer the stabilizer measurement results. We use the SE circuits proposed in \cite{subsystemSurf} to implement gauge generator readout for SHYPS codes in our simulations. This type of SE circuit belongs to the family of \textit{bare-auxiliary} gadgets \cite{readout_methods}, where a single auxiliary qubit is used to readout a stabilizer or gauge generator. In what follows we explain the details of our SE strategy.

\subsubsection{Circuit fault analysis} 
\label{supmat:circuit-fault-analysis}

The bare-auxiliary method for extracting syndromes provides little protection against error-spread from auxiliary qubits to data qubits. However, this is not an issue for SHYPS codes. Consider the top circuit in Figure \ref{supmat:fig:gauge-readout-supmat}, which measures the $Z$-gauge generator $g_{Z, l} = Z_iZ_jZ_k$, where the subscript $Z$ denotes the fact that the gauge generator is a composed purely of single-qubit $Z$ operators and the subscript $l$ is used to represent an arbitrary index. In this figure, numbered boxes are included not to represent operations, but to differentiate moments in the circuit.

\begin{figure}[!h]
    \includegraphics[width=\linewidth]{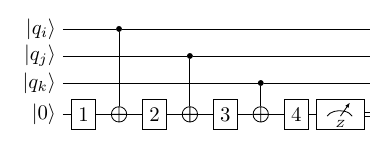}
    \hspace{0.06\linewidth}
    \includegraphics[width=\linewidth]{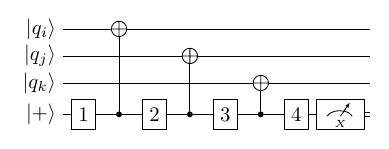}
    \caption{Circuits to measure the eigenvalue of $Z$-type (top) and $X$-type (bottom) gauge generators in an SHYPS code. The numbered boxes represent circuit moments.}
    \label{supmat:fig:gauge-readout-supmat}
\end{figure}

If a single $Z$ error occurs on the auxiliary qubit, depending on the circuit moment at which it happens (positions $1-4$), the error will propagate to one of the following data qubit errors: $\{Z_iZ_jZ_k, Z_jZ_k, Z_k\}$. Error $Z_iZ_jZ_k$ is the gauge generator $g_{Z,l}$ itself, so it can be ignored. By multiplying the second error $Z_jZ_k$ by $g_{Z,l}$, we end up with $Z_i$, which is a weight-$1$ error. The third error is already a weight-$1$ error. Thus, a single $Z$ error on the auxiliary qubit can lead to at most one $Z$ error on the data qubits, modulo gauge generators. An $X$ error on the auxiliary qubit cannot spread to data qubits, so its only effect is flipping the measurement outcome, which will produce a measurement error. 

 The circuit is fault-tolerant against $Z$-errors, as the single $Z$-faults on the auxiliary qubits are equivalent to at worst single $Z$-errors on the data qubits. It is fault-tolerant against $X$ errors because they cannot spread to data qubits. Since a single-qubit error on the auxiliary qubits leads to a single-qubit error in the code block, the circuit is fault-tolerant and hence it is protected from high-weight hook errors \cite{hook}. A similar analysis can be applied to show that the second circuit shown in Figure \ref{supmat:fig:gauge-readout-supmat}, which measures the $X$-gauge generator $g_{X,l} = X_iX_jX_k$, is also fault-tolerant.

 Note that $X$ errors on the auxiliary qubits will alter the measurement outcome, indicating the presence of $X$-errors on the data when there are none. This problem is not unique to this method of syndrome extraction, and is what leads to the frequently seen notion of repeated stabilizer measurements with a number of times matching the code distance.

\subsubsection{Scheduling SE for SHYPS codes}
\label{supmat:scheduling-se-circuit}

There are two approaches to scheduling gauge generator measurement circuits for SHYPS codes in the minimum-possible depth. The first relies on exploiting the particular structure of SHYPS codes while the second makes use of the coloration circuit approach from \cite{beverland-sched}. Both produce gauge generator measurement circuits of the same depth. For the simulations in this paper, we use the coloration circuit to schedule gauge generator readout.

\subsubsection*{Structure-based scheduling} 

Recall the structure of the $X$-gauge generators for the SHYPS code from \cref{supmat:sec:SHSsuppMat}:
\[
G_X = H \otimes I_{n_r},
\]
where $H$ is the over-complete $n_r \times n_r$ parity check matrix of the classical $(n_r,r,d_r)$-simplex code. Moreover, $H$ is the cyclic parity check matrix corresponding to a choice of weight-3 parity-check polynomial $h(x)=1+x^{d_1}+x^{d_1+d_2}$.
So there exist $n_r^2$ gauge generators $r_i\otimes e_j$, corresponding to a choice of rows $r_i$ and $e_j$ from matrices $H$ and $I_n$, respectively. Furthermore, the cyclic structure of $H$ clearly implies that $g = X^{r_i \otimes e_j}$ is supported on qubits $q_{i,j}$, $q_{i+d_1,j}$, and $q_{i+d_1+d_2, j}$, where we denote qubits in a codeblock by $q$ and qubit labels are combined modulo $n_r$.

Hence, we may schedule measurement of all $n_r^2$ $X$-gauge generators in $5$ time steps: 
\begin{enumerate}
    \item Preparation of the auxiliary states $a_{i,j} = \ket{+}_{i,j}$
    \item Apply $\prod CNOT(q_{i,j},a_{i,j})$
    \item Apply $\prod CNOT(q_{i+d_1,j},a_{i,j})$
    \item Apply $\prod CNOT(q_{i+d_1+d_2,j},a_{i,j})$
    \item Measure auxiliary qubits $a_{i,j}$.
\end{enumerate}
This requires a total $3n_r^2$ physical CNOT gates. Note that other configurations of the CNOT circuit are possible, for example Steps $2$, $3$, and $4$ could be taken in any order. A similar schedule applies for extracting the $Z$-gauge syndromes. A naive composition of the $X$ and $Z$-gauge measurement circuits would yield a depth-$10$ circuit. However, a depth-$8$ circuit is possible if: $X$
gauge generator measurement qubits are initialized during the last moment of $Z$ gauge generator extraction, and $Z$ gauge generator measurement qubits are measured in the first moment of $X$ gauge generator extraction.

\subsubsection*{Coloration circuit approach} 
\label{supmat:edge-coloring-scheduling}

We can alternatively schedule gauge generator readout circuits for any SHYPS code using the edge-coloring algorithm depicted in Algorithm \ref{supmat:alg:coloration}, which is a slightly modified version of the algorithm proposed in \cite{beverland-sched}. Specifically, this algorithm works with a modified pair of Tanner graphs $T'_X$ and $T'_Z$, where $T'_X$ and $T'_Z$ are Tanner graphs whose stabilizer check nodes have been substituted by gauge generator nodes while all other aspects of the algorithm remain exactly the same. 

\SetKwComment{Comment}{\# }{}
\begin{algorithm} 
\caption{Edge Coloring Circuit}\label{supmat:alg:coloration}
\KwData{The Tanner graphs $T'_X$ and $T'_Z$, as well as their minimum edge colorings $\mathcal{C}_X$ and $\mathcal{C}_Z$}
\KwResult{Gauge measurement circuit for $X$ and $Z$ gauge generators of an SHYPS code}
\Comment{Z Gauge Generators}
 Initialize all auxiliary qubits that will measure $Z$ gauge generators in the $\ket{0}$ state\;
 \For{$c \in \mathcal{C}_Z$}{
  In the same circuit moment, apply $CNOT_{i\rightarrow j}$ gates from the $i$-th data qubit (control) to the $j$-th auxiliary qubit (target) supported on an edge $\{i, j\}$ with color $c$\;
 }
 \Comment{X Gauge Generators}
 Initialize all auxiliary qubits that will measure $X$ gauge generators in the $\ket{+}$ state\;
 \For{$c \in \mathcal{C}_X$}{
  In the same circuit moment, apply $CNOT_{i\leftarrow j}$ gates from the $j$-th auxiliary qubit (control) to the $i$-th data qubit (target) supported on an edge $\{i, j\}$ with color $c$\;
  }
\end{algorithm}

\subsubsection*{Depth optimality of SE circuits}
\label{supmat:depth-optimality-of-se-circuits}

A critical aspect to consider in the context of SE scheduling is evaluating how \textit{good} the obtained readout circuits are, since there is no guarantee that methods such as the coloration circuit achieve the minimum possible depth \cite{beverland-sched}. Note how, in principle, it is possible to schedule CNOTs for $X$- and $Z$- generators in the same circuit moment such that these interleaved circuits achieve lower depth than the coloration approach, which schedules CNOTs for $X$- and $Z$- stabilizers separately. For instance, the SE circuits produced by the edge-coloring approach for surface codes have depth $10$, while the interleaved scheduling approach yields circuits of depth $6$.

In the context of SHYPS codes, both the structure-based and edge-coloring approaches produce the minimum depth circuits that implement gauge generator readout. This is because SHYPS codes do not allow for interleaving CNOTs involved in $X$ and $Z$ gauge generators, as there are no idling qubits at any given circuit moment of the gauge generator readout circuit. In other words, since every data qubit is involved in each circuit moment, there is no way of pulling CNOT gates through the circuit so that they occur earlier and the depth is reduced.

\subsubsection{Stabilizer aggregation for SHYPS codes} 
\label{supmat:stabilizer-aggregation-for-shs-codes}

Once the gauge generators of an SHYPS code have been measured, we obtain the stabilizer measurement outcomes via stabilizer aggregation. We can determine how to perform the aggregation by exploiting the structure of SHYPS codes. Recall  from \cref{supmat:sec:SHSsuppMat} that the $X$-gauge generators and $X$-stabilizers are given by
\[
G_X = H \otimes I_{n_r},\ S_X = H \otimes G,
\]
respectively. We can write $S_X=(I_{n_r} \otimes G) G_X$. This means that the nonzero entries in the rows of the classical generator matrix $G$ indicate which gauge generator measurement results must be combined to compute each stabilizer measurement outcome: \textit{The indices of the nonzero entries in row $i$ of $G$ are the gauge generators whose combination yields stabilizer generator $i$ of the SHYPS code}. The same argument can be followed to compute the $Z$-stabilizer aggregation.

\subsection{Monte Carlo simulation data}
\label{supmat:monte-carlo-simulation-data}

We produce our simulation results by Monte Carlo sampling memory and logic circuits with \texttt{Stim} and decoding over different physical error rates. We schedule our Monte Carlo simulations with enough runs to sample at least $100$ logical errors for every physical error rate. In this section, we explain how we compute the uncertainties of our simulation results and how we normalize the logical error rate for Figures \ref{fig:memory-simulations-plot} and \ref{fig:logic-simulations-plot} in the main text.

\begin{table*}[h!]
\centering
{\renewcommand{\arraystretch}{1.2}
    \setlength\tabcolsep{5mm}
    \begin{tabular}{l|c|c|c}
         \textbf{Parameters} & SHYPS $[49, 9, 4]$ & SHYPS $[49, 9, 4]$ & SHYPS $[49, 9, 4]$\\
         \hline
         Sliding Window & (2,1) & (3,1) & (4,1) \\ 
         BP Iterations & 100 & 250 & 500\\ 
         MS Static Scaling Factor & 0.1 & 0.1 & 0.1 \\
         LSD Order $\mu$ \ & 1 & 1 & 1 \\ 
         \\
         \\
         \textbf{Parameters} & \ SHYPS $[225, 16, 8]$ & \ SHYPS $[225, 16, 8]$ & \ SHYPS $[225, 16, 8]$\\
         \hline
         Sliding Window & (2,1) & (3,1) & (4,1) \\ 
         BP Iterations & 2000 & 4000 & 6000\\ 
         MS Static Scaling Factor & 0.85 & 0.85 & 0.85 \\
         LSD Order $\mu$ \ & 4 & 4 & 4 \\ 
    \end{tabular}}
    \caption{Decoding parameters for memory simulations}
    \label{table:decoder-parameters-memory}
\end{table*}

\begin{table*}[ht]
    \centering
    {\renewcommand{\arraystretch}{1.2}
    \setlength\tabcolsep{5mm}
    \begin{tabular}{l|c}
         \textbf{Parameters} & \ SHYPS $[49, 9, 4]$ \\
         \hline
         Sliding Window & (3,1) \\ 
         BP Iterations & 1500\\ 
         MS Static Scaling Factor \ & $0.07$ \\ 
         LSD Order $\mu$ \ & 8 \\ 
    \end{tabular}}
    \caption{Decoding parameters for Clifford Simulation}
    \label{table:decoder-parameters-clifford}
\end{table*}

\subsubsection{Uncertainties of simulation results}

We use the 95\% confidence interval, shown as shaded regions in Figs.~\ref{fig:memory-simulations-plot} and~\ref{fig:logic-simulations-plot} of the main text as well as in Figs.~\ref{supmat:fig:swd-shs-49-9-4} and~\ref{supmat:fig:swd-shs-225-16-8} in~\cref{supmat:additional-simulation-results} of the supplementary material, to portray the uncertainties associated to our simulation results. We compute the 95\% confidence interval as follows.

A Monte Carlo simulation that executes $n_s$ runs of independent Bernoulli trials with an observed failure rate of $\tilde{p}$ has a variance 
\begin{equation}
    \sigma^2 = \frac{\tilde{p}(1 - \tilde{p})}{n_s}. 
\end{equation}
Its standard deviation can thus be directly calculated as
\begin{equation}
    \sigma = \sqrt{\frac{\tilde{p}(1 - \tilde{p})}{n_s}}.
\end{equation}
Thus, we can establish the uncertainty bounds of the true failure probability $p$ using the observed failure rate $\tilde{p}$ within a specified confidence interval:
\begin{equation}
    p = \tilde{p} \pm z \cdot \sigma,
\end{equation}
where, for example, $z = 1.96$ corresponds to a $95\%$ confidence level.

Furthermore, when the logical error rates are normalized and scaled based on the number of logical observables $v$, the number of SE rounds $s$, and the number of copies of the code $m$ -- details of which will be provided in the next subsection -- we apply the \textit{propagation of uncertainty} to update the uncertainty bounds accordingly.

\subsubsection{Normalizing and scaling logical error rates}
\label{supmat:normalizing-and-scaling-logical-error-rates}

At the end of our Monte Carlo simulations, we obtain the number of logical error events for a given number of Monte Carlo simulation runs. Let $p_{v,s}$ be the \textit{observed logical error rate} at the end of a simulation that has $v$ logical observables and $s$ syndrome extraction rounds. Note that $p_{v,s}$ is also widely referred to as the \textit{shot error rate}. In \cref{fig:memory-simulations-plot} and \ref{fig:logic-simulations-plot} of the main text, we report a different quantity, the \textit{logical error rate per syndrome extraction round}, which we denote as $p_{v, 1}$. Next, we describe how $p_{v, 1}$ is calculated from $p_{v,s}$, $v$, and $s$.

We begin by determining the \textit{observed logical error rate per logical observable}, represented by $p_{1,s}$. Since any error or flip affecting the logical observables is treated as a single logical error event, we can directly define the relationship between $p_{v,s}$, $p_{1,s}$, and $v$ as follows:
\begin{equation}
    p_{v,s} = 1 - (1 -  p_{1,s}) ^ v.
    \label{eq:pvd}
\end{equation}
By rearranging~\eqref{eq:pvd}, we derive  
\begin{equation}
    p_{1,s} = 1 - (1 - p_{v,s}) ^ {1/v}.
    \label{eq:p1d_1}
\end{equation}

Next, we calculate the \textit{logical error rate per logical observable per syndrome extraction round}, denoted by $p_{1,1}$. Consider that each of the logical observables undergo $s$ successive Bernoulli trials, where an outcome of $0$ represents success and $1$ signifies failure, with each trial having a failure probability of $p_{1,1}$. Consequently, after performing $s$ consecutive trials, the final result is effectively the XOR of the $s$ trials, leading to an observed failure probability of $p_{1,s}$. This is equivalent to stating that each logical observable passes through a binary symmetric channel (BSC) repeated $s$ times in sequence. Therefore, we can express the relationship between $p_{1,s}$, $p_{1,1}$, and $s$ as follows:
\begin{equation}
    p_{1,s} = \frac{1 - \left(1 - 2 p_{1,1} \right) ^ s }{2} 
    \label{eq:p1d_2}
\end{equation}
and by rearranging~\eqref{eq:p1d_2}, we obtain
\begin{equation}
    p_{1,1} = \frac{1 - \left(1 - 2 p_{1,s} \right) ^ {1/s} }{2} 
    \label{eq:p11}
\end{equation}

Finally, we calculate the \textit{logical error rate per syndrome extraction round}, denoted by $p_{v,1}$. Similar to~\eqref{eq:pvd}, we can explicitly express the relationship between $p_{v,1}$, $p_{1,1}$, and $v$ as follows:
\begin{equation}
    p_{v,1} = 1 - \left(1 - p_{1,1} \right) ^ {v}.
    \label{eq:pv1_1}
\end{equation}

By substituting~\eqref{eq:p1d_1} into~\eqref{eq:p11}, and substituting the result into~\eqref{eq:pv1_1}, we obtain a final expression for the logical error rate per syndrome extraction round, $p_{v,1}$, based on the observed logical error rate $p_{v,s}$, the number of logical observables $v$, and the number of syndrome extraction rounds $s$ as follows:
\begin{equation}
    p_{v,1} = 1 - \left( \frac{1 + \left[ 2 (1 - p_{v,s} ) ^ {1/v} - 1 \right] ^ {1/s} }{2} \right) ^ {v}.
    \label{eq:eq:pv1_2}
\end{equation}
The above calculation is equivalent to the method implemented in \texttt{sinter}, a package that integrates directly with \texttt{Stim}, to calculate the so-called \textit{piece error rate} from the \textit{shot error rate}.

To produce the scaled surface code data in \cref{fig:memory-simulations-plot} and the two-block memory data in \cref{fig:logic-simulations-plot} of the main text we calculate the logical error rate per syndrome extraction round $p_{v,1}$ for multiple patches of the same code. We denote this quantity by $p_{v,1}^m$, where $m$ is the number of code patches. We compute the \textit{scaled} logical error rate per syndrome extraction round based on the following equation:
\begin{equation}
    p^m_{v, 1} = 1 - (1 - p_{v, 1}) ^ {m}.
\end{equation}

\subsection{Decoding details}
\label{supmat:decoding-details}

The results for SHYPS codes reported in \cref{fig:memory-simulations-plot} and \ref{fig:logic-simulations-plot} of the main text were generated using a proprietary implementation of a sliding window decoder with min-sum (MS) Belief Propagation + Localized Statistics Decoding (BPLSD)  \cite{bplsd} as the constituent decoder. We used the open-source BPLSD implementation available at \cite{Roffe2024LDPCv2}. Surface code data was generated using the \texttt{pymatching} package \cite{dems-higgot}. Here we give details about our decoding approach. We provide the specific parameter configurations for MS BPLSD (scaling factor, iterations, and LSD order $\mu$) used for each simulation in tables \ref{table:decoder-parameters-memory} and \ref{table:decoder-parameters-clifford}. Decoding parameter values where optimized via Monte Carlo simulation.

\subsubsection{BPLSD}
\label{supmat:BPOSD}

BPLSD is a two-stage decoder that combines BP with LSD. Belief propagation (BP) achieves excellent decoding performance for classical LDPC codes \cite{sum-product, belief-prop}. However, standalone BP decoding may not perform as well for QLDPC codes due to a variety of reasons, including degenerate error patterns, short-cycles in the decoding graph, and split beliefs \cite{roffe2020, degeneracy, trapping-sets}. The quantum-specific issues standalone BP suffers can be alleviated by providing the output of a BP decoder to a secondary decoder for reprocessing. Until recently, ordered statistics decoding (OSD) was the best performing reprocessing decoder for BP-based decoding of QLPDC codes. OSD relies on sorting the reliability of BP outputs and systematizing the decoding graph to improve performance. The resulting combination, referred to as BPOSD, is a high-performing general decoder that has become the state of the art for decoding QLDPC codes \cite{roffe2020, PanteleevOSD, surfacecodealgos}. Unfortunately, the improvements in decoding performance provided by BPOSD come with a price. OS decoding involves an unavoidable inversion step over the entire detector error matrix (which may contain $10,000+$ columns in practice), which drastically increases the complexity when compared to standalone BP decoding. Significant effort is being being made to curtail the complexity of BPOSD without sacrificing error correction performance \cite{closed-branch-decoder, riverlane-ac, bplsd, iolius2024almostlinear}. 

BPLSD \cite{bplsd} is a recently proposed decoding algorithm that achieves decoding performance on par with BPOSD, but with significantly reduced computational complexity. The algorithm is based on the observation that, in the sub-threshold regime, errors are typically sparse and confined to disconnected, localized regions of the decoding graph. BPLSD employs a novel cluster growth strategy to efficiently identify these regions in parallel. Matrix inversions are then applied only to the localized subgraphs, in contrast to the global detector graph used in BPOSD, resulting in notably improved decoding efficiency. The performance of BPLSD matches that of BPOSD in terms of decoding accuracy, while offering significantly reduced complexity. Furthermore, we observe for SHYPS codes that the higher-order post-processing variant, BPLSD-$\mu$, consistently outperforms comparable post-processing techniques applied to BPOSD across a range of physical error rates.

\subsubsection*{BPLSD parameters}
\label{supmat:decoding-parameters}

BPLSD offers a variety of tunable parameters that can be adjusted to improve decoding performance:

\begin{itemize}
    \item \textbf{Maximum number of BP iterations, $\text{BP}_{\text{it}}^{\text{max}}$}: BP is an iterative message-passing algorithm that will not generally terminate on its own \cite{sum-product}. As such, the maximum number of BP iterations represents the number of times that messages are allowed to be exchanged over the decoding graph, and, generally, more iterations result in better performance. If the BP decoder fails to find a converging solution after reaching the maximum number of BP iterations, its outputs will be provided to the LSD for a subsequent attempt at decoding.
    \item \textbf{LSD-$\mu$ or LSD-order}: Inspired by higher-order OSD, LSD also allows for improved correction capabilities via higher order reprocessing. In OSD, the so-called \textit{order} or \textit{search depth} $w \in \mathbb{Z}_{0}^+$, determines the number of candidates to evaluate as possible solutions to the decoding problem. A larger $w$ incurs higher complexity while providing better error correction performance. LSD offers a similar parameter, denoted by $\mu \leq 0$, that can be tuned for higher order reprocessing. As is the case for BPLSD with no higher order reprocessing, LSD-$\mu$ achieves decoding improvements on par with BPOSD-$w$ at reduced complexity by localizing operations.
    \item \textbf{MS scaling factor}: In the log‑likelihood ratio domain, the exact check‑node update rule for a BP decoder, commonly referred to as the \emph{sum–product algorithm} (SPA), requires the calculation of the $\tanh$ function and its inverse. The \emph{min–sum algorithm} (MSA) variant dispenses with these transcendental functions and keeps only the sign and the minimum magnitude. While MSA is hardware‑friendly, it systematically over-estimates the reliability of every bit, culminating in a performance loss with respect to SPA. A simple but effective remedy to such performance decrements is the introduction of the so-called \emph{static scaling factor} $\alpha$, which is chosen offline to deflate the resultant log-likelihood ratio coming from the MSA update rule. This results in a check-node update rule referred to as the \emph{normalized min–sum algorithm} (NMSA) \cite{chen2002near}. The simple heuristic $\alpha \approx \frac{1}{\sqrt{d_c - 1}}$, where $d_c$ is the average check nodes degree, provides a generic rule of thumb based on the check‑node degree to select values for $\alpha$. It serves as a useful starting point that can be further refined with a parameter space search via Monte Carlo simulation to find the static scaling factor value that produces the best performance. Classically, a density evolution (DE) and/or extrinsic information transfer (EXIT) analysis can be used to find the optimal static scaling factor for NMSA. However, the DCMs in circuit-level noise decoding simulations exhibit significantly more short cycles than those encountered in the classical realm, which curtails the effectiveness of these classical methods. On top of that, since LSD only runs on the rare convergence failures of BP, tuning $\alpha$ for standalone BP does not guarantee optimum performance for the cascaded BPLSD. In practice, the LSD stage may actually benefit from an $\alpha$ value that would be sub-optimal for a standalone BP decoder. Therefore, the recommended procedure to find the scaling factor value that yields the best decoding performance is a direct Monte Carlo parameter sweep using the combined decoder. 
\end{itemize}

\subsubsection{Sliding window decoder}
\label{supmat:sliding-window-decoder}

The dimensions of a detector check matrix (DCM), which we denote by $\mathbf{H}_{\text{DCM}}$, grow with the number of SE rounds. This implies that decoding successively deeper circuits with many SE rounds becomes increasingly difficult. For reference, the DCM associated to our logical Clifford simulation has $10,668$ rows and $387,590$ columns. However, empirical evidence shows that DCMs generally possess the characteristics of classical spatially-coupled LDPC codes \cite{Huang2023, Berent2024, Gong2024}. In the context of classical coding theory, these codes can be decoded efficiently using a sliding window decoder (SWD) \cite{iyengar2011windowed}. Instead of decoding across the entire DCM, which may be prohibitively large or add unnecessary decoding latency, SWD performs sequential decoding on smaller \textit{window} decoding regions or subsets of the DCM. Decoding within the decoding regions is conducted with a standard decoding algorithm, such as standalone BP, BPOSD, or BPLSD. Once the window is decoded, the decoder shifts forward by a specified \textit{commit} step to the next decoding region of the DCM. To ensure continuity, consecutive windows typically overlap. The overlapping decoding region allows the decoder to account for correlations between adjacent subsets of the DCM. By focusing on smaller, localized regions, SWD reduces the computational load compared to decoding across the entire DCM at once. Memory usage is also significantly reduced since only the subsets of the DCM within the window decoding region need to be stored and processed. Despite its reduced complexity, SWD often achieves near-optimal performance for spatially-coupled LDPC codes as long as the information obtained from localized decoding propagates well to adjacent decoding regions.

The full DCM for multiple SE rounds has the following structure:
\begin{equation}
    \mathbf{H}_{\text{DCM}} = \left[ 
        \begin{array}{llllll}
             H_{1,1}    &           &           &           &               & \\
             H_{2,1}    & H_{2,2}   &           &           &               & \\
                        & H_{3,2}   & H_{3,3}   &           &               & \\
                        &           & H_{4,3}   & H_{4,4}   &               & \\
                        &           &           & \cdots    & \cdots        & \\
                        &           &           &           & H_{t,t-1}   & H_{t,t} \\
        \end{array}
    \right],
\end{equation}
where $t$ is the total number of sets of the detectors used in the simulation. In our simulation setup, $t = d + 2$, where $d$ is the number of SE rounds, and the additive factor of $2$ comes from the detectors defined after transversal initialization and transversal measurement. This formula holds for both memory and Clifford circuit simulations. 

Given the measured detector values $\mathbf{s}$, the decoding problem becomes finding the most likely error pattern $\widehat{\boldsymbol{e}}$ that satisfies
\begin{equation}
    \mathbf{s}^T = \mathbf{H}_{\text{DCM}} \cdot \widehat{\mathbf{e}}^T,
\end{equation}
where
\begin{align}
    \widehat{\mathbf{e}} &= \left[ \widehat{e}_{1} \ \ \widehat{e}_{2} \ \ \widehat{e}_{3} \ \ \cdots \ \ \widehat{e}_{t} \right], \\
    \mathbf{s} &= \left[ s_{1} \ \ s_{2} \ \ s_{3} \ \ \cdots \ \ s_{t} \right].
\end{align}
Based on the above expansion, we can write the decoding problem as the following set of equations:
\begin{align}
    s^T_{1} &= H_{1,1} \cdot \widehat{e}_{1}^T \nonumber \\
    s^T_{2} &= H_{2,1} \cdot \widehat{e}_{1}^T + H_{2,2} \cdot \widehat{e}_{2}^T \nonumber \\
    & \vdots \\
    s^T_{t-1} &= H_{t-1,t-2} \cdot \widehat{e}_{t-2}^T + H_{t-1,t-1} \cdot \widehat{e}_{t-1}^T \nonumber \\
    s^T_{t} &= H_{t,t-1} \cdot \widehat{e}_{t-1}^T + H_{t,t} \cdot \widehat{e}_{t}^T \nonumber
\end{align}
Since both formulations of the problem are equivalent, we can re-express the original decoding problem into a set of smaller decoding problems. 

An $\text{SWD}(w,c)$ has two defining parameters: the window size $w$ and the commit size $c$. The window size $w$ determines the number of sets of detectors that are used as inputs for each decoding round, while the commit size $c$ determines the number of detectors that we fix for committing the error correction. To envisage how an $\text{SWD}(w,c)$ works, take $\text{SWD}(3,1)$ as an example and let $i$ denote the index of the decoding round. 

In the first decoding round $(i = 1)$ of an SWD with $w=3$, we will take the detector values $\left[ s_{1} \ \ s_{2} \ \ s_{3} \right]$ and apply our chosen decoding algorithm to produce estimated errors $\left[ \widehat{e}_{1} \ \ \widehat{e}_{2} \ \ \widehat{e}_{3} \right]$ using the following segment of the DCM:
\begin{equation}
    \mathbf{H}^{(i = 1)}_{\text{DCM}} = \left[ 
        \begin{array}{lll}
             H_{1,1}    &           &           \\
             H_{2,1}    & H_{2,2}   &           \\
                        & H_{3,2}   & H_{3,3}   \\
        \end{array}
    \right].
\end{equation}
For $c = 1$, we commit the correction of the detector values $s_{1}$, which means that we commit correction on the estimated error $\widehat{e}_{1}$. After we commit to the estimated error, we move on to the next decoding round $(i = 2)$, where we now take the detector values $\left[ s_{2} \ \ s_{3} \ \ s_{4} \right]$ at the start of the decoding round. However, $s_{2}$ involves the estimated error $\widehat{e}_{1}$ that we committed in the previous decoding round $(i = 1)$, and we do not want to double-correct this error. To circumvent this, we update $s_{2}$ in this decoding round $(i = 2)$ as follows:
\begin{equation}
    {s^{T}_{2}}' = s^T_{2} + H_{2,1} \cdot \widehat{e}_{1}^T.
\end{equation}

Now we perform decoding using the detector values $\left[ s'_{2} \ \ s_{3} \ \ s_{4} \right]$ and apply our chosen decoding algorithm to produce estimated errors $\left[ \widehat{e}_{2} \ \ \widehat{e}_{3} \ \ \widehat{e}_{4} \right]$ on the following segment of the DCM:
\begin{equation}
    \mathbf{H}^{(i = 2)}_{\text{DCM}} = \left[ 
        \begin{array}{lll}
             H_{2,2}    &           &           \\
             H_{3,2}    & H_{3,3}   &           \\
                        & H_{4,3}   & H_{4,4}   \\
        \end{array}
    \right].
\end{equation}
We use the output of the second decoding round to commit $\widehat{e}_{2}$: the correction of the detector values $s_{2}$. Remember that we have committed to the correction of $\widehat{e}_{1}$ in the first decoding round, so we do not need to commit on it again. This overlapping decoding procedure is performed sequentially until we reach the window decoding region that involves the detector values $\left[ s_{t-2} \ \ s_{t-1} \ \ s_{t} \right]$. Recall that we need to update $s_{t-2}$ to $s'_{t-2}$ accordingly before starting this final round. In the last decoding round $(i = t - 2)$, we commit to the estimated errors in the entirety of $\left[ e_{t-2} \ \ e_{t-1} \ \ e_{t} \right]$ after applying our chosen decoding algorithm on the following segment of the DCM:
\begin{equation}
    \mathbf{H}^{(i = t - 2)}_{\text{DCM}} = \left[ 
        \begin{array}{lll}
             H_{t-2,t-2}    &               &           \\
             H_{t-1,t-2}    & H_{t-1,t-1}   &           \\
                            & H_{t, t-1}    & H_{t,t}   \\
        \end{array}
    \right].
\end{equation}
Finally, we obtain the estimated error $\widehat{\mathbf{e}} = \left[ \widehat{e}_{1} \ \ \widehat{e}_{2} \ \ \widehat{e}_{3} \ \ \cdots \ \ \widehat{e}_{t} \right]$ after performing $(t - 2)$ sequential decoding rounds using an $\text{SWD}(3,1)$.

\subsection{Single-shot aspects of SHYPS}
\label{supmat:additional-simulation-results}

\begin{figure}[th!]
    \includegraphics[width=\linewidth]{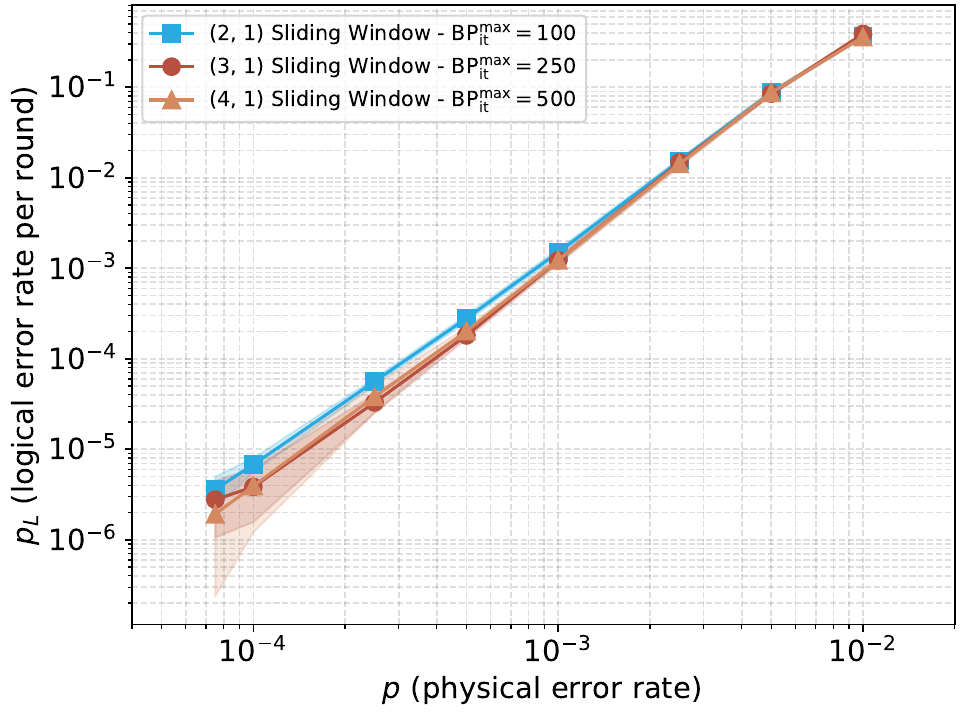}
    \caption{Logical error rate performance of the $[49,9,4]$ SHYPS code using a sliding window decoder when choosing different window sizes.}
    \label{supmat:fig:swd-shs-49-9-4}
\end{figure}

\begin{figure}[th!]
    \includegraphics[width=\linewidth]{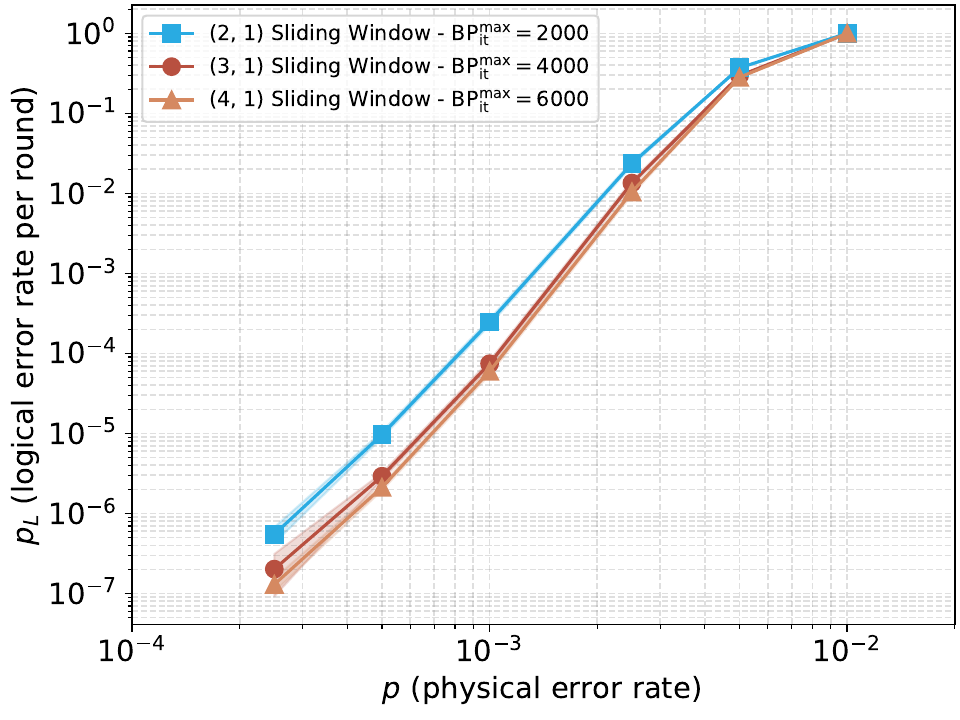}
    \caption{Logical error rate performance of the $[225,16,8]$ SHYPS code using a sliding window decoder when choosing different window sizes.}
    \label{supmat:fig:swd-shs-225-16-8}
\end{figure}

In this section, we discuss the features of SHYPS codes that make them decodable in constant depth. We stress that our goal is not to demonstrate full single-shot fault tolerance in the formal sense, but rather to identify and analyze practical single-shot features in SHYPS codes.

To this end, we investigate two representative examples: the $[49,9,4]$ and $[225,16,8]$ codes, shown in Figures \ref{supmat:fig:swd-shs-49-9-4} and \ref{supmat:fig:swd-shs-225-16-8}, respectively. For both, we observe signatures of single-shot behavior by confirming the existence of pseudo-thresholds under circuit-level noise simulation for all window sizes. Additionally, we find that decreasing the window size does not noticeably impact the logical error rate performance of SHYPS codes. For the $[49,9,4]$ SHYPS code (see Figure \ref{supmat:fig:swd-shs-49-9-4}), performance is practically the same for all window sizes, even while also decreasing $\text{BP}_{\text{it}}^{\text{max}}$. For the $[225,16,8]$ SHYPS code (see Figure \ref{supmat:fig:swd-shs-225-16-8}), similar behaviour is observed, with near identical logical error rates for $w = 3$ and $w = 4$. The logical error rate performance slightly degrades for $w = 2$ ($\text{BP}_{\text{it}}^{\text{max}} = 2000$), but this gap can likely be closed with further optimization of the BP parameters. A constant logical error rate suppression with decreasing window size is indicative of single-shot behaviour.

Following the approach proposed in \cite{lin2025abelianmulticyclecodessingleshot}, we also examine the confinement profile of the $[49,9,4]$ and $[225,16,8]$ SHYPS codes. The confinement profile is defined as the sequence of minimum syndrome weights corresponding to errors of weight $1$ up to $(d - 1)$. For both codes, we find a flat confinement profile of $[3, 3, 3, \ldots, 3]$. Furthermore, SHYPS codes have constant single-shot distance $d_{ss} = 3$ across the whole family. While this does not qualify as \emph{good} confinement in a formal sense \cite{Quintavalle.2021}, it still indicates that for sufficiently low physical error rates, the decoder can still leverage some confinement-like properties. 

In summary, while SHYPS codes do not satisfy all the formal criteria for single-shot fault tolerance, they demonstrate key practical signatures, such as logical error rate stability with decreasing window size and constant single-shot distance. This enables the use of a single syndrome extraction round between each logical operation and makes SHYPS codes decodable in constant depth.

\subsection{Simulating larger SHYPS memories}

Here we provide memory data for the $r=5$ SHYPS code with parameters $[[961, 25, 16]]$ code, the largest member of the family we have simulated. The results, alongside memory performance for smaller SHYPS codes, are shown in Figure \ref{fig:shyps-mems}. As is generally observed for subsystem code constructions, increasing the code distance tends to reduce the pseudo-threshold due to the additional gauge degrees of freedom and the correspondingly larger decoding search space~\cite{Poulin2005,subsystemSurf}. Consistent with this behavior, the pseudo-threshold for the $[[961,25,16]]$ SHYPS code is $p_{\mathrm{th}} = 0.0018$, lower than those of the smaller members of the family ($p_{\mathrm{th}} = 0.0032$ for $[[49,9,4]]$ and $p_{\mathrm{th}} = 0.0035$ for $[[225,16,8]]$). Nonetheless, the boost in error suppression provided by increased distance means that the $[[961, 25, 16]]$ code quickly catches and outpaces the logical error rate of the $[[225, 16, 8]]$ code. More specifically, the $[[961, 25, 16]]$ code achieves logical error rates of $10^{-6}$ at physical error rates approximately $6\times10^{-4}$. Such performance is commensurate with current trapped-ion hardware, which has demonstrated single-, two-qubit-, and SPAM-error rates in this regime \cite{helios}, while comparable fidelities are rapidly approaching feasibility in leading superconducting platforms \cite{willow}.


Considering that the $[[961,25,16]]$ SHYPS code provides a $5.85\times$ reduction in physical-qubit overhead relative to the $[[5625,25,15]]$ rotated surface code (modeled as $25$ copies of the $[[225,1,15]]$ code), and that continued progress in decoding and post-selection techniques may further relax the required physical error rate, the $r{=}4$ and $r{=}5$ SHYPS codes represent credible candidates for early-scale fault-tolerant computation.

\begin{figure}[t!]
\centering
\setlength{\abovecaptionskip}{0pt}
\setlength{\belowcaptionskip}{-1em}
    \includegraphics[width=\linewidth]{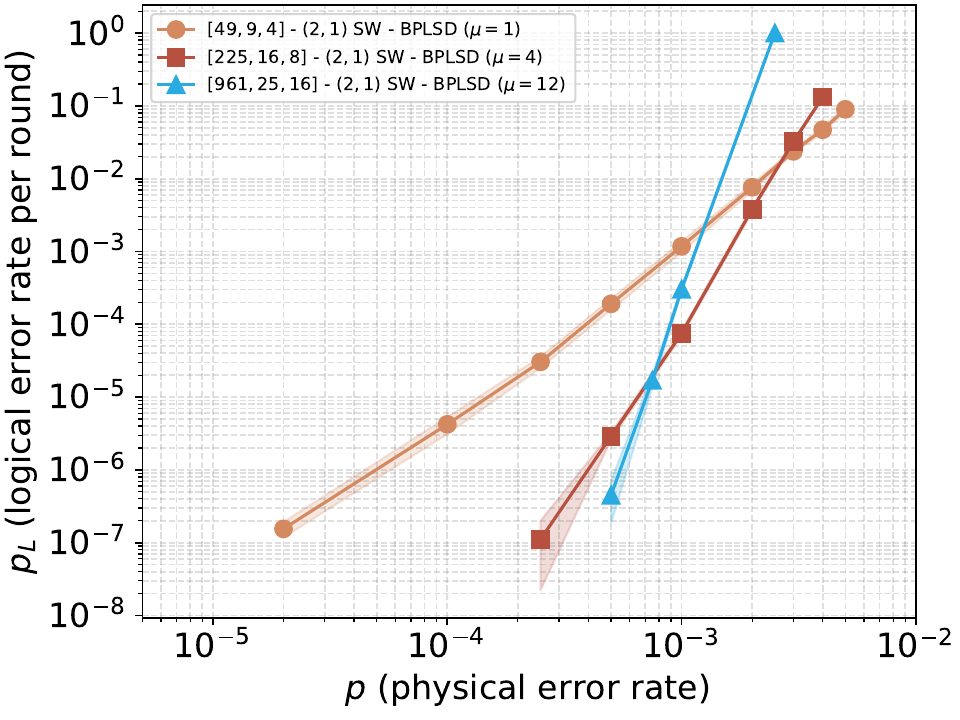}
    \caption{\centering Simulation results for quantum memories under circuit-level noise for SHYPS codes. For these simulations, only $Z$-type detectors are used.}
    \label{fig:shyps-mems}
\end{figure}

\section{{Physical connectivity}} \label{supmat:connectivity}

\begin{figure}[ht]
    \includegraphics[width=0.7\linewidth]{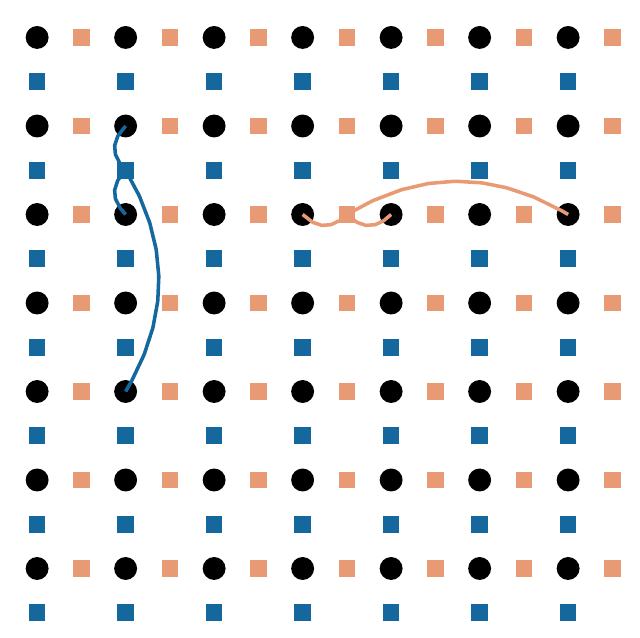}
    \caption{Gauge generator connectivity for $r=3$ SHYPS code. Black circles denote data qubits. Orange (blue) squares indicate auxiliary qubits used to measure 
$X$ ($Z$) gauge generators. Each gauge generator has weight 3 and is implemented using three CNOT gates, depicted as edges connecting the corresponding data qubits to a single auxiliary qubit. Two representative gauge generators are illustrated: one $X$-type and one $Z$-type. The full set of gauge generators can be obtained by translating this pattern under periodic boundary conditions.}
    \label{supmat:fig:stabs}
\end{figure}
The physical connectivity requirements of SHYPS codes are best understood by looking at its constituent classical codes. Consider a physical layout for the simplex code in which its Tanner graph is arranged on a line, with alternating data and parity-check nodes. Each parity check connects to its adjacent data nodes as well as to one additional data node via a ``long-range" check with periodic boundary conditions, consistent with the cyclic structure of the code.

The subsystem hypergraph product construction described in Eq.~\eqref{supmat:eqn:subsystem-hgp} naturally lifts this 1D structure to the 2D grid layout illustrated in Fig.~\ref{supmat:fig:stabs}. Here, data qubits are shown as black circles and $Z$ ($X$) auxiliary qubits are represented by orange (blue) squares. Auxiliary qubits serve as $X$- and $Z$-type parity checks along rows and columns, respectively.

The implementation of logical gates in SHYPS codes includes physical two-qubit gates, specifically CNOT and CZ gates. In the following, we examine the corresponding two-qubit gate connectivity requirements.

\subsection{Diagonal operators}
\begin{figure*}[t!]
    \includegraphics[width=0.8\linewidth]{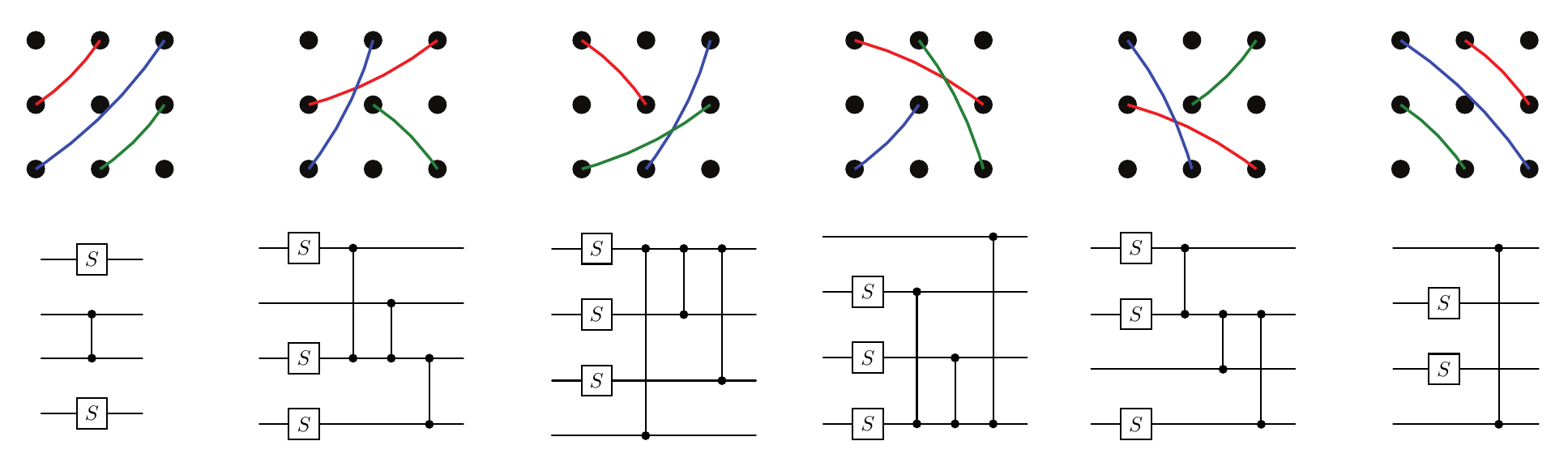}
    \caption{Diagonal generators for $r=2$. Black circles represent data qubits, edges indicate physical CZ gates (colors for visualization purposes only), and isolated vertices each support a physical S gate. Each of the six generators corresponds to an automorphism  $\sigma_i$. All generators are related to each other by a row and column involution. For example, the second generator is related to the first by the column permutation $(23)$ and the row permutation $\text{id}$. The logical action of each generator is shown as an equivalent quantum circuit. Note that $r = 2$ was chosen for visualization purposes and is not part of the SHYPS code family, which requires $r \ge 3$. That is, the six logical generators shown here do not span the full Diagonal group.}
    \label{supmat:fig:diagonal-gens}
\end{figure*}

Physical diagonal generators are given in Eq.~\eqref{supmat:eq:physicalDiag}. 
Note that we can rewrite this
\begin{align}
    U(\sig):= \mat{I & (\sig \otimes \sig^T)\cdot \tau_{n_r} \\ 0  & I}
    = 
    \pi_\sigma
    \mat{I & \tau_{n_r} \\ 0  & I}
    \pi_\sigma^T \,.
\end{align}
where $\pi_\sig$ is the physical row-column permutation matrix
\begin{align}
    \pi_\sig = \mat{ \alpha \otimes \beta & \\ & \alpha \otimes \beta}
\end{align}
and where we have decomposed $\sigma = \alpha \cdot \beta$ as a product of two involutions. Thus, all physical generators are related to $U(\text{id})$  by pairwise row and column swaps, where $\text{id}\in S_n$ denotes the identity permutation.

The physical diagonal operator $U(\text{id})$ corresponds to a physical circuit in which CZ gates are applied between qubits at positions $(i, j)$ and $(j, i)$ (with additional $S$ gates at positions $(i,i)$). This construction is shown in the leftmost panel of Fig.~\ref{supmat:fig:diagonal-gens} for $r=2$.
The remaining panels in Fig.~\ref{supmat:fig:diagonal-gens} display all additional diagonal generators, each obtained as a row and column involution of $U(\text{id})$, along with circuits representing their logical action. Note that the small value of $r$ is chosen for visualization clarity and does not represent a member of the SHYPS family, which requires $r \geq 3$.

\subsection{CNOT operators}
\begin{figure}[th!]
    \includegraphics[width=\linewidth]{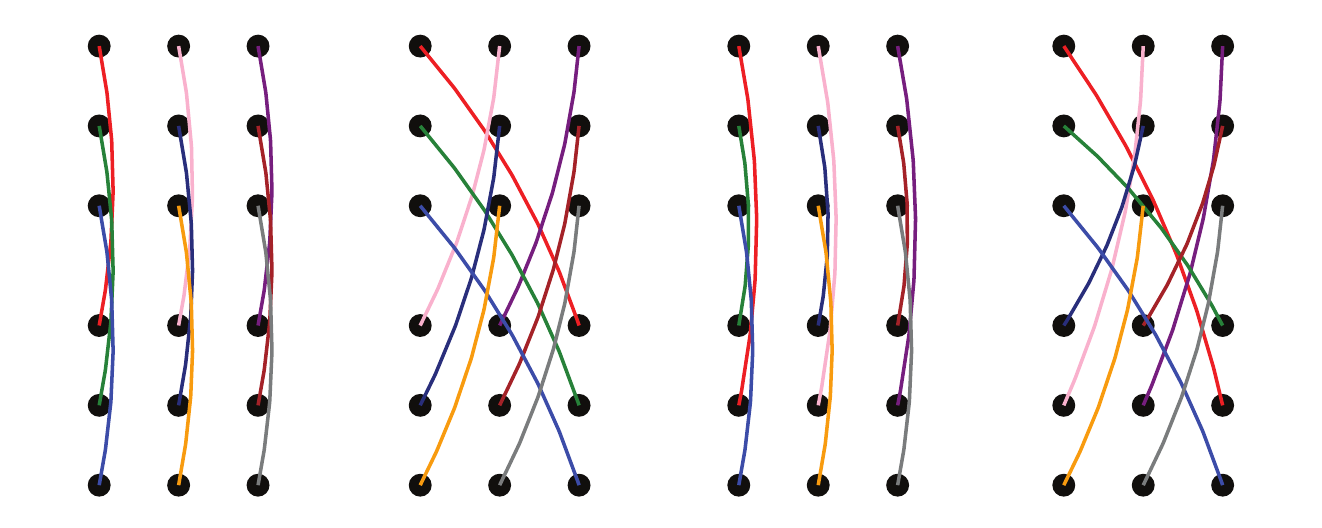}
    \caption{Cross-block CNOT generators. Each black circle represents a data qubit, arranged in two code blocks. Edges denote physical CNOT gates, with control qubits on the upper block and targets on the lower block (colors are for visualization purposes only). The leftmost diagram shows the standard transversal CNOT. The remaining diagrams illustrate additional valid CNOT generators obtained by applying a row permutation $\sigma_1 \in \{\text{id},\ (12)\}$ and a column permutation $\sigma_2 \in \{\text{id},\ (132)\}$ to the targets.}
    \label{supmat:fig:cnot-gens}
\end{figure}
Logical CNOT operators require transversal CNOT circuits between two code blocks, depicted in Fig~.\ref{supmat:fig:generalised-transversal-cnot}.
For completeness, we illustrate this on a grid layout in Fig.~\ref{supmat:fig:cnot-gens}. All generators are related to each other
by a permutation $\pi = \sigma_1 \otimes \sigma_2$, where $\sigma_1$ ($\sigma_2$) are row (column) permutations from the automorphism group.

\subsection{Hardware implementation}
Given that the gauge generators of SHYPS codes include a single non-local or ``long-range" check, syndrome extraction for SHYPS codes does not require stringent hardware connectivity. This means that SHYPS quantum memories are accessible to hardware platforms with slightly better than nearest neighbor connectivity. However, implementing logical operations necessitates all-to-all connectivity within code blocks and between separate blocks, implying that using SHYPS codes for computation will only be suitable for hardware platforms that naturally support non-local connectivity. Promising candidates include photonic architectures that enable non-local entanglement distribution via optical switch networks~\cite{inc2024distributedquantumcomputingsilicon}, as well as ion-trap and neutral atom systems~\cite{Xu2023}.

\end{document}